\documentclass[10pt]{article}
\usepackage{amsfonts,amsmath,amssymb,amsthm}
\usepackage[pdftex]{graphicx}
\usepackage[hmargin=1in,vmargin=1in]{geometry}
\usepackage{natbib} 
\usepackage{setspace}
\usepackage{enumerate}
\usepackage[toc,title,titletoc,header]{appendix}
\usepackage[colorlinks,citecolor=blue,hypertexnames=false]{hyperref}
\usepackage{etoolbox}
\usepackage{booktabs}
\usepackage{multirow}
\allowdisplaybreaks
\def\qed{\rule{2mm}{2mm}}

\newtheorem{theorem}{Theorem}[section]
\newtheorem{lemma}{Lemma}[section]

\newtheorem{proposition}{Proposition}[section]
\theoremstyle{definition}

\newtheorem{remark}{Remark}[section]
\newtheorem{assumption}{Assumption}[section]

\AtEndEnvironment{remark}{~\qed}
\AtEndEnvironment{example}{~\qed}

\DeclareMathOperator{\var}{Var}
\DeclareMathOperator{\cov}{Cov}
\DeclareMathOperator{\diag}{diag}
\DeclareMathOperator{\tr}{tr}
\DeclareMathOperator{\blp}{BLP}

\DeclareMathOperator*{\argmin}{argmin}

\begin{document}

\title{\vspace{-.5in}\Large Graph-Laplacian Variance Estimators for Finely Stratified Experiments \thanks{Some material in this paper first appeared in a retired draft, \emph{A New Design-Based Estimator for Finely Stratified Experiments}. The fourth author acknowledges support from the National Science Foundation through grant SES-2419008. We thank OpenAI's Codex for helpful discussions about exposition, notation, and manuscript editing.}}

\author{Yuehao Bai \\
Department of Economics \\
University of Southern California \\
\url{yuehao.bai@usc.edu}
\and
Xun Huang \\
Department of Economics \\
University of Chicago \\
\url{xhuang520@uchicago.edu}
\and
Joseph P.\ Romano\\
Departments of Economics \& Statistics \\
Stanford University\\
\url{romano@stanford.edu}
\and
Azeem M.\ Shaikh \\
Department of Economics \\
University of Chicago \\
\url{amshaikh@uchicago.edu}
\and
Max Tabord-Meehan \\
Department of Economics \\
University of Toronto \\
\url{m.tabordmeehan@utoronto.ca}
}

\begin{spacing}{1.2}
\maketitle
\end{spacing}

\vspace{-0.2in}

\begin{spacing}{1}
\begin{abstract}
This paper considers design-based inference on the average treatment effect in finely stratified experiments, where uncertainty arises only from the randomized treatment assignment. We focus on settings in which units are first stratified into groups of fixed size according to baseline covariates and then, within each group, exactly one unit is assigned to treatment (or, symmetrically, exactly one unit is assigned to control). 
%In this setting, we introduce a class of graph-Laplacian estimators for the variance of the difference-in-means estimator. These estimators place strata on a weighted graph whose vertices are the strata and whose edges determine which between-stratum treatment effect differences are aggregated to construct the estimator. 
In this setting, we introduce a class of graph-Laplacian variance estimators in which strata form the vertices of a weighted graph and edge weights determine how between-stratum comparisons are aggregated.
The canonical estimator of \cite{imai2008variance} corresponds to a complete graph with edge weights normalized so that each stratum has weighted degree one, while a paired-stratum estimator arises from a perfect matching graph. For the subclass of degree-calibrated graphs, in which each vertex has weighted degree one, we derive an exact bias identity showing that the corresponding estimators are upward-biased, with bias governed by squared differences in the true stratum-level treatment effects across adjacent strata. As a result, any such estimator may be used for valid inference. The identity further suggests that paired-stratum estimators constructed from a covariate-based perfect matching can induce small biases when treatment effects vary smoothly with the covariates. 
% However, we demonstrate that the resulting estimators may suffer from large worst-case bias in the absence of such smoothness and that, among graph-Laplacian estimators constructed from degree-calibrated graphs, the complete-graph estimator is minimax optimal for normalized bias under a weak bound on treatment-effect heterogeneity.
Without such smoothness, however, we show that paired-stratum estimators can exhibit large worst-case bias, and that, within the class of degree-calibrated estimators, the complete-graph estimator is minimax optimal for normalized bias under a weak bound on treatment-effect heterogeneity. Motivated by this contrast, we propose a regularized graph estimator that controls worst-case normalized bias while preserving much of the locality of the paired-stratum estimator. Simulations illustrate the resulting tradeoff between locality and worst-case protection.

\end{abstract}
\end{spacing}



\thispagestyle{empty} 
\newpage
\setcounter{page}{1}

\section{Introduction} \label{sec:intro}
This paper considers the problem of design-based inference on the average treatment effect in finely stratified experiments. Here, by ``design-based'' we mean that there is no sampling uncertainty and the only randomness comes from the treatment assignment mechanism. We consider stratified experiments in which units are first partitioned into strata of fixed size $k$ according to baseline covariates and then, within each stratum, exactly $\ell < k$ units are assigned uniformly at random to treatment and the remainder to control. Our main focus is the finely stratified case, in the terminology of \cite{fogarty2018mitigating}, in which $\min\{\ell, k - \ell\} = 1$. A prominent special case of this framework is a matched pairs design, in which $k = 2$ and $\ell = 1$.

In this setting, we first establish, by way of motivation, a result showing that standard confidence intervals for the average treatment effect based on the usual difference-in-means estimator require a variance estimator that is at least asymptotically upward-biased. We then review why constructing upward-biased variance estimators is particularly challenging when $\ell = 1$ (or $k - \ell = 1$), because in these cases, the within-stratum sample variance of treated outcomes (respectively, control outcomes) is identically zero. With this in mind, we introduce a new class of variance estimators, which we call \emph{graph-Laplacian} estimators, that remain well-defined in this challenging case. They view the strata as vertices of a weighted graph and average weighted squared differences in the corresponding stratum-level treatment-effect estimates across pairs of strata connected by an edge. Equivalently, they are graph-Laplacian quadratic forms in the stratum-level treatment-effect estimates.

This class of estimators interpolates between two extreme constructions. At one end is the estimator proposed by \cite{imai2008variance}, which is based on the sample variance formed by treating the stratum-level difference-in-means estimates as independent observations. In the graph-Laplacian framework, this estimator arises from the complete graph that assigns equal weight to each edge, such that each vertex has weighted degree one. At the other end is a paired-stratum estimator obtained from a perfect matching. The matching pairs the strata so that each stratum is compared to exactly one other stratum, in analogy with collapsed-strata variance estimators in survey sampling \citep[see][p.~109]{sarndal1992model}. More generally, we call a graph degree-calibrated if each vertex has weighted degree one; both the normalized complete graph and any perfect-matching graph satisfy this condition. We derive a closed-form expression for the bias of graph-Laplacian estimators constructed from degree-calibrated graphs, which shows that they are upward-biased and hence can be used for valid inference. The same expression shows that the magnitude of this bias is governed by treatment-effect differences across weighted edges, suggesting that large edge weights should be placed between strata with similar treatment effects.
 In practice, however, the true stratum-level treatment effects are unobserved, and so an empirical implementation would choose the graph using observed variables thought to be informative about treatment-effect heterogeneity, such as baseline covariates. Intuitively, this approach should work well when stratum-level average treatment effects vary smoothly with the covariates used to form the matching. This motivates constructing the graph using a minimum-cost perfect matching based on these baseline covariates.

The preceding discussion suggests a tension that motivates the rest of the paper. In a design-based analysis, the potential outcomes and covariates are treated as fixed and the assignment mechanism is the only source of randomness; the randomization framework itself imposes no relationship between baseline covariates and the stratum-level average treatment effects. Without further restrictions on the finite population, we show that no two distinct degree-calibrated graphs can be uniformly ordered in terms of their biases. We then show that, among graph-Laplacian estimators constructed from degree-calibrated graphs, the complete-graph estimator is uniquely minimax optimal for normalized bias over populations satisfying a bound on the variance of the stratum-level average treatment effects. In contrast, a minimum-cost perfect matching exposes the estimator to larger worst-case normalized bias when its matched pairs do not align with the heterogeneity in treatment effects.

We then study the complementary setting in which baseline covariates are in fact informative about treatment-effect heterogeneity. To do so, we introduce an asymptotic framework motivated by a thought experiment in which the finite population is modeled as having been drawn once, in an i.i.d.\ fashion, from a well-behaved probability distribution under which potential outcomes vary smoothly with baseline covariates. The limiting bias of a variance estimator constructed from a degree-calibrated graph is then governed by the extent to which its edge weights exploit this covariate information. The complete graph, which does not exploit covariate information, has a limiting bias that is larger than that of a minimum-cost perfect matching based on baseline covariates, unless treatment effects are homogeneous. As a result, except in this special case, confidence intervals based on a minimum-cost perfect matching are asymptotically strictly shorter than those based on a complete graph.

Combining insights from these two perspectives, we argue that minimum-cost perfect matching on baseline covariates is often the natural choice, and in ``smooth'' regimes can lead to the least bias and the shortest confidence intervals among the estimators we study. At the same time, the minimax analysis demonstrates that in a purely design-based setting a perfect matching is a strong commitment to this idealized regime. We therefore propose a \emph{regularized} minimum-cost graph construction that keeps the baseline-covariate cost criterion used for the perfect matching while imposing a spectral constraint on the graph-Laplacian. This constraint controls worst-case normalized bias while allowing the selected graph to exploit enough covariate information to attain the same limiting bias as minimum-cost perfect matching in the smooth regime. In this sense, even a small amount of regularization can protect against adversarial treatment-effect configurations while preserving much of the locality that makes the perfect-matching estimator attractive in practice. The simulations illustrate this tradeoff: confidence intervals based on a minimum-cost perfect matching are shortest in the smooth regime, intervals based on the complete graph are shortest in the adversarial regime, and in the transition case, regularized graphs produce the shortest intervals.

The literature on stratified block randomization dates back to at least \cite{fisher1935design}. For the case in which $\ell = 1$ or $k - \ell = 1$, the variance estimator in \cite{imai2008variance} has been extended to different settings by \cite{imai2009essential}, \cite{fogarty2018mitigating}, \cite{pashley2021insights}, \cite{de2024level} and \cite{zhu2024design-based}. \cite{fogarty2018regression-assisted} studies regression adjustment for matched pair designs, and \cite{liu2020regression-adjusted} study regression adjustment for stratified designs (primarily for the case where $\min\{\ell, k - \ell\} > 1$). \cite{ding2017paradox} considers randomization inference for matched pairs, alongside other randomization schemes. These papers study modifications of the Imai estimator for alternative experimental designs or with regression adjustment; they do not propose variance estimators based on comparisons between nearby stratum-level estimates. One recent exception is \cite{cytrynbaum2026coupling}, who employ a paired-stratum estimator for randomized experiments that generalize stratified designs to settings with complex treatments. Related ideas have a longer history in the survey-sampling literature, where collapsed-strata variance estimators have been proposed when fine stratification prevents direct within-stratum variance estimation \citep[see, e.g.,][p.~109]{sarndal1992model}. Our graph-Laplacian estimators can be viewed as a weighted collapsed-strata construction: they aggregate squared differences between pairs of stratum-level treatment-effect estimates according to graph weights. Imposing degree calibration on the underlying graph results in a simple finite-population bias identity, which is the starting point for our analysis. Our use of spectral restrictions to trade off covariate locality against worst-case robustness is also conceptually related to \cite{harshaw2024balancing}, who explore an analogous tradeoff in the choice of the treatment-assignment mechanism rather than the variance estimator, balancing global covariate means against worst-case precision. Finally, we note that all of the experimental-design papers above employ a design-based framework as we do in this paper. As a consequence, their analyses differ from work which studies inference in finely stratified experiments from a superpopulation perspective \citep[][]{bai2022optimality, bai2022inference, bai2024inference-b, bai2024covariate, bai2024inference, jiang2024bootstrap, cytrynbaum2024covariate, cytrynbaum2021optimal, bai2025efficiency}. See also \cite{abadie2008estimation}, who consider superpopulation estimation of the variance in matched pair designs, conditional on the covariates. Paired-stratum variance estimators are common in this work, but valid inference is justified under conditions in which the stratum pairings become close in the limit. In contrast, in the design-based setting we obtain valid inference without imposing such restrictions. We use similar restrictions only later, to quantify when such estimators lead to small biases.

The rest of the paper is organized as follows. Section \ref{sec:setup} fixes notation and the design-based framework. Section \ref{sec:main} first explains in what sense upward-biased variance estimation is necessary for standard confidence intervals, then reviews existing estimators and introduces the graph-Laplacian class. Section \ref{sec:minimax} uses the finite-population bias identity to compare graphs in a minimax environment and defines the regularized minimum-cost graph. Section \ref{sec:variances} turns to the complementary smooth-covariate regime and compares the limiting biases of the graph estimators. Section \ref{sec:sims} illustrates the same tradeoff in simulations, and Section \ref{sec:recommendations} provides recommendations for empirical practice.

\section{Setup and Notation} \label{sec:setup}

Consider an experiment consisting of $i \in \{1, \ldots, n\}$ units. For the $i$th unit, let $Y_{i} \in \mathbb R$ denote their observed outcome, $D_i \in \{0, 1 \}$ denote their received treatment, $X_i \in \mathbb R^{p_X}$ denote their observed, baseline covariates, and $Y_i(d)$ denote their potential outcome under treatment $d \in \{0,1\}$. As usual, the observed outcomes are related to the potential outcomes via the relationship
\begin{equation}\label{eq:obsY}
Y_i = Y_i(1)D_i + Y_i(0)(1 - D_i)~.
\end{equation}
In what follows, it will be convenient to use the following shorthand notation: for a generic vector $A_i$ indexed by $i$, let $A^{(n)} := (A_i: 1 \leq i \leq n)$.  For later use, we also define $W_i := (Y_i(1),Y_i(0),X_i)'$. 

In the design-based framework that we maintain throughout the paper, the potential outcomes and covariates of the units in the experiment are modeled as nonrandom quantities, with the only source of randomness arising from the treatment assignment mechanism. Our analysis concerns stratified designs in which the covariates are used to stratify units in the experiment into groups, i.e., strata, of \emph{fixed} size $k$, and then, within each group, $\ell < k$ units are chosen uniformly at random to assign to treatment and the remainder to control; of particular interest to us will be the special case when $\min\{\ell, k - \ell\} = 1$ (i.e. a \emph{finely stratified} design). For ease of exposition, we assume throughout that $n = mk$, so that $m$ denotes the number of strata. Formally, we model the strata as a partition of $\{1, \dots, n\}$:
$$\Lambda_n := \{ \lambda_j \subseteq \{1, \dots, n\}: 1 \leq j \leq m \}~,$$ 
with $|\lambda_j| = k$. It is worth emphasizing that we have suppressed in the notation the fact that each $\lambda_j$ (and therefore $\Lambda_n$ itself) can depend on $X^{(n)}$. Since we maintain that $k$ is fixed throughout the paper, when we write that $n \to \infty$, it should be understood that $m \to \infty$. Using this notation, the (joint) distribution of treatment assignment is characterized by the following assumption:

\begin{assumption} \label{ass:dn}
Treatment status $D^{(n)}$ is assigned independently for each $1 \leq j \leq m$ so that $$(D_i : i \in \lambda_j) \sim \text{Unif}\bigg (\bigg \{(d_1, \ldots, d_k) \in \{0,1\}^k : \sum_{1 \leq i \leq k} d_i = \ell \bigg \} \bigg)~.$$
\end{assumption}

\noindent In other words, $\ell$ out of $k$ units in each stratum are treated uniformly at random, independently across strata.

Our parameter of interest is the finite population average treatment effect, given by $$\Delta_n := \bar Y_n(1) - \bar Y_n(0)~,$$ where, for $d \in \{0,1\}$, 
$$\bar Y_n(d) := \frac{1}{n} \sum_{1 \leq i \leq n} Y_i(d)~.$$  
A natural estimator of $\Delta_n$ is given by the usual difference-in-means estimator, i.e., 
$$\hat \Delta_n := \frac{1}{n(1)} \sum_{1 \leq i \leq n} Y_i D_i - \frac{1}{n(0)} \sum_{1 \leq i \leq n} Y_i (1 - D_i)~,$$ 
where $n(1) = \sum_{1 \le i \le n} D_i = n \eta$, with $\eta := \ell/k$ the proportion of the $n$ units that are treated, and $n(0) = \sum_{1 \le i \le n}(1 - D_i) = n(1 - \eta)$.  Note that because we have assumed that $n = mk$, both the number of treated units $n \eta = m \ell$ and the number of untreated units $n (1 - \eta) = m (k - \ell)$ are integers.

\begin{remark}\label{rem:cluster}
The analysis in this paper can be readily extended to cluster RCTs where the treatment is implemented at the cluster level \citep[see, for instance,][]{imai2009essential, su2021model, de2024level, bai2024inferencecluster}. Let $i \in \{1, \dots, n\}$ denote the $i$th cluster and $g_i$ denote the number of units in the $i$th cluster. Let $D_i$ denote the treatment assignment for the $i$th cluster, which we emphasize applies to all units in this cluster. For each $i$ and $1 \leq t \leq g_i$, let $Y_{i, t}(1)$, $Y_{i, t}(0)$ denote the nonrandom potential outcomes for the $t$th unit in the $i$th cluster and let $Y_{i, t}$ denote its observed outcome. Suppose the parameter of interest is the average treatment effect across all units:
\[ \Delta_n^\dagger = \frac{1}{\sum_{1 \leq i \leq n} g_i} \sum_{1 \leq i \leq n} \sum_{1 \leq t \leq g_i} (Y_{i,t}(1) - Y_{i,t}(0))~. \]
Note that $\Delta_n^\dagger = \frac{1}{n}\sum_{1 \leq i \leq n} (Y_i^\dagger(1) -  Y_i^\dagger(0))$, where for $d \in \{0, 1\}$,
\[ Y_i^\dagger(d) = \frac{1}{\frac{1}{n}\sum_{1 \leq i \leq n} g_i} \sum_{1 \leq t \leq g_i} Y_{i,t}(d)~. \]
Accordingly, let $Y_i^\dagger = Y_i^\dagger(1) D_i + Y_i^\dagger(0) (1 - D_i)$. A natural estimator for $\Delta^{\dagger}_n$ is given by
\[ \hat \Delta_n^\dagger = \frac{1}{n(1)} \sum_{1 \leq i \leq n} Y_i^\dagger D_i - \frac{1}{n(0)}\sum_{1 \le i \le n} Y_i^\dagger (1 - D_i)~, \]
where $n(1)$ and $n(0)$ are defined as before, and in this case are the numbers of treated and untreated \emph{clusters}. The analysis in this paper then immediately applies by replacing $Y_i$ by $Y^{\dagger}_i$ throughout.
\end{remark}

\section{Upward-biased Variance Estimation and Graph-Laplacians}\label{sec:main}
In this section, we introduce a new class of variance estimators for finely stratified experiments, whose properties we study formally in Sections \ref{sec:minimax} and \ref{sec:variances}. We begin in Section \ref{sec:motivation} with a result that highlights the importance of variance estimators that are at least \emph{asymptotically} upward-biased, defined precisely in the statement of  Theorem \ref{thm:upward} below. Of course, a simple sufficient condition for an estimator to be asymptotically upward-biased is that it is upward-biased at every $n$.  This leads us, in Section \ref{sec:review}, to review prior proposals for constructing variance estimators in stratified experiments that are upward-biased under only the assumption that $D^{(n)}$ satisfies Assumption \ref{ass:dn}. In light of these prior proposals, Section \ref{sec:main_var} introduces a new class of  estimators of $\var[\hat{\Delta}_n]$ for finely stratified experiments which we call \emph{graph-Laplacian} estimators, and characterizes conditions under which estimators in this class are guaranteed to be upward-biased for $\var[\hat{\Delta}_n]$ under only Assumption \ref{ass:dn}. This characterization provides a unifying framework for many existing variance estimators, identifies which members of the graph-Laplacian class are expected to have the smallest upward bias in practice, and sets the stage for the analyses that follow in later sections.
%We further discuss in Remark \ref{rem:better_bounds} some improvements upon these estimators that typically result in estimators that are only asymptotically upward-biased under assumptions stronger than Assumption \ref{ass:dn} and therefore result in valid inference less generally than their upward-biased counterparts.  
\subsection{Motivating Result}\label{sec:motivation}
Our first theorem formalizes the sense in which an asymptotically upward-biased estimator for $\var[\hat{\Delta}_n]$ is required for valid inference on $\Delta_n$. Although results of this type are well-known, we provide it here for completeness and note in particular that the theorem establishes sufficiency \emph{and necessity}, which to our knowledge has not been formally documented in prior work.

\begin{theorem}\label{thm:upward}
Suppose $D^{(n)}$ satisfies Assumption \ref{ass:dn} and consider a sequence of finite populations such that 
\begin{equation}\label{eq:pop_moments}
\frac{1}{n}\max_{d \in \{0, 1\}}\max_{1 \leq i \leq n}Y_i(d)^2 \rightarrow 0
\end{equation}
as $n \rightarrow \infty$, and 
\begin{equation}\label{eq:pop_nondegen}
0 < \liminf_{n \to \infty} n \cdot \var[\hat{\Delta}_n] \le \limsup_{n \to \infty} n \cdot \var[\hat \Delta_n] < \infty~.
\end{equation}
Further assume that $\tilde{V}_n \geq 0$ is an estimator of $\var[\hat \Delta_n]$ such that
\begin{equation} \label{eq:consistent}
n \big | \tilde{V}_n - E[\tilde{V}_n] \big | \xrightarrow{P} 0
\end{equation}
as $n \rightarrow \infty$. Then,
\begin{equation} \label{eq:valid}
\liminf_{n \rightarrow \infty} P\Big\{\hat{\Delta}_n - \sqrt{\tilde{V}_n}\cdot z_{1-\alpha/2} \le \Delta_n \leq \hat{\Delta}_n + \sqrt{\tilde{V}_n}\cdot z_{1 - \alpha/2}\Big\} \ge 1 - \alpha~,    
\end{equation}
where $z_{1-\alpha/2}$ is the $1 - \alpha/2$ quantile of the standard normal distribution, if and only if
\begin{equation} \label{eq:upward}
\liminf_{n \to \infty} n \big ( E[\tilde{V}_n] - \var[\hat{\Delta}_n] \big ) \geq 0~.
\end{equation}
\end{theorem}

Note that because $\tilde{V}_n$ is defined as an estimator of $\var[\hat \Delta_n]$ and not the asymptotic variance of $\sqrt{n}(\hat \Delta_n - \Delta_n)$, the conditions in the theorem are all stated with a scaling by $n$. We refer to $n \big(E[\tilde V_n]-\var[\hat\Delta_n]\big)$ as the normalized bias of $\tilde V_n$. When \eqref{eq:upward} holds, we say that $\tilde{V}_n$ is \emph{asymptotically upward-biased} for $\var[\hat{\Delta}_n]$.  A simple sufficient condition for \eqref{eq:upward} is, of course, that 
\begin{equation} \label{eq:exactupward}
E[\tilde{V}_n] \ge \var[\hat \Delta_n] \text{ for all } n \geq 1~.
\end{equation}
When \eqref{eq:exactupward} holds, we say that $\tilde{V}_n$ is upward-biased for $\var[\hat{\Delta}_n]$. Theorem \ref{thm:upward} thus demonstrates that, under mild regularity conditions, confidence intervals for $\Delta_n$ of the form $[\hat \Delta_n \pm  \sqrt{\tilde{V}_n}\cdot z_{1 - \alpha/2}]$ are valid in the sense of having asymptotically at least nominal coverage if and only if $\tilde{V}_n$ is asymptotically upward-biased for $\var[\hat{\Delta}_n]$. These assumptions include weak restrictions on the finite population (i.e., \eqref{eq:pop_moments}), a non-degeneracy condition (i.e., \eqref{eq:pop_nondegen}), and a requirement that $n\cdot\tilde{V}_n$ is consistent for its expectation (i.e., \eqref{eq:consistent}). This last condition will typically hold under some further restrictions on finite population moments; see, e.g., Theorem \ref{thm:main} below. Finally, under the conditions of the theorem, if the biases of two variance estimators are ordered for all sufficiently large $n$, then their corresponding confidence-interval lengths inherit the same weak asymptotic ordering. Thus, smaller upward bias leads to asymptotically weakly shorter confidence intervals.

In what follows, we study variance estimators that can be shown to satisfy \eqref{eq:exactupward}, and therefore \eqref{eq:upward}, whenever $D^{(n)}$ satisfies Assumption \ref{ass:dn}. Such estimators yield confidence intervals for $\Delta_n$ that are valid whenever the remaining hypotheses of Theorem \ref{thm:upward} are satisfied. It is worth emphasizing, however, that establishing \eqref{eq:upward} for estimators that do not satisfy \eqref{eq:exactupward} may in general require conditions beyond those imposed by Theorem \ref{thm:upward}; Remarks \ref{rem:better_bounds} and \ref{rem:AI} provide examples. Our primary focus will thus be on studying a class of estimators that are upward-biased in the sense of \eqref{eq:exactupward}, because their validity follows under the hypotheses of Theorem \ref{thm:upward} alone.

\subsection{Review of Some Upward-Biased Estimators of \texorpdfstring{$\var[\hat{\Delta}_n]$}{Var[Delta-hat]}}\label{sec:review}

In this section, we review some upward-biased estimators of $\var[\hat{\Delta}_n]$.  To this end, recall that it can be shown using standard arguments \citep[see, for instance,][]{imbens2015causal} that
\begin{equation} \label{eq:varDeltahat}
\var[\hat{\Delta}_n] = \frac{1}{nm}\sum_{1 \leq j \leq m}\left(\frac{S^2_j(1)}{\eta} + \frac{S^2_j(0)}{1 - \eta} - S_{j, \Delta}^2\right)~,    
\end{equation}
where
\[S^2_j(d) := \frac{1}{k-1}\sum_{i \in \lambda_j}\left(Y_i(d)- \bar Y_{j, n}(d)\right)^2~, \hspace{5mm} S_{j, \Delta}^2 := \frac{1}{k-1}\sum_{i \in \lambda_j}\left(Y_i(1) - Y_i(0) - \Delta_{j,n}\right)^2~,\]
with
\[\bar Y_{j, n}(d) := \frac{1}{k}\sum_{i \in \lambda_j}Y_i(d)~, \hspace{5mm} \Delta_{j,n} := \frac{1}{k}\sum_{i \in \lambda_j}\left(Y_i(1) - Y_i(0)\right)~.\]
In settings where $\min\{\ell, k - \ell\} > 1$, the construction of an upward-biased variance estimator is straightforward: an unbiased estimator of $S^2_j(d)$ for $d \in \{0, 1\}$ is given by the sample variance of the outcomes for units within stratum $j$ assigned to treatment $d$, which we denote by $\hat{S}^2_j(d)$ \citep[see][]{imbens2015causal, pashley2021insights}. A simple estimator of $\var[\hat{\Delta}_n]$ is thus given by
\begin{equation}\label{eq:var_coarse}\frac{1}{nm}\sum_{1 \le j \le m}\left(\frac{\hat{S}^2_j(1)}{\eta} + \frac{\hat{S}^2_j(0)}{1 - \eta}\right)~,
\end{equation}
and its bias is $\frac{1}{nm}\sum_{1 \le j \le m}S_{j, \Delta}^2 \ge 0$. Note that this estimation strategy exploits the lower bound $S_{j, \Delta}^2 \ge 0$; we return to this observation in Remark \ref{rem:better_bounds} below.

This estimation strategy fails, however, in finely stratified designs (i.e. when $\min\{\ell, k - \ell\} = 1$) since in this case the corresponding within-stratum sample variance is identically zero.  In these settings, a canonical estimator considered in the literature instead computes the variance of the stratum-level average treatment effect estimates \citep[][]{imai2008variance}:
\begin{equation} \label{eq:imai}
\hat{V}_n^{\rm IM} := \frac{1}{m (m - 1)}\sum_{1 \leq j \leq m}\left(\hat{\Delta}_{j,n} - \hat{\Delta}_n\right)^2~,    
\end{equation}
where
\[\hat{\Delta}_{j,n} := \frac{1}{\ell}\sum_{i \in \lambda_j}Y_iD_i - \frac{1}{k - \ell}\sum_{i \in \lambda_j}Y_i(1-D_i) \]
is the difference in means in the $j$th stratum. This variance estimator serves as the basic scaffolding for many of the recent estimators proposed in the literature on inference in stratified randomized experiments: it is a special case of the small-block variance estimator proposed in \cite{pashley2021insights} when all strata are the same size \citep[see also][]{zhu2024design-based}, a special case of the pair-cluster variance estimator of \cite{de2024level} in the setting of an individual-level randomized experiment, and a special case of the regression-based estimator due to \cite{fogarty2018mitigating}. Note that the bias of $\hat{V}^{\rm IM}_n$ is given by \citep[see][]{imbens2015causal,fogarty2018mitigating}
\begin{equation} \label{eq:imai_bias}
E[\hat{V}^{\rm IM}_n] - \var[\hat{\Delta}_n] = \frac{1}{m(m-1)}\sum_{1 \leq j \leq m}(\Delta_{j,n} - \Delta_n)^2 \geq 0~,    
\end{equation}
so that $\hat{V}^{\rm IM}_n$ is an upward-biased estimator of $\var[\hat{\Delta}_n]$. Moreover, under mild assumptions, it is straightforward to establish using Chebyshev's inequality that $\hat V_n^{\rm IM}$ is consistent for its expectation in the sense of \eqref{eq:consistent}, and thus can be used to construct valid confidence intervals in the sense of \eqref{eq:valid}. Note that unlike \eqref{eq:var_coarse}, $\hat V_n^{\rm IM}$ obtains its upward bias through the variation in the stratum-level treatment effect estimates across all strata. This observation motivates our proposal in the next section, where we introduce variance estimators for finely stratified designs that can borrow information across strata in a more structured way.

\begin{remark}\label{rem:better_bounds}
In some circumstances, it may be further possible to bound $S_{j, \Delta}^2$ from below by a positive number instead of zero, and this tighter bound could then be used to construct less conservative estimators of $\var[\hat \Delta_n]$. For instance, it follows from the Cauchy--Schwarz inequality that
\[ \sum_{i \in \lambda_j} (Y_i(1) - \bar Y_{j, n}(1)) (Y_i(0) - \bar Y_{j, n}(0)) \leq \bigg ( \sum_{i \in \lambda_j} (Y_i(1) - \bar Y_{j, n}(1))^2 \bigg )^{1/2} \bigg ( \sum_{i \in \lambda_j} (Y_i(0) - \bar Y_{j, n}(0))^2 \bigg )^{1/2}~,\]
from which we can deduce the lower bound 
\[S_{j, \Delta}^2 \ge (S_j(1) - S_j(0))^2 \ge 0~.\]
In fact, it is sometimes possible to achieve even tighter bounds; see \cite{aronow2014sharp} for details. We note, however, that we are not aware of estimators based on lower bounds other than $S_{j, \Delta}^2 \ge 0$ which are necessarily guaranteed to be upward-biased, and in general additional assumptions  need to be imposed in order to guarantee that the resulting estimators are at least asymptotically upward-biased. See, for instance, Corollary 1 in \cite{aronow2014sharp}. As a consequence, these improved estimators may not lead to valid confidence intervals whenever these additional assumptions fail to hold. 
\end{remark}

\begin{remark}\label{rem:hybrid_designs}
The fixed-$k$, fixed-$\ell$ design considered in this paper is useful for isolating the variance-estimation problem that arises in finely stratified experiments, and captures many designs employed in practice. In some applications, however, strata may have different sizes and assignment fractions, and the same experiment may contain both ``large'' strata with at least two treated and two control units and ``fine'' strata with only one treated or one control unit. A simple way to handle such hybrid designs is to decompose the experiment into subcollections of strata to which the same variance-estimation rule will be applied, as in \cite{pashley2021insights}. To make this idea concrete, suppose the strata are indexed by $1 \le j \le m$, let $n_j=|\lambda_j|$, and partition the strata into disjoint subcollections $\mathcal G_1,\ldots,\mathcal G_R$. For instance, $\mathcal G_r$ might collect all strata which are triples with one treated unit. For $1 \le r \le R$, define
\[
N_r=\sum_{j\in\mathcal G_r}n_j~,\qquad
q_r=\frac{N_r}{n}~,\qquad
\hat\Delta^{(r)}_n=\sum_{j\in\mathcal G_r}\frac{n_j}{N_r}\hat\Delta_{j,n}~.
\]
Then, the overall estimator can be written as
\[
\sum_{1 \le r \le R}q_r\hat\Delta^{(r)}_n~,
\]
the variance of which equals $\sum_{1 \le r \le R}q_r^2\var[\hat\Delta^{(r)}_n]$
because of the independence of treatment assignments across strata. If $\mathcal G_r$ consists of large strata, then $\var[\hat\Delta^{(r)}_{n}]$ can be estimated by applying the within-stratum estimator \eqref{eq:var_coarse} directly to the strata inside $\mathcal{G}_r$.
If $\mathcal G_r$ consists of fine strata, one can instead estimate $\var[\hat\Delta^{(r)}_{n}]$ from comparisons across the stratum-level estimators inside $\mathcal G_r$ using for instance \eqref{eq:imai} or the graph-Laplacian estimators introduced below. The final variance estimator is then obtained by adding these subcollection-specific components. 
\end{remark}

\subsection{Graph-Laplacian Estimators of \texorpdfstring{$\var[\hat\Delta_n]$}{Var[Delta-hat]}}\label{sec:main_var}
We now introduce a class of variance estimators which generalize the estimator in \eqref{eq:imai} in that they borrow information across strata and thus remain well-defined even when $\min\{\ell,k-\ell\}=1$. The construction has a graph-theoretic interpretation: the strata form the vertices of a graph, and the estimator places weights on the graph's edges. Each edge then contributes a
weighted squared difference of the treatment-effect estimates for the two adjacent strata. One motivation for this construction is as follows: by the independence of treatment assignments across strata, the variance of $\hat\Delta_n$ can be written as
\[
\var[\hat\Delta_n]
=
\frac{1}{m^2}\sum_{1 \le j \le m} \var[\hat\Delta_{j,n}]~.
\]
As explained in Section \ref{sec:review}, in finely stratified designs the stratum-level variances $\var[\hat\Delta_{j,n}]$ generally cannot be estimated entirely within strata. However, for two distinct strata $1 \leq a \neq b \leq m$,
\[
E[(\hat\Delta_{a,n}-\hat\Delta_{b,n})^2]
=
\var[\hat\Delta_{a,n}]+\var[\hat\Delta_{b,n}]+(\Delta_{a,n}-\Delta_{b,n})^2~.
\]
Thus, squared differences between stratum-level difference-in-means estimators contain information about the corresponding stratum-level variances, but also include a bias term reflecting heterogeneity in the true stratum-level average treatment effects. This suggests choosing which squared differences to average
so that the resulting heterogeneity term is small. We show below that this choice can equivalently be formulated as deciding how to construct a weighted graph whose vertices are the strata.

Towards this end, for strata $1\leq a,b\leq m$, let $\omega_{ab,m}$ denote the nonrandom weight assigned to the edge between strata $a$ and $b$. Throughout the remainder of the paper, we maintain without further comment that the graph weights are nonnegative and symmetric, with zero diagonal:
\begin{equation} \label{eq:graph_weights}
\omega_{ab,m}=\omega_{ba,m}\ge 0,\qquad \omega_{aa,m}=0~.
\end{equation}
We define the weighted degree of stratum $1\leq j\leq m$ by
\[
\mathrm{deg}_{j,m}:=\sum_{b\ne j}\omega_{jb,m}~.
\]
For a vector of graph weights $\omega_m = (\omega_{ab,m}: 1 \le a, b \le m)$, the \emph{graph-Laplacian variance estimator} is given by
\begin{equation}\label{eq:graph_estimator}
\hat V_n(\omega_m)
:=
\frac{1}{m^2}\sum_{1\le a<b\le m}
\omega_{ab,m}(\hat\Delta_{a,n}-\hat\Delta_{b,n})^2~.
\end{equation}
$\hat{V}_n(\omega_m)$ arises naturally if we view it as the quadratic form associated with the graph-Laplacian corresponding to $\omega_m$ \citep[see, for instance][]{chung1997spectral, vonluxburg2007tutorial}. Formally, the graph-Laplacian $L_m(\omega_m)$ is the $m \times m$ symmetric matrix with diagonal entries
\begin{equation} \label{eq:laplacian_diag}
\big [ L_m(\omega_m) \big ]_{jj}=\mathrm{deg}_{j,m}
\end{equation}
for $1 \leq j \leq m$ and off-diagonal entries
\begin{equation} \label{eq:laplacian_offdiag}
\big [ L_m(\omega_m) \big ]_{ab}=-\omega_{ab,m}
\end{equation}
for $1 \leq a, b \leq m$, $a \neq b$. For any $x=(x_1,\ldots,x_m)'$, it follows that
\begin{equation} \label{eq:x_quadratic}
x'L_m(\omega_m)x
=
\sum_{1\le a<b\le m}\omega_{ab,m}(x_a-x_b)^2~,       
\end{equation}
so that \eqref{eq:graph_estimator} implies that the variance estimator can be written as
\begin{equation} \label{eq:laplacian_quadratic}
\hat{V}_n(\omega_m) = \frac{1}{m^2}\hat{\delta}_n'L_m(\omega_m)\hat{\delta}_n~,
\end{equation}
where $\hat{\delta}_n = (\hat{\Delta}_{1,n}, \ldots, \hat{\Delta}_{m,n})'$. 

The estimator $\hat{V}^{\rm IM}_n$ in \eqref{eq:imai} corresponds to the graph-Laplacian estimator in \eqref{eq:graph_estimator} based on the complete graph $\omega_m^{\rm CG}$, which assigns weight $1/(m-1)$ to every edge:
\[
\omega_{ab,m}^{\rm CG}=\frac{1}{m-1}~,
\qquad 1 \leq a\ne b \leq m~.
\]
To see this, note that by the identity
\[
\frac{1}{m}\sum_{1\le a<b\le m}(x_a-x_b)^2 = \sum_{1 \le j \le m} (x_j-\bar x)^2~,
\]
we obtain
\[
\hat V_n(\omega_m^{\rm CG})
=
\frac{1}{m^2}
\sum_{1\le a<b\le m}\frac{1}{m-1}
(\hat\Delta_{a,n}-\hat\Delta_{b,n})^2
= \frac{1}{m(m-1)}\sum_{1 \leq j \leq m}(\hat \Delta_{j,n} - \hat \Delta_n)^2= \hat V_n^{\rm IM}~.
\]
From this perspective, the complete graph places equal weight across every stratum pair, regardless of how similar the treatment effects are between strata. Alternatively, we could consider matching together only those strata whose stratum-level treatment effects we expect to be most similar, by constructing a graph corresponding to a \emph{perfect matching} of the strata. Suppose $m$ is even and define
\[
\mathcal P_m
=
\left\{
\{a_j,b_j\}: 1 \le j \le m/2
\right\}
\]
to be a perfect matching of the stratum indices: that is, the sets \(\{a_j,b_j\}\) are disjoint two-element subsets that partition \(\{1,\ldots,m\}\). Set $\omega_{ab,m}=I\{\{a,b\}\in\mathcal P_m\}$ for $1 \leq a \neq b \leq m$. Then,
\[
\hat V_n(\omega_m)
=
\frac{1}{m^2}\sum_{1 \leq a < b \leq m: \{a,b\}\in\mathcal P_m}
(\hat\Delta_{a,n}-\hat\Delta_{b,n})^2~,
\]
is the associated graph-Laplacian estimator.

Theorem \ref{thm:main} derives the finite-population bias of the graph-Laplacian estimator and states conditions under which the estimator is upward-biased and consistent for its expectation.
% This equals the paired-strata estimator from the preceding draft.  To see this, define
% \[
% \hat \tau_n^2=\frac{1}{m}\sum_{1 \le j \le m}\hat\Delta_{j,n}^2,
% \qquad
% \hat\kappa_n=\frac{2}{m}\sum_{j=1}^{m/2}
% \hat\Delta_{2j-1,n}\hat\Delta_{2j,n}.
% \]
% Then
% \[
% \frac{1}{m^2}\sum_{j=1}^{m/2}
% (\hat\Delta_{2j-1,n}-\hat\Delta_{2j,n})^2
% =
% \frac{1}{m}\left\{
% \frac{1}{m}\sum_{1 \le j \le m} \hat\Delta_{j,n}^2
% -
% \frac{2}{m}\sum_{j=1}^{m/2}\hat\Delta_{2j-1,n}\hat\Delta_{2j,n}
% \right\}
% =\frac{1}{m}(\hat\tau_n^2-\hat\kappa_n).
% \]
% If $m$ is odd, one can use the same leftover-stratum convention as in the paired construction; the degree-calibrated graph formulation itself is most transparent for the perfect-matching case with even $m$, or for fractional graphs such as the complete graph and the linear-programming graphs introduced below.

\begin{theorem}\label{thm:main}
Suppose Assumption \ref{ass:dn} holds. Let $\omega_m=(\omega_{ab,m}: 1 \le a,b \le m)$ be nonrandom graph weights that satisfy \eqref{eq:graph_weights}. Then,
\begin{enumerate}[\rm (a)]
\item The bias of $\hat{V}_n(\omega_m)$ is given by
\begin{equation}\label{eq:bias}
E[\hat V_n(\omega_m)]-\var[\hat\Delta_n]
=
\frac{1}{m^2}\sum_{1\le a<b\le m}
\omega_{ab,m}(\Delta_{a,n}-\Delta_{b,n})^2
+
\frac{1}{m^2}\sum_{1 \le j \le m}(\mathrm{deg}_{j,m}-1)\var[\hat{\Delta}_{j,n}]~.
\end{equation}
Consequently, if $\mathrm{deg}_{j,m}\ge 1$ for every $1 \le j \le m$, then $\hat V_n(\omega_m)$ is upward-biased for $\var[\hat\Delta_n]$.
\item 
Suppose in addition that
\begin{equation}\label{eq:deg_bound}
\max_{1\le j\le m}\mathrm{deg}_{j,m}\le C
\end{equation}
for some constant $C<\infty$ and
\begin{equation}\label{eq:four_moment}
\frac{1}n\sum_{1 \le i \le n} Y_i(d)^4=o(n), \qquad d\in\{0,1\}~.
\end{equation}
Then,
\[
n \big | \hat V_n(\omega_m)-E[\hat V_n(\omega_m)] \big | \xrightarrow{P} 0~,
\]
so that \eqref{eq:consistent} holds as $n \rightarrow \infty$.
\end{enumerate}
\end{theorem}
Theorem \ref{thm:main} justifies using $\hat{V}_n(\omega_m)$ for inference about $\Delta_n$ and further demonstrates that the magnitude of its bias depends explicitly on how the graph $\omega_m$ places weights on the strata vertices. Interpreted as a function of $\omega_m$, the bias is smaller when the graph places more weight on pairs of strata which have similar stratum-level average treatment effects. We say that the graph is \emph{degree-calibrated} if
\begin{equation}\label{eq:degree_calibration}
\mathrm{deg}_{j,m} = \sum_{b\ne j}\omega_{jb,m}=1
\text{ for every }1 \le j \le m~.
\end{equation}
Degree calibration means that each stratum has weighted degree equal to one, just as in a perfect matching, but allows for this unit of edge weight to be distributed fractionally across several neighboring strata.  If $\omega_m$ is degree-calibrated then $\mathrm{deg}_{j,m} = 1$ for all $1 \le j \le m$ and thus by Theorem \ref{thm:main},
\[
E[\hat V_n(\omega_m)]-\var[\hat\Delta_n]
=
\frac{1}{m^2}\sum_{1\le a<b\le m}
\omega_{ab,m}(\Delta_{a,n}-\Delta_{b,n})^2\ge 0~.
\]
As a consequence, degree-calibration guarantees that $\hat{V}_n(\omega_m)$ is upward-biased and ensures that the ``excess'' bias in the second component of the expression in \eqref{eq:bias} is zero. Accordingly, in the results that follow we will restrict ourselves to graph-Laplacian estimators constructed from degree-calibrated graphs. The specific choice of graph controls the size of the bias; this bias will be small whenever vertices joined by large-weight edges have similar average treatment effects. 

In practice, the stratum-level treatment effects will not be known, so the graph must be constructed using baseline observed covariates. With this in mind, for $1 \le j \le m$, let
\[ \bar X_{j,n} = \frac{1}{k}\sum_{i \in \lambda_j}X_i\]
denote the stratum-level covariate means. For pairs of strata $1 \le a, b \le m$ let $c_{ab,m}$ denote a pairwise stratum \emph{cost} constructed from $\bar X_{a,n}$ and $\bar X_{b,n}$. It could be, for instance, the squared Euclidean distance $\|\bar X_{a,n} - \bar X_{b,n}\|^2$. Suppose for the moment that $m$ is even and define
\[
\mathcal M_m
:=
\left\{
\omega = (\omega_{ab})_{1 \leq a, b \leq m}:\omega_{ab}=\omega_{ba}\in\{0,1\} \text{ and } \omega_{aa}=0~,
\sum_{b\ne a}\omega_{ab}=1\text{ for } 1 \leq a \leq m
\right\}~,
\]
so that each $\omega_m\in\mathcal M_m$ is a perfect matching of the strata. A natural implementation of the perfect-matching estimator chooses a \emph{minimum-cost} perfect matching:
\begin{equation} \label{eq:perfect_matching}
\omega_m^{\rm PM}\in\argmin_{\omega_m\in\mathcal M_m}
\sum_{1\le a<b\le m} c_{ab,m}\omega_{ab,m}~.    
\end{equation}
Methods to solve such minimum-cost perfect matching problems are well-studied in the literature: see, for instance, \cite{derigs1988solving}, \cite{greevy2004optimal}. The bias of $\hat{V}_n(\omega_m^{\rm PM})$ should be small whenever $\omega_m^{\rm PM}$ joins vertices with similar average treatment effects. In Section \ref{sec:variances}, we formalize this intuition by defining an asymptotic benchmark, denoted $V^{\rm obs}$, and providing conditions under which the minimum-cost perfect matching attains it.

The intuitive appeal of this construction is that it uses baseline covariates to place edge weight on strata that are expected to have similar treatment effects. However, the design-based paradigm does not immediately ensure that this will be the case. In the remainder of the paper we consider the bias properties of this construction, and degree-calibrated graph choices more generally, from two perspectives. In Section \ref{sec:minimax} we quantify how costly a perfect matching can be from a minimax-bias perspective when no restrictions link the baseline covariates to the unknown treatment effects; this motivates a regularized version of the perfect matching construction. In Section \ref{sec:variances} we introduce an asymptotic framework where baseline covariates are informative about treatment-effect heterogeneity and contrast the limiting behavior of $\hat{V}_n(\omega_m)$ for ``local'' graphs like $\omega_m^{\rm PM}$ versus the ``non-local'' complete graph $\omega_m^{\rm CG}$. 

\begin{remark}\label{rem:AI}
The nearest-neighbor variance estimator developed in \cite{abadie2008estimation}, originally proposed for estimating the conditional variance $\var[\hat{\Delta}_n|X^{(n)}]$ in a super-population framework, could also in principle be used as an estimator of $\var[\hat{\Delta}_n]$ in the design-based setting we consider here. Indeed, their estimator can also be written as a graph-Laplacian estimator, but we note that the estimator will not generally be degree-calibrated.  Indeed, for each stratum $1 \le j \le m$, let $\eta(j) \in \{1, \ldots m\} \setminus \{j\}$ be stratum $j$'s ``nearest neighbor'' (in practice, $\eta(j)$ would be selected using baseline covariates, but for the purposes of this discussion it is not relevant exactly how $\eta(j)$ is selected).  The Abadie--Imbens estimator is given by
\[
\hat V_n^{\rm AI}
=
\frac{1}{2m^2}\sum_{1 \le j \le m}
(\hat\Delta_{j,n}-\hat\Delta_{\eta(j),n})^2~.
\]
Equivalently,
\[
\hat V_n^{\rm AI}=\hat V_n(\omega_m^{\rm AI})~,
\qquad
\omega_{ab,m}^{\rm AI}
=
\frac{1}{2} \big ( I\{\eta(a)=b\}+ I\{\eta(b)=a\} \big )~,
\]
so
\[
\mathrm{deg}_{j,m}=\frac{1+|\{b:\eta(b)=j\}|}{2}~.
\]
Therefore, the estimator $\hat V_n^{\rm AI}$ is degree-calibrated only in the special case where every stratum is used exactly once as another stratum's nearest neighbor. Outside of this special case, $\mathrm{deg}_{j,m}$ is not guaranteed to be larger than $1$ and as a consequence the estimator is not generally upward-biased under Assumption \ref{ass:dn} alone.
\end{remark}

\section{Minimax Bias of Graph-Laplacian Variance Estimators}\label{sec:minimax}

Theorem \ref{thm:main} demonstrates that variance estimators constructed from degree-calibrated graphs are guaranteed to be upward-biased, and thus through the lens of Theorem \ref{thm:upward} any such estimator can be used to construct valid confidence intervals for $\Delta_n$. In Section \ref{sec:main_var}, we argued that the bias-minimizing graph should join vertices with similar stratum-level average treatment effects. However, in practice these treatment effects are not known a priori.  In this section, we study this bias minimization problem from a minimax perspective in a purely design-based environment.

The first result of this section establishes that variance estimators constructed from degree-calibrated graphs cannot be uniformly ordered in terms of their biases without further restrictions on the finite population.
\begin{proposition}\label{prop:no_uniform_ordering}
Let $\omega_m$ and $\tilde \omega_m$ be two degree-calibrated graphs. If
\begin{equation}\label{eq:order}
\sum_{1 \le a<b \le m}\omega_{ab,m}(x_a-x_b)^2
\le
\sum_{1 \le a<b \le m}\tilde \omega_{ab,m}(x_a-x_b)^2
\end{equation}
for every $x\in\mathbb R^m$, then $\omega_m=\tilde \omega_m$.
\end{proposition}

Proposition \ref{prop:no_uniform_ordering} shows that the bias of two non-identical degree-calibrated graph-Laplacian variance estimators, including those obtained from the complete graph and a perfect-matching graph in particular, cannot be uniformly ordered over all finite populations. The result follows from the quadratic-form representation of the bias: if $L_m$ and $\tilde L_m$ are the graph-Laplacians associated with $\omega_m$ and $\tilde \omega_m$, then \eqref{eq:order} implies $\tilde L_m - L_m$ must be positive semidefinite. Degree-calibration further implies that the diagonal elements of $\tilde L_m - L_m$ must be zero. Combining these two facts implies that $\tilde L_m - L_m$ must be identically zero.

In light of Proposition \ref{prop:no_uniform_ordering}, a natural next question is which graph-Laplacian estimator is \emph{minimax} bias-optimal. Towards that end, recall from Theorem \ref{thm:main} that for a degree-calibrated graph, the (normalized) bias is given by
\[
n(E[\hat V_n(\omega_m)]-\var[\hat\Delta_n])
=
k \cdot \frac{1}{m}\delta_n'L_m(\omega_m)\delta_n~,
\]
where $\delta_n = (\Delta_{1, n}, \dots, \Delta_{m, n})'$. Since $k$ is fixed throughout the paper, the minimax comparison is unchanged by this constant factor.
Define the projection matrix
\[P_m = I_m - \frac{1}{m}\iota_m\iota_m'~,\]
and, for $x\in\mathbb R^m$, write $\|x\|_m^2=m^{-1}\sum_{1 \le j \le m}x_j^2$.  Then, $P_m \delta_n = (\Delta_{1, n} - \Delta_n, \dots, \Delta_{m, n} - \Delta_n)'$, and $\|P_m \delta_n\|_m^2$ is the population variance of stratum-level average treatment effects. Let
\[\Omega_m:= \left\{\omega = (\omega_{ab})_{1 \leq a, b \leq m}: \omega_{ab} = \omega_{ba} \ge 0 \text{ for } 1 \leq a, b \leq m, \omega_{aa} = 0 \text{ and } \sum_{b \ne a}\omega_{ab} = 1 \text{ for } 1 \leq a \leq m \right\}~,\]
denote the set of all degree-calibrated graphs on $m$ strata. For conformable matrices $A$ and $B$, we write $A \preceq B$ whenever $B - A$ is positive semidefinite. The theorem below characterizes the graph-Laplacian estimator that is minimax optimal within $\Omega_m$ for the normalized bias over the class of populations with bounded stratum-level treatment effect variance.
%Note that because $\omega_m \in \Omega_m$ is degree-calibrated, its Laplacian satisfies
%\[ L_m(\omega_m) = I_m - \omega_m~. \]
\begin{theorem}\label{thm:minimax}
Fix $C>0$. Then,
\begin{enumerate}[\rm (a)]
    \item \[
\inf_{\omega_m\in\Omega_m}
\sup_{\delta_n \in \mathbb R^m: \|P_m \delta_n\|_m^2\le C}
\frac{1}{m}\delta_n'L_m(\omega_m)\delta_n
=
C\frac{m}{m-1}~,
\]
where the infimum is uniquely attained by $\omega_m^{\rm CG}$. 
\item If $\omega_m\in\Omega_m$ satisfies $L_m(\omega_m)\preceq \kappa_m P_m$ for $\kappa_m \geq 1$, then
\[
\sup_{\delta_n \in \mathbb R^m: \|P_m \delta_n\|_m^2\le C}
\frac{1}{m}\delta_n'L_m(\omega_m)\delta_n
\le
\kappa_m C~.
\]
\item Suppose $m$ is even and let $\omega_m\in\mathcal M_m$ be any perfect matching graph. Then,
\[
\sup_{\delta_n \in \mathbb R^m: \|P_m \delta_n\|_m^2\le C}
\frac{1}{m}\delta_n'L_m(\omega_m)\delta_n
=
2C~.
\]
\end{enumerate}
\end{theorem}
Theorem \ref{thm:minimax} (a) demonstrates that the corresponding graph-Laplacian estimator $\hat V_n(\omega_m^{\rm CG})$ is uniquely minimax optimal for normalized bias within $\Omega_m$ when treatment-effect heterogeneity is restricted only through a bound on the variance of stratum-level average treatment effects. Theorem \ref{thm:minimax}(b) further shows that if the graph-Laplacian of $\omega_m$ is ``bounded'' in an appropriate sense by $\kappa_m P_m$, then its maximum normalized bias is bounded proportionally to the constant $\kappa_m$. Note that in order for $L_m(\omega_m) \preceq \kappa_m P_m$ to be feasible, it has to be the case that $\kappa_m \geq \frac{m}{m - 1}$. Moreover, it follows from direct calculation that $L_m(\omega_m^{\rm CG}) = \frac{m}{m - 1} P_m$, so that the bound in part (b) is attained by $\omega_m^{\rm CG}$ when $\kappa_m=\frac{m}{m-1}$ and coincides with the minimax value in part (a). To gain some intuition for the role of $\kappa_m$, let $e_j \in \mathbb R^m$ denote the $j$th standard basis vector. Applying the definition of $\kappa_mP_m - L_m(\omega_m)$ being positive semidefinite to the vector $e_a - e_b$ for $1 \leq a \neq b \leq m$, we obtain by direct calculation that
\[\omega_{ab,m} \le \kappa_m - 1~,\]
so that the constraint $L_m(\omega_m) \preceq \kappa_m P_m$ forces the graph to ``spread'' weight over a larger set of neighbors when $\kappa_m$ is small. The complete graph spreads weight over all possible neighbors. In contrast, a perfect matching graph applies each edge weight to a single neighbor; Theorem \ref{thm:minimax}(c) computes the corresponding maximum normalized bias in this case.
%Moreover, note from \eqref{eq:imai_bias} that this variance is proportional to the bias of $\hat V_n(\omega_m^{\rm CG})$.
Combining Theorem \ref{thm:minimax}(a) and (c), we see that the worst-case normalized bias of the variance estimator constructed from any perfect matching is $2(m-1)/m$ times the worst-case normalized bias of $\hat{V}_n(\omega_m^{\rm CG})$. Thus, while a perfect matching may have much smaller bias when it pairs strata with similar treatment effects, Theorem \ref{thm:minimax} demonstrates that it also exposes the estimator to larger worst-case normalized bias.

The discussion above motivates us to consider the following \emph{regularized} minimum-cost optimization problem, which keeps the baseline-covariate cost criterion used for the perfect matching while imposing a spectral constraint on the graph-Laplacian that controls worst-case normalized bias:
\begin{equation}\label{eq:sdp}
\begin{aligned}
\omega_m^{\rm SDP}(\kappa_m)
\in \argmin_{\omega_m \in \Omega_m}\quad
& \sum_{1 \le a < b \le m} c_{ab,m} \omega_{ab,m} \\
\text{subject to}\quad
& L_m(\omega_m) \preceq \kappa_m P_m~.
\end{aligned}
\end{equation}
Note that $\omega_m^{\rm SDP}(\kappa_m)$ is the fractional relaxation of the minimum-cost perfect matching problem, with an additional regularization constraint. This problem is a semidefinite program which can be solved using existing software packages such as {\tt CVXR}. Recalling the second conclusion of Theorem \ref{thm:minimax}, this constraint ensures that the maximum normalized bias of $\omega_m^{\rm SDP}$ over the set of populations such that $\|P_m\delta_n\|_m^2 \le C$ is bounded by $\kappa_m C$.  In Section \ref{sec:variances} we explain how a careful choice of the tuning parameter $\kappa_m$ allows $\hat{V}_n(\omega_m^{\rm SDP}(\kappa_m))$ to attain the observable asymptotic benchmark defined there, while simultaneously ensuring that its worst-case normalized bias in the sense of Theorem \ref{thm:minimax} can be made arbitrarily close to the minimax value, asymptotically.

\begin{remark}\label{rem:alternatives}
Instead of using $\omega_m^{\rm SDP}(\kappa_m)$, one might instead consider combining the benefits of $\omega_m^{\rm PM}$ and the complete graph $\omega_m^{\rm CG}$ using a convex combination 
\[
(1-\pi_m)\omega_m^{\rm CG}+\pi_m \omega_m^{\rm PM}~,
\]
with $0 \le \pi_m \le 1$. This graph is degree-calibrated and thus would produce an upward-biased graph-Laplacian variance estimator. However, this graph cannot generally accomplish both the goals of attaining the observable asymptotic benchmark defined in Section \ref{sec:variances} and guarding against worst-case normalized bias at the same time. Approaching the minimax normalized-bias bound requires $\pi_m\to0$, while attaining the observable asymptotic benchmark generally requires $\pi_m\to1$. The regularized graph in \eqref{eq:sdp} is designed to avoid this conflict by spreading each stratum's edge weight across many nearby strata: the regularization constraint controls worst-case normalized bias, while a careful selection of the edges allows the resulting estimator to  match the bias of the minimum-cost perfect matching estimator asymptotically.
\end{remark}

The next section studies the complementary setting in which the baseline covariates used to construct the graph are informative about treatment-effect heterogeneity. In that setting, ``local'' graph sequences can reduce asymptotic bias relative to the complete graph while preserving the finite-population upward-bias guarantee from Theorem \ref{thm:main}.

% Note however, that $\omega_m$ being degree-calibrated does not require that $\omega_m$ be a perfect matching, and in practice the integrality restriction can force nonlocal comparisons. For instance, if a group of three strata are all close to one another and far away from the other strata, a perfect matching must match one of these three strata to a stratum outside of the group.  In contrast, a \emph{fractional} graph can keep all comparisons within the group by setting
% \[
% \omega_{12}=\omega_{13}=\omega_{23}=\frac{1}2~,
% \]
% which still maintains degree calibration. As an illustration of this phenomenon in practice, consider the randomized evaluation of Afghanistan's National Solidarity Programme in \citep{beath2015nsp}.  The evaluation covered 500 villages in 10 districts, with the 50 villages in each district paired into 25 matched pairs. Viewing these matched village pairs as our strata, a perfect matching cannot remain entirely within district. Related issues may arise in experiments with many small sites, districts, schools, or markets when these local sites contain an odd number of original strata.

% These considerations motivate replacing the integral matching problem by its fractional relaxation.  Define the class of degree-calibrated fractional graphs by
% \[
% \Omega_m
% :=
% \left\{
% \omega:\omega_{ab}=\omega_{ba}\ge 0,\ \omega_{aa}=0,
% \sum_{b\ne a}\omega_{ab}=1\text{ for all }a
% \right\}.
% \]

\section{Asymptotics for Graph-Laplacian Variance Estimators} \label{sec:variances}

In this section, we provide a framework that permits us to further discriminate among the variance estimators of $\var[\hat{\Delta}_n]$ introduced in Sections \ref{sec:main} and \ref{sec:minimax} by characterizing the limit of $nE[\hat{V}_n(\omega_m)]$ in a framework where baseline covariates are smoothly related to the underlying treatment effect heterogeneity. The framework below is formalized by restrictions on the sequence of finite populations, strata, and graph weights.  These restrictions are motivated by a limiting thought experiment in which $W^{(n)}$ is itself a realization of an i.i.d.\ sample from a fixed probability distribution, and the strata $\Lambda_n$ are formed in such a way that units with similar baseline observed covariates are grouped together.  
%As explained in Remark \ref{rem:equi}, however, the restrictions can also be shown to hold in other contexts.  
% For a random vector $\tilde W=(\tilde Y(1),\tilde Y(0),\tilde X)\sim Q$, define
% \[
% \mu_d(x):=E_Q[\tilde Y(d)\mid \tilde X=x],
% \qquad
% \theta(x):=\mu_1(x)-\mu_0(x),
% \]
% and
% \[
% \Omega
% :=
% E_Q\left[\var_Q[\tilde Y(1)-\tilde Y(0)\mid \tilde X]\right]~.
% \]
% We also define
% \begin{equation} \label{eq:Vobs}
% V^{\rm obs}
% :=
% E_Q\left[
% \frac{\var_Q(\tilde Y(1)\mid \tilde X)}{\eta}
% +
% \frac{\var_Q(\tilde Y(0)\mid \tilde X)}{1-\eta}
% \right].
% \end{equation}

\begin{assumption}\label{ass:super}
For some random vector $\tilde W=(\tilde Y(1),\tilde Y(0),\tilde X)\sim Q$ with
$\max_{d\in\{0,1\}}E_Q[\tilde Y(d)^2]<\infty$, $(W^{(n)}, \Lambda_n)$ satisfy the following requirements:
\begin{enumerate}[(a)]
\item As $n\to\infty$,
\begin{align*}
\frac{1}n\sum_{1 \le i \le n} Y_i(d)^r & \to E_Q[\tilde Y(d)^r] \text{ for }d\in\{0,1\},\ r\in\{1,2\} \\
\frac{1}n\sum_{1 \le i \le n} Y_i(1)Y_i(0) & \to E_Q[\tilde Y(1)\tilde Y(0)]~.
\end{align*}
\item For every $d,d'\in\{0,1\}$,
\[
\frac{1}{m}\sum_{1 \le j \le m}
\frac{1}{k(k-1)}
\sum_{\substack{i,i'\in\lambda_j\\ i\ne i'}}
Y_i(d)Y_{i'}(d')
\to
E_Q \left [ E_Q[\tilde{Y}(d)\mid \tilde{X}]E_Q[\tilde{Y}(d')\mid \tilde X] \right ]~.
\]
\end{enumerate}
\end{assumption}

Assumption \ref{ass:super} imposes deterministic convergence conditions on the fixed finite populations.  Part (a) requires the empirical moments of the potential outcomes to converge to their counterparts under $Q$; these conditions hold by the law of large numbers if $W^{(n)}$ is modeled as an i.i.d.\ sample from $Q$.  Part (b) requires within-stratum outcome products to converge as though units in the same stratum share a common limiting covariate value.  This condition can be verified for standard stratification schemes that group together units with similar covariate values, under the same i.i.d.\ thought experiment. See \cite{bai2022inference}, \cite{bai2024inference}, and \cite{cytrynbaum2021optimal} for related arguments.

Using Assumption \ref{ass:super}, we have the following expression for the limit of $\var[\hat{\Delta}_n]$ after appropriate normalization:
\begin{theorem}\label{thm:limit_V}
Under Assumptions \ref{ass:dn} and \ref{ass:super},
\[n\cdot\var[\hat{\Delta}_n] \rightarrow V~,\]
where
\[V := E_Q\left[\frac{\var_Q[\tilde{Y}(1)|\tilde{X}]}{\eta} + \frac{\var_Q[\tilde{Y}(0)|\tilde{X}]}{1- \eta}\right] - E_Q\left[\var_Q[\tilde{Y}(1) - \tilde{Y}(0)\mid \tilde{X}]\right]~.\]
\end{theorem}
Theorem \ref{thm:limit_V} characterizes the limiting variance of $\hat{\Delta}_n$ when the sequence of finite populations satisfies the additional structure posited in Assumption \ref{ass:super}.  The expression for $V$ is the limiting analogue of the variance decomposition in \eqref{eq:varDeltahat}, with the second component being the average conditional variance of the treatment effect.  In light of the discussion in Remark \ref{rem:better_bounds}, we define the observable asymptotic benchmark for graph-Laplacian estimators as
\[V^{\rm obs}:= E_Q\left[\frac{\var_Q[\tilde{Y}(1)|\tilde{X}]}{\eta} + \frac{\var_Q[\tilde{Y}(0)|\tilde{X}]}{1- \eta}\right]~.\]
To study when graph-Laplacian estimators attain this benchmark, we impose one additional high-level assumption on the sequence of populations and graph weights.
\begin{assumption} \label{ass:limit_graph_products}
For the sequence of graphs $(\omega_m)_{m \geq 1}$, there exists a map $\Gamma_\omega:\mathbb R^{p_X} \rightarrow \mathbb R^{r_\omega}$ such that, for every $d,d'\in\{0,1\}$,
\[
\frac{1}{m}\sum_{1 \le a \le m}\sum_{b\ne a}\omega_{ab,m}
\bar Y_{a,n}(d)\bar Y_{b,n}(d')
\to
E_Q\left[E_Q[\tilde{Y}(d)\mid\Gamma_\omega(\tilde{X})]E_Q[\tilde{Y}(d')\mid\Gamma_\omega(\tilde{X})]\right]~.
\]
\end{assumption}
Assumption \ref{ass:limit_graph_products} is a graph-specific product-limit condition.  It states that average products of stratum-level potential outcomes across graph edges converge to a limit determined by $\Gamma_\omega(\tilde X)$.  The map $\Gamma_\omega$ summarizes the covariate information retained by the graph when it forms cross-stratum comparisons.  Note that Assumption \ref{ass:limit_graph_products} is not necessarily guaranteed to hold for any sequence $(\omega_m)$, but we document the functional form of $\Gamma_\omega$ for the two most relevant cases in Propositions \ref{prop:complete_graph_product}--\ref{prop:local_graph_product} below.

The next theorem establishes the limit of the expectation of the variance estimator:
\begin{theorem}\label{thm:main_var}
Suppose Assumption \ref{ass:dn} holds and Assumptions \ref{ass:super}--\ref{ass:limit_graph_products} hold for a sequence of degree-calibrated graphs $(\omega_m)_{m \geq 1}$.
Then
\[
nE[\hat V_n(\omega_m)]
\to
V_\omega:= V^{\rm obs}+k E_Q\left[\var_Q \big [E[\tilde Y(1) - \tilde Y(0) | \tilde X] \mid \Gamma_\omega(\tilde{X}) \big ]\right]~.
\]
Furthermore, if \eqref{eq:pop_moments}--\eqref{eq:pop_nondegen} hold and $n\{\hat V_n(\omega_m)-E[\hat V_n(\omega_m)]\}\xrightarrow{P}0$, then
\[
P \left\{\hat \Delta_n - \sqrt{\hat V_n(\omega_m)}\cdot z_{1 - \alpha / 2} \le  \Delta_n \le \hat \Delta_n + \sqrt{\hat V_n(\omega_m)}\cdot z_{1 - \alpha / 2} \right\}
\to
2\Phi(z_{1-\alpha/2}/\varsigma_\omega)-1~,
\]
where $\varsigma_\omega=(V/V_\omega)^{1/2}$.
\end{theorem}
Theorem \ref{thm:main_var} demonstrates that $\hat V_n(\omega_m)$ attains the observable upper bound $V^{\rm obs}$ whenever the conditional average treatment effect $E[\tilde Y(1) - \tilde Y(0) | \tilde X]$ is measurable with respect to $\Gamma_\omega(\tilde X)$; in particular, this holds when the graph product limit has $\Gamma_\omega(\tilde X)=\tilde X$.  If the graph retains less covariate information, then the excess term is the residual variation in $E[\tilde Y(1) - \tilde Y(0) | \tilde X]$ left after conditioning on that coarser information.  

We verify in Proposition \ref{prop:complete_graph_product} below that, for the complete graph $\omega_m^{\rm CG}$, $\Gamma_\omega(\tilde X) \equiv 1$ is constant, so the excess term is $k\var_Q[E[\tilde Y(1) - \tilde Y(0) | \tilde X]]$ and vanishes only under homogeneous conditional average treatment effects.

\begin{proposition}\label{prop:complete_graph_product}
Suppose Assumption \ref{ass:super}(a) holds. Then, for $(\omega_m^{\rm CG})_{m \geq 1}$, Assumption \ref{ass:limit_graph_products} holds with $\Gamma_\omega \equiv 1$. If in addition Assumptions \ref{ass:dn} and \ref{ass:super}(b) hold, then $V_{\omega^{\rm CG}} = V^{\rm obs} + k \var_Q[E_Q[\tilde Y(1) - \tilde Y(0) | \tilde X]]$.
\end{proposition}

Next, we present sufficient conditions under which Assumption \ref{ass:limit_graph_products} holds with $\Gamma_\omega(x) = x$, so that $V_\omega = V^{\rm obs}$ in Theorem \ref{thm:main_var}. To begin, we require that a suitable decomposition of the potential outcomes exists:

\begin{assumption} \label{ass:lipschitz}
For $d\in\{0,1\}$ and $1 \leq i \leq n$, there exists $\mu_d(X_i)$ and $u_i(d)$ such that $Y_i(d)=\mu_d(X_i)+u_i(d)$ and
\begin{enumerate}[(a)]
\item $\mu_d:\mathbb R^{p_X}\to\mathbb R$ is Lipschitz.
\item For every $d,d'\in\{0,1\}$,
\[ \frac{1}n\sum_{1 \le i \le n}\mu_d(X_i)\mu_{d'}(X_i)\to E_Q \left [ E_Q[\tilde{Y}(d)\mid \tilde{X}]E_Q[\tilde{Y}(d')\mid \tilde{X}] \right ]~, \]
\item For $\bar \mu_{j,n}(d)=k^{-1}\sum_{i\in\lambda_j}\mu_d(X_i)$ and $\bar u_{j,n}(d)=k^{-1}\sum_{i\in\lambda_j}u_i(d)$,
\begin{align*}
\frac{1}{m}\sum_{1 \le a \le m}\sum_{b\ne a}\omega_{ab,m}\bar u_{a,n}(d)\bar \mu_{b,n}(d') & \to 0 \\
\frac{1}{m}\sum_{1 \le a \le m}\sum_{b\ne a}\omega_{ab,m}\bar \mu_{a,n}(d)\bar u_{b,n}(d') & \to 0 \\
\frac{1}{m}\sum_{1 \le a \le m}\sum_{b\ne a}\omega_{ab,m}\bar u_{a,n}(d)\bar u_{b,n}(d') & \to 0~.
\end{align*}
\end{enumerate} 
\end{assumption}

The Lipschitz condition in Assumption \ref{ass:lipschitz}(a) ensures that, along graph edges with small covariate distance, the smooth components of $Y_i(d)$ are close. Assumption \ref{ass:lipschitz}(b) further guarantees that the average products of the smooth components converge to the desired limit. The residual-product conditions in Assumption \ref{ass:lipschitz}(c) require the remaining graph-weighted products involving $u_i(d)$ to be asymptotically negligible.  Under an i.i.d.\ thought experiment with $\mu_d(\tilde X)=E_Q[\tilde Y(d)\mid\tilde X]$, these residual conditions can be verified using arguments similar to those in \cite{bai2024inference}.

We also require a graph-locality condition, defined in terms of the following graph-weighted average distance between covariates:
\begin{equation}\label{eq:graph_edge_diameter}
\mathcal D_n(\omega_m)
=
\frac{1}{m}
\sum_{1 \le a \le m}\sum_{b\ne a}\omega_{ab,m}
\max_{i,i'\in\lambda_a\cup\lambda_b}\|X_i-X_{i'}\|^2~.
\end{equation}

\begin{assumption} \label{ass:local}
The sequence of graphs $(\omega_m)_{m \geq 1}$ satisfies $\mathcal D_n(\omega_m)\to 0$.    
\end{assumption}

Assumption \ref{ass:local} requires the graph-weighted average squared distance across adjacent strata to vanish. In Proposition \ref{prop:graph_locality} below, we will provide sufficient conditions for Assumption \ref{ass:local} for $\omega_m^{\rm PM}$ in \eqref{eq:perfect_matching} and $\omega_m^{\rm SDP}(\kappa_m)$ in \eqref{eq:sdp}.

Assumptions \ref{ass:lipschitz} and \ref{ass:local} together ensure Assumption \ref{ass:limit_graph_products} holds with $\Gamma_\omega(x) = x$.

\begin{proposition}\label{prop:local_graph_product}
Suppose the potential outcomes $Y_i(d), 1 \leq i \leq n, d \in \{0, 1\}$ satisfy Assumption \ref{ass:lipschitz}. Further suppose $(\omega_m)_{m \geq 1}$ is a sequence of degree-calibrated graphs such that Assumption \ref{ass:local} holds. Then, Assumption \ref{ass:limit_graph_products} holds for $\Gamma_\omega(x) = x$.
\end{proposition}

To conclude the discussion, we provide sufficient conditions for Assumption \ref{ass:local} to hold for the graphs $\omega_m^{\rm PM}$ and $\omega_m^{\rm SDP}(\kappa_m)$ and use these conditions to derive a parsimonious choice of $\kappa_m$ to use in practice. To describe the conditions, define the coordinate-range for the stratum means as
\begin{equation}\label{eq:mean_range}
M_{n}^2
:=
\max_{1\le h\le p_X}
\left(
\max_{1\le j\le m}\bar X_{j,n}^{(h)}
-
\min_{1\le j\le m}\bar X_{j,n}^{(h)}
\right)^2~,
\end{equation}
where $\bar X_{j,n}^{(h)}$ is the $h$th component of $\bar X_{j,n}$. Further define
\begin{equation}\label{eq:within_stratum_radius}
R_n =\frac{1}{m}\sum_{1 \le j \le m} \max_{i\in\lambda_j}\|X_i-\bar X_{j,n}\|^2~.
\end{equation}

\begin{proposition}\label{prop:graph_locality}
Assume $R_n \to 0$. Then, the following statements hold:
\begin{enumerate}[\rm (a)]
\item Suppose $m$ is even. Let $\omega_m^{\rm PM}$ be given by \eqref{eq:perfect_matching} with $c_{ab, m} = \|\bar X_{a,n} - \bar X_{b,n}\|^2$. Also suppose the stratum means satisfy the range condition
\begin{equation}\label{eq:pm_range_condition}
M_{n}^2m^{-2/(p_X+1)}\to0~.
\end{equation}
Then, Assumption \ref{ass:local} holds for $(\omega_m^{\rm PM})_{m \geq 1}$.
\item Let $\omega_m^{\rm SDP}(\kappa_m)$ be given by \eqref{eq:sdp} with $c_{ab, m} = \|\bar X_{a,n} - \bar X_{b,n}\|^2$. Suppose there is an integer sequence $(b_m)_{m \geq 1}$ such that
\begin{equation}\label{eq:block_size_tuning}
b_m\to\infty~,
\qquad
b_m=o(m)~,
\qquad
\frac{b_m}{b_m-1}\le \kappa_m~,
\end{equation}
and the stratum means satisfy the range condition
\begin{equation}\label{eq:lp_range_condition}
M_{n}^2\left(\frac{m}{b_m}\right)^{-2/(p_X+1)}
\to0~.
\end{equation}
Then, Assumption \ref{ass:local} holds for $(\omega_m^{\rm SDP}(\kappa_m))_{m \geq 1}$.
\end{enumerate}
\end{proposition}
The assumption $R_n \rightarrow 0$ ensures that the original stratification $\Lambda_n$ is sufficiently ``local,'' in the sense that units within the same stratum have similar covariate values. Analogous sufficient conditions to \eqref{eq:pm_range_condition} for this kind of stratum-level locality are documented in \cite{bai2022inference} and \cite{cytrynbaum2021optimal}. The range condition \eqref{eq:pm_range_condition} in Proposition \ref{prop:graph_locality}(a) is automatically satisfied if the stratum means have uniformly bounded coordinate range, so that $M_n=O(1)$.  More generally, if $M_{n}^2=O(m^{2/\zeta})$ for some $\zeta>p_X+1$, then \eqref{eq:pm_range_condition} holds. Therefore, the minimum-cost perfect matching graph is local whenever the original strata are internally ``tight'' and the stratum means do not spread out too quickly.

Proposition \ref{prop:graph_locality}(b) provides the corresponding statement for $\omega_m^{\rm SDP}$. Recall from Theorem \ref{thm:minimax} that the maximum normalized bias of $\hat{V}_n(\omega_m^{\rm SDP}(\kappa_m))$ is governed directly by $\kappa_m$, and if $\kappa_m \to 1$ as $m \to \infty$, then it approaches the minimax normalized-bias bound.  Proposition \ref{prop:graph_locality}(b) shows that $\kappa_m \to 1$ is compatible with the graph-locality condition, as long as a suitable $b_m$ can be chosen to satisfy \eqref{eq:lp_range_condition} while $b_m/(b_m - 1) \to 1$.  If $M_n=O(1)$, then \eqref{eq:lp_range_condition} holds for any choice of $b_m$ satisfying \eqref{eq:block_size_tuning}.  More generally, if $M_n^2=O(m^{2/\zeta})$ for $\zeta>p_X+1$, then \eqref{eq:lp_range_condition} requires $b_m=o(m^{1-(p_X+1)/\zeta})$; if we take $b_m=\lceil m^{\gamma}\rceil+1$, this requires $0<\gamma<1-\frac{p_X+1}{\zeta}$.  For this choice,
\[
\frac{b_m}{b_m-1}
=
1+\frac{1}{\lceil m^\gamma\rceil}
\le
1+m^{-\gamma}~.
\]

We conclude this section by suggesting a heuristic method for choosing $\gamma$ in practice. For $0 < \gamma < 1$, define
\[
\kappa_m(\gamma)
=
\max\left\{\frac{m}{m-1},1+m^{-\gamma}\right\}~.
\]
Suppose that $M_n=O(1)$. Under the preceding choice of $b_m$, the left-hand side of \eqref{eq:lp_range_condition} is
\[
M_{n}^2\left(\frac{m}{b_m}\right)^{-2/(p_X+1)} = O(m^{-2(1-\gamma)/(p_X+1)})~,
\]
and therefore converges to zero.
Thus, larger values of $\gamma$ send $\kappa_m(\gamma)$ to one more quickly, but slow down the sufficient rate of convergence for \eqref{eq:lp_range_condition} required in Proposition \ref{prop:graph_locality}. Equating the two exponents leads to a rate-balanced choice of $\gamma$:
\begin{equation}\label{eq:gamma_rate_balance}
\bar\gamma(p_X)
=
\frac{2}{p_X+3}~.
\end{equation}
Although this calculation provides a parsimonious choice for $\gamma$, it relies formally on the assumption that the coordinate-range of the covariates is bounded. In the next section, we assess how sensitive the performance of our recommendation is to this assumption via simulation.

\begin{remark}\label{rem:fogarty}
\cite{fogarty2018mitigating} proposes an alternative generalization of the Imai estimator that uses regression adjustment of the stratum-level treatment-effect estimates on baseline covariates. Although his estimator is not a graph-Laplacian estimator in the sense of Section \ref{sec:main_var}, it can be analyzed under a limiting framework similar to that developed in this section. In Appendix \ref{app:fogarty}, we show that the limit of its expectation is
\[
V^{\rm obs}
+
k\var_Q\left[
E[\tilde Y(1) - \tilde Y(0) | \tilde X]-\blp_Q(\tilde Y(1) - \tilde Y(0) | 1,\tilde X)
\right]~,
\]
where $\blp_Q(\cdot \mid 1, \tilde{X})$ denotes the best linear predictor on a constant and $\tilde{X}$, and hence coincides with $V^{\rm obs}$ only when $E[\tilde Y(1) - \tilde Y(0) | \tilde X]$ is linear in $\tilde X$.
\end{remark}

\section{Simulations} \label{sec:sims}
In this section, we conduct a design-based simulation to illustrate the main findings in the previous sections and to assess the parsimonious regularization constant $\kappa_m(\bar{\gamma})$ proposed in \eqref{eq:gamma_rate_balance}, in multivariate and heavy-tailed designs. We select the treatment effects to interpolate between a design that is favorable for the cost-minimizing perfect matching in \eqref{eq:perfect_matching} and a design that is adversarial against it. We set $k=2$, so the strata are pairs and $n=2m$. For each design, we first draw a ``master'' population with $m_{\max}=250$ pair strata. To form nested finite sub-populations, we randomly order the master strata and, for each reported value of $n$, keep the first $m=n/2$ strata in that ordering. The reported population sizes are $n\in\{100,252,500\}$. For each population size and design, the covariates, potential outcomes, and graph weights are held fixed in simulations, and the only source of randomness is from treatment assignment within pairs.

We consider three covariate designs. In Model A, we draw a scalar stratum center i.i.d.\ $\bar X_j^{\max}\sim U[0,1]$ for each $1\le j\le m_{\max}$. Model B uses a three-dimensional correlated Gaussian covariate vector, with stratum centers drawn i.i.d.\ according to
\[
\bar X_j^{\max}
\sim
N_3\left(
\begin{pmatrix}
.5\\
.5\\
.5
\end{pmatrix},
\frac{1}{12}
\begin{pmatrix}
1 & .5 & .25\\
.5 & 1 & .5\\
.25 & .5 & 1
\end{pmatrix}
\right),
\qquad
1\le j\le m_{\max}~.
\]
Thus, each coordinate has the same marginal mean and variance as $U[0,1]$. Model C features a heavy-tailed design: using the same uniform draws as in Model A, we construct $\bar X_j^{\max}=.5+T_{3j}/6$, where $T_{3j}$ is obtained by applying the $t_3$ quantile transformation to the corresponding uniform draw, and the scaling ensures the stratum center has the same variance as $U[0,1]$.

For Models A and C, we draw a stratum-shock $\eta_j$ independently across strata from a normal distribution with mean zero and standard deviation $.10/m_{\max}$. For Model B, we instead draw
\[
\eta_j
\sim
N_3\left(
\begin{pmatrix}
0\\
0\\
0
\end{pmatrix},
\left(\frac{.10}{m_{\max}}\right)^2
\begin{pmatrix}
1 & .5 & .25\\
.5 & 1 & .5\\
.25 & .5 & 1
\end{pmatrix}
\right)~.
\]
In every design, we set the two unit-level covariates in stratum $j$ so that
\[
\{X_i:i\in\lambda_j\}
=
\{\bar X_{j}^{\max}+\eta_j,\bar X_{j}^{\max}-\eta_j\}~.
\]
This construction makes $\bar X_j^{\max}$ the exact stratum-level covariate mean. 

To define $Y_i(0)$, we first define the scalar score
\[
s(x)
=
\begin{cases}
x-.5,
& \text{if } p_X=1~,\\[3pt]
\displaystyle
\frac{x_1+x_2+x_3-1.5}
{\sqrt{11/2}},
& \text{if } p_X=3~.
\end{cases}
\]
The variance of $s(\bar X_j^{\max})$ equals $1/12$ in every design. The untreated potential outcomes are then generated according to
\[
Y_i(0)=1.5s(X_i)+\varepsilon_i~,
\qquad
\varepsilon_i\sim N(0,.35^2)~,
\]
with the shocks drawn independently across units. The potential-outcome shocks and the ordering used to form the nested sub-populations are shared across all three designs. For any vector $u=(u_1,\ldots,u_m)$, let $\operatorname{std}_m(u)$ denote the vector with $j$th component
\[
\operatorname{std}_m(u)_j
=
\frac{u_j-\bar u}{s_u}
\text{ for }
\bar u=\frac{1}{m}\sum_{1 \le j \le m} u_j
\text{ and }
s_u^2=\frac{1}{m-1}\sum_{1 \le j \le m}(u_j-\bar u)^2~.
\]

As mentioned above, the treatment effects are chosen to interpolate between a ``smooth'' regime that is favorable to perfect matching and one that is adversarial against perfect matching. Writing $s(\bar X^{\max})$ for the vector of master-population scores, we specify the ``smooth'' regime favorable to perfect matching using a vector $\theta = (\theta_1, \ldots , \theta_m)'$, where
\[
\theta_j
:=
\operatorname{std}_{m_{\max}}\{s(\bar X^{\max})\}_j~.
\]
Each nested sub-population inherits the corresponding unit-level covariates, untreated potential outcomes, stratum centers, and values of $\theta_j$ from its master population.

To define the adversarial regime, for each nested population containing $m$ strata, and hence $n=2m$ units, we construct the minimum-cost perfect-matching graph using the raw squared Euclidean costs
\[
c_{ab,m}
=
\|\bar X_a^{\max}-\bar X_b^{\max}\|^2,
\qquad
a\ne b~.
\]
Write the pairs in this matching as $\{a_q,b_q\}$, $1\le q\le m/2$, ordered so that $s(\bar X_{a_q}^{\max})\leq s(\bar X_{b_q}^{\max})$. Define the $m$-dimensional vector $v^{\rm PM}$ with entries
\[
v_{a_q}^{\rm PM}=-\frac{1}{\sqrt m}
\text{ and}\qquad
v_{b_q}^{\rm PM}=\frac{1}{\sqrt m}
\]
for each pair. This vector is an eigenvector associated with the largest eigenvalue of $L_m(\omega_m^{\rm PM})$; as shown in the proof of Theorem \ref{thm:minimax}, normalized bias is maximized when the centered treatment-effect vector is an eigenvector associated with the largest eigenvalue of the Laplacian. We therefore specify the adversarial regime by standardizing $v^{\rm PM}$ within the sub-population:
\[
\xi_j :=\operatorname{std}_m(v^{\rm PM})_j~.
\]
For $\nu\in[0,1]$, the treated potential outcomes are
\[
Y_i(1)=Y_i(0)+\Delta_j(\nu),
\qquad
i\in\lambda_j~,
\]
where
\[
\Delta_j(\nu)
:=
3\operatorname{std}_{m}\{r(\nu)\}_j~,
\]
with $r(\nu) = (r_1(\nu), \ldots, r_m(\nu))'$ and
\[
r_j(\nu):=\nu\theta_j+(1-\nu)\xi_j~.
\]
Thus, $\nu=1$ corresponds to the favorable smooth regime and $\nu=0$ corresponds to the adversarial regime. The scale factor keeps the sample standard deviation of $\Delta_j(\nu)$ equal to three for every value of $\nu$. We consider $\nu\in\{0,0.5,1\}$.

We compare the complete graph $\omega_m^{\rm CG}$, the minimum-cost perfect-matching graph $\omega_m^{\rm PM}$, and three spectrally constrained graphs $\omega_m^{\rm SDP}(\kappa_m(\gamma))$ that solve the regularized graph problem in Section \ref{sec:minimax} for three different choices of $\gamma$ and
\[
\kappa_m(\gamma)
=
\max\left\{
\frac{m}{m-1},
1+m^{-\gamma}
\right\}~.
\]
The choices of $\gamma$ include a weakly regularized comparison $\gamma=0.1$, a deliberately highly regularized comparison $\gamma=0.75$, and the parsimonious choice
\[
\bar\gamma(p_X)=\frac{2}{p_X+3}~,
\]
which equals $0.5$ for Models A and C and $\frac{1}{3}$ for Model B. For Models A and B, almost every resulting fixed-population sequence satisfies $M_n^2 = O(1)$ and $M_n^2=O(\log m)$, respectively. To guarantee the range condition \eqref{eq:lp_range_condition} in Proposition \ref{prop:graph_locality}(b), we therefore need
\[
 M_n^2m^{-2(1-\gamma)/(p_X+1)}\to0
\]
for every choice of $\gamma < 1$. For the $t_3$ design, by contrast, almost every resulting fixed-population sequence satisfies $M_n^2=O(m^{\frac{2}{3}+\epsilon})$ for every $\epsilon>0$, and therefore we need $\gamma<\frac{1}{3}$.

For each finite population and each reported value of $\nu$, we draw 5000 independent treatment assignments, assigning exactly one unit to treatment in each stratum. For each assignment, we compute the usual difference-in-means estimator $\hat\Delta_n$ and, for each graph, the variance estimator $\hat V_n(\omega_m)$ in \eqref{eq:graph_estimator}. Table \ref{tab:simulation_graph_adversarial} reports the coverage and average length of nominal $95\%$ confidence intervals
\[
\hat\Delta_n
\pm
z_{0.975}\sqrt{\hat V_n(\omega_m)}~.
\]

\begin{table}[ht]
    \centering
    \footnotesize
    \begin{tabular}{cccccccc}
\toprule
        \textbf{Population size} & \textbf{Estimator} &
        \multicolumn{2}{c}{$\nu=0$} &
        \multicolumn{2}{c}{$\nu=.5$} &
        \multicolumn{2}{c}{$\nu=1$} \\
        \cmidrule(lr){3-4} \cmidrule(lr){5-6} \cmidrule(lr){7-8}
        & & Coverage & Length & Coverage & Length & Coverage & Length \\
        \hline
        \multicolumn{8}{l}{\textit{Panel A: $U[0,1]$ covariates}} \\
\multirow{5}{*}{100} & PM & 1.000 & 2.344 & 1.000 & 1.707 & 0.960 & 0.303 \\
 & SDP, $\gamma=.1$ & 1.000 & 2.086 & 1.000 & 1.472 & 0.961 & 0.300 \\
 & SDP, parsimonious & 1.000 & 1.774 & 1.000 & 1.269 & 0.991 & 0.385 \\
 & SDP, $\gamma=.75$ & 1.000 & 1.710 & 1.000 & 1.323 & 1.000 & 0.780 \\
 & CG & 1.000 & 1.686 & 1.000 & 1.687 & 1.000 & 1.686 \\
\hline
\multirow{5}{*}{252} & PM & 1.000 & 1.487 & 1.000 & 1.074 & 0.948 & 0.181 \\
 & SDP, $\gamma=.1$ & 1.000 & 1.271 & 1.000 & 0.910 & 0.948 & 0.181 \\
 & SDP, parsimonious & 1.000 & 1.103 & 1.000 & 0.795 & 0.977 & 0.206 \\
 & SDP, $\gamma=.75$ & 1.000 & 1.073 & 1.000 & 0.804 & 1.000 & 0.369 \\
 & CG & 1.000 & 1.063 & 1.000 & 1.063 & 1.000 & 1.063 \\
\hline
\multirow{5}{*}{500} & PM & 1.000 & 1.057 & 1.000 & 0.755 & 0.944 & 0.128 \\
 & SDP, $\gamma=.1$ & 1.000 & 0.892 & 1.000 & 0.635 & 0.945 & 0.128 \\
 & SDP, parsimonious & 1.000 & 0.776 & 1.000 & 0.555 & 0.963 & 0.136 \\
 & SDP, $\gamma=.75$ & 1.000 & 0.759 & 1.000 & 0.558 & 0.999 & 0.221 \\
 & CG & 1.000 & 0.755 & 1.000 & 0.754 & 1.000 & 0.755 \\
\midrule
        \multicolumn{8}{l}{\textit{Panel B: correlated $N_3$ covariates}} \\
\multirow{5}{*}{100} & PM & 1.000 & 2.345 & 1.000 & 1.813 & 1.000 & 0.512 \\
 & SDP, $\gamma=.1$ & 1.000 & 2.036 & 1.000 & 1.526 & 1.000 & 0.534 \\
 & SDP, parsimonious & 1.000 & 1.797 & 1.000 & 1.318 & 1.000 & 0.572 \\
 & SDP, $\gamma=.75$ & 1.000 & 1.706 & 1.000 & 1.428 & 1.000 & 1.103 \\
 & CG & 1.000 & 1.687 & 1.000 & 1.687 & 1.000 & 1.686 \\
\hline
\multirow{5}{*}{252} & PM & 1.000 & 1.487 & 1.000 & 1.126 & 0.997 & 0.269 \\
 & SDP, $\gamma=.1$ & 1.000 & 1.261 & 1.000 & 0.930 & 0.999 & 0.290 \\
 & SDP, parsimonious & 1.000 & 1.132 & 1.000 & 0.837 & 1.000 & 0.345 \\
 & SDP, $\gamma=.75$ & 1.000 & 1.070 & 1.000 & 0.864 & 1.000 & 0.594 \\
 & CG & 1.000 & 1.063 & 1.000 & 1.063 & 1.000 & 1.063 \\
\hline
\multirow{5}{*}{500} & PM & 1.000 & 1.058 & 1.000 & 0.781 & 0.985 & 0.163 \\
 & SDP, $\gamma=.1$ & 1.000 & 0.886 & 1.000 & 0.636 & 0.988 & 0.167 \\
 & SDP, parsimonious & 1.000 & 0.798 & 1.000 & 0.580 & 0.997 & 0.207 \\
 & SDP, $\gamma=.75$ & 1.000 & 0.758 & 1.000 & 0.594 & 1.000 & 0.382 \\
 & CG & 1.000 & 0.755 & 1.000 & 0.755 & 1.000 & 0.754 \\
\midrule
        \multicolumn{8}{l}{\textit{Panel C: $t_3$ covariates}} \\
\multirow{5}{*}{100} & PM & 1.000 & 2.344 & 1.000 & 1.794 & 0.998 & 0.490 \\
 & SDP, $\gamma=.1$ & 1.000 & 2.115 & 1.000 & 1.577 & 0.999 & 0.505 \\
 & SDP, parsimonious & 1.000 & 1.777 & 1.000 & 1.395 & 1.000 & 0.713 \\
 & SDP, $\gamma=.75$ & 1.000 & 1.710 & 1.000 & 1.456 & 1.000 & 1.090 \\
 & CG & 1.000 & 1.686 & 1.000 & 1.687 & 1.000 & 1.686 \\
\hline
\multirow{5}{*}{252} & PM & 1.000 & 1.487 & 1.000 & 1.119 & 0.986 & 0.224 \\
 & SDP, $\gamma=.1$ & 1.000 & 1.272 & 1.000 & 0.944 & 0.988 & 0.230 \\
 & SDP, parsimonious & 1.000 & 1.103 & 1.000 & 0.849 & 1.000 & 0.361 \\
 & SDP, $\gamma=.75$ & 1.000 & 1.073 & 1.000 & 0.887 & 1.000 & 0.592 \\
 & CG & 1.000 & 1.063 & 1.000 & 1.063 & 1.000 & 1.063 \\
\hline
\multirow{5}{*}{500} & PM & 1.000 & 1.057 & 1.000 & 0.766 & 0.994 & 0.181 \\
 & SDP, $\gamma=.1$ & 1.000 & 0.890 & 1.000 & 0.643 & 0.996 & 0.198 \\
 & SDP, parsimonious & 1.000 & 0.776 & 1.000 & 0.583 & 1.000 & 0.280 \\
 & SDP, $\gamma=.75$ & 1.000 & 0.759 & 1.000 & 0.614 & 1.000 & 0.417 \\
 & CG & 1.000 & 0.755 & 1.000 & 0.754 & 1.000 & 0.755 \\
\bottomrule
    \end{tabular}
    \caption{Graph-adversarial mixture simulation. Panel A uses bounded $U[0,1]$ stratum centers, Panel B uses correlated three-dimensional Gaussian stratum centers, and Panel C uses scaled $t_3$ stratum centers. The parsimonious row uses $\gamma=.5$ in Panels A and C and $\gamma=1/3$ in Panel B. Treatment effects interpolate between a smooth function of the stratum centers $\bar X_j$ and the leading Laplacian eigenvector of the PM graph. The entries are coverage probabilities and average lengths for nominal $95\%$ intervals across randomized assignments.}
    \label{tab:simulation_graph_adversarial}
\end{table}

Table \ref{tab:simulation_graph_adversarial} illustrates the intended robustness--locality tradeoff in all three covariate designs. When $\nu=0$, the treatment-effect vector is chosen adversarially to the perfect-matching graph. The confidence intervals arising from the perfect-matching graph are then the longest, while stronger regularization moves the SDP interval lengths toward that of the complete graph interval. When $\nu=1$, treatment effects are smooth in the baseline score, and the ranking of confidence interval lengths is reversed. At the intermediate value $\nu=.5$, the dimension-specific parsimonious choice produces the shortest intervals in all three designs.

\begin{table}[ht]
    \centering
    \footnotesize
    \begin{tabular}{ccccc}
\toprule
        \textbf{Population size} & \textbf{Estimator} &
        $\mathcal D_n(\omega)$ & $\lambda_{\max}\{L_m(\omega)\}$ & $\max \omega_{ab}$ \\
        \hline
        \multicolumn{5}{l}{\textit{Panel A: $U[0,1]$ covariates}} \\
\multirow{5}{*}{100} & PM & 0.001 & 2.000 & 1.000 \\
 & SDP, $\gamma=.1$ & 0.001 & 1.676 & 0.676 \\
 & SDP, parsimonious & 0.004 & 1.141 & 0.141 \\
 & SDP, $\gamma=.75$ & 0.033 & 1.053 & 0.053 \\
 & CG & 0.171 & 1.020 & 0.020 \\
\hline
\multirow{5}{*}{252} & PM & 0.000 & 2.000 & 1.000 \\
 & SDP, $\gamma=.1$ & 0.000 & 1.617 & 0.617 \\
 & SDP, parsimonious & 0.002 & 1.089 & 0.089 \\
 & SDP, $\gamma=.75$ & 0.016 & 1.027 & 0.027 \\
 & CG & 0.168 & 1.008 & 0.008 \\
\hline
\multirow{5}{*}{500} & PM & 0.000 & 2.000 & 1.000 \\
 & SDP, $\gamma=.1$ & 0.000 & 1.576 & 0.576 \\
 & SDP, parsimonious & 0.001 & 1.063 & 0.063 \\
 & SDP, $\gamma=.75$ & 0.010 & 1.016 & 0.016 \\
 & CG & 0.174 & 1.004 & 0.004 \\
\midrule
        \multicolumn{5}{l}{\textit{Panel B: correlated $N_3$ covariates}} \\
\multirow{5}{*}{100} & PM & 0.050 & 2.000 & 1.000 \\
 & SDP, $\gamma=.1$ & 0.060 & 1.676 & 0.676 \\
 & SDP, parsimonious & 0.092 & 1.272 & 0.271 \\
 & SDP, $\gamma=.75$ & 0.263 & 1.053 & 0.053 \\
 & CG & 0.499 & 1.020 & 0.020 \\
\hline
\multirow{5}{*}{252} & PM & 0.029 & 2.000 & 1.000 \\
 & SDP, $\gamma=.1$ & 0.034 & 1.617 & 0.617 \\
 & SDP, parsimonious & 0.063 & 1.200 & 0.199 \\
 & SDP, $\gamma=.75$ & 0.219 & 1.027 & 0.027 \\
 & CG & 0.495 & 1.008 & 0.008 \\
\hline
\multirow{5}{*}{500} & PM & 0.018 & 2.000 & 1.000 \\
 & SDP, $\gamma=.1$ & 0.022 & 1.576 & 0.576 \\
 & SDP, parsimonious & 0.046 & 1.159 & 0.159 \\
 & SDP, $\gamma=.75$ & 0.196 & 1.016 & 0.016 \\
 & CG & 0.505 & 1.004 & 0.004 \\
\midrule
        \multicolumn{5}{l}{\textit{Panel C: $t_3$ covariates}} \\
\multirow{5}{*}{100} & PM & 0.007 & 2.000 & 1.000 \\
 & SDP, $\gamma=.1$ & 0.008 & 1.676 & 0.676 \\
 & SDP, parsimonious & 0.019 & 1.141 & 0.141 \\
 & SDP, $\gamma=.75$ & 0.050 & 1.053 & 0.053 \\
 & CG & 0.124 & 1.020 & 0.020 \\
\hline
\multirow{5}{*}{252} & PM & 0.002 & 2.000 & 1.000 \\
 & SDP, $\gamma=.1$ & 0.002 & 1.617 & 0.617 \\
 & SDP, parsimonious & 0.011 & 1.089 & 0.089 \\
 & SDP, $\gamma=.75$ & 0.034 & 1.027 & 0.027 \\
 & CG & 0.118 & 1.008 & 0.008 \\
\hline
\multirow{5}{*}{500} & PM & 0.004 & 2.000 & 1.000 \\
 & SDP, $\gamma=.1$ & 0.006 & 1.576 & 0.576 \\
 & SDP, parsimonious & 0.016 & 1.063 & 0.063 \\
 & SDP, $\gamma=.75$ & 0.040 & 1.016 & 0.016 \\
 & CG & 0.140 & 1.004 & 0.004 \\
\bottomrule
    \end{tabular}
    \caption{Graph diagnostics for the graph-adversarial mixture simulation. Panel A uses bounded $U[0,1]$ stratum centers, Panel B uses correlated three-dimensional Gaussian stratum centers, and Panel C uses scaled $t_3$ stratum centers. Because $\mathcal D_n(\omega)$ uses raw squared Euclidean distance, its level is dimension-dependent. Locality is measured by $\mathcal D_n(\omega)$; maximum normalized bias is proportional to the largest eigenvalue of $L_m(\omega)$.}
    \label{tab:simulation_graph_diagnostics}
\end{table}

Table \ref{tab:simulation_graph_diagnostics} compares the spectral and locality properties of each graph. The perfect-matching graph has the largest Laplacian eigenvalue, while the complete graph has the smallest. Among the regularized graphs, both the largest Laplacian eigenvalue and the maximum edge weight decrease as the spectral constraint tightens across all three models. In contrast, the locality measures depend directly on the covariate distribution. Under the uniform and Gaussian designs, Proposition \ref{prop:graph_locality} establishes locality for all three fixed-$\gamma$ regularized graph sequences. Under the $t_3$ design, however, the sufficient condition requires $\gamma<\frac{1}{3}$, and therefore covers $\gamma=0.1$ but not the parsimonious choice $\gamma=0.5$ or $\gamma=0.75$. Correspondingly, the locality measure computed for the parsimonious choice does not decline steadily across the three population sizes in Panel C. We emphasize, however, that since all the graphs are constructed to be degree calibrated, they are guaranteed to be finite-population upward-biased and thus deliver asymptotically valid inference via Theorem \ref{thm:upward}.

\section{Recommendations for Empirical Practice}\label{sec:recommendations}
Based on our theoretical results as well as the simulation study above, we conclude with some recommendations for empirical practice. For estimating $\var[\hat{\Delta}_n]$ in a finely stratified experiment, a natural starting point is the graph-Laplacian estimator using a minimum-cost perfect matching (or its fractional relaxation) based on stratum-level covariate summaries such as $\bar X_{j,n}=k^{-1}\sum_{i\in\lambda_j}X_i$.  In the well-behaved smooth regime formalized in Section \ref{sec:variances}, this estimator attains the observable asymptotic benchmark $V^{\rm obs}$ in large samples. 

At the same time, the minimax results in Section \ref{sec:minimax} demonstrate that a perfect matching is a strong commitment to this smooth regime and the regularized fractional matching formulation offers a way to temper this commitment;  by imposing a constraint on the graph Laplacian, it controls the worst-case normalized bias while preserving much of the local structure of a perfect matching. Therefore when sensitivity to worst-case heterogeneity is a concern, even a small amount of regularization can provide a transparent compromise.

We further recommend reporting diagnostics such as the locality measure $\mathcal D_n(\omega)$ and the largest eigenvalue of $L_m(\omega)$. The locality measure records the extent to which the graph compares nearby strata in terms of their baseline covariates, which helps assess the quality of matches in the smooth regime. On the other hand, the largest eigenvalue of $L_m(\omega)$ diagnoses worst-case normalized bias outside of this regime.

\clearpage
\appendix

\section{Proofs of Main Results}
In the appendix, we use $\mathrm{tr}(A)$ to denote the trace of a square matrix $A$ and $a_n \lesssim b_n$ to denote that there exists $c > 0$ such that $a_n \leq c b_n$ for all $n \geq 1$.

\subsection{Proof of Theorem \ref{thm:upward}}
First, we show that the assumptions in \eqref{eq:pop_moments} and \eqref{eq:pop_nondegen} imply the conditions in Theorem \ref{thm:normal}, so that 
\begin{equation}\label{eq:normal}
\frac{\hat \Delta_n - \Delta_n}{\var[\hat \Delta_n]^{1/2}} \stackrel{d}{\rightarrow} N(0,1)~.
\end{equation}
To that end, we first show 
\begin{equation}\label{eq:num_zero}
\frac{1}{n}\max_{1 \leq j \leq m} \max_{i \in \lambda_j} (R_i - \bar R_j)^2 \rightarrow 0
\end{equation}
for
\begin{align}
\label{eq:Ri} R_i & := \frac{Y_i(1)}{\eta} + \frac{Y_i(0)}{1 - \eta} \\
\label{eq:Rbar} \bar R_j & := \frac{1}{k} \sum_{i \in \lambda_j} R_i~,
\end{align}
and then show that $\frac{1}{n}\sum_{1 \leq j \leq m} \sum_{i \in \lambda_j} (R_i - \bar R_j)^2$ is bounded away from zero. Note that from the inequality $(a + b)^2 \le 2(a^2 + b^2)$, 
\[\max_{i \in \lambda_j}(R_i - \bar{R}_j)^2 \lesssim \max_{i \in \lambda_j}\frac{1}{\eta^2}(Y_i(1) - \bar Y_{j, n}(1))^2 + \frac{1}{(1 - \eta)^2}\max_{i \in \lambda_j}(Y_i(0) - \bar Y_{j, n}(0))^2~.\]
Since by the triangle inequality $\max_{i \in \lambda_j}|Y_i(d) - \bar Y_{j, n}(d)| \le 2\max_{i \in \lambda_j}|Y_i(d)|$, 
\begin{align*}
\frac{1}{n}\max_{1 \leq j \leq m}\max_{i \in \lambda_j}(Y_i(d) - \bar Y_{j, n}(d))^2 &\lesssim \frac{1}{n}\max_{1 \leq i \leq n}Y_i(d)^2 \rightarrow 0~,
\end{align*}
where the final convergence follows from \eqref{eq:pop_moments}. The convergence in \eqref{eq:num_zero} therefore follows. Next, since it is shown in the proof of Theorem \ref{thm:normal} that $\frac{1}{n}\sum_{1 \leq j \leq m} \sum_{i \in \lambda_j} (R_i - \bar R_j)^2$ is proportional to $n \var[\hat{\Delta}_n]$, we immediately obtain from \eqref{eq:pop_nondegen} that $\frac{1}{n}\sum_{1 \leq j \leq m} \sum_{i \in \lambda_j} (R_i - \bar R_j)^2$ is bounded away from zero. We thus obtain \eqref{eq:highlevel} in Theorem \ref{thm:normal} and so \eqref{eq:normal} follows. To complete the proof, we apply Lemma \ref{lem:ratio} with $X_n = \sqrt n(\hat \Delta_n - \Delta_n)$, $\sigma_n^2 = n \var[\hat \Delta_n]$, and $\hat \sigma_n^2 = n \tilde V_n$. \qed

\subsection{Proof of Theorem \ref{thm:main}}
For $1 \leq i \leq n$, Assumption \ref{ass:dn} implies $E[D_i]=\ell/k$ and $E[1-D_i]=(k-\ell)/k$.  Therefore,
\[
E[\hat\Delta_{j,n}]
=
\frac{1}{\ell}\sum_{i\in\lambda_j}E[D_i]Y_i(1)
-
\frac{1}{k-\ell}\sum_{i\in\lambda_j}E[1-D_i]Y_i(0)
=
\Delta_{j,n}~.
\]
Because treatment assignments are independent across strata, $\hat\Delta_{a,n}$ and $\hat\Delta_{b,n}$ are independent for $1 \leq a\ne b \leq m$, and therefore
\begin{align*}
E[(\hat\Delta_{a,n}-\hat\Delta_{b,n})^2] & =\var[\hat\Delta_{a,n}-\hat\Delta_{b,n}] +(E[\hat\Delta_{a,n}-\hat\Delta_{b,n}])^2 \\
& =\var[\hat\Delta_{a,n}]+\var[\hat\Delta_{b,n}]+(\Delta_{a,n}-\Delta_{b,n})^2~.    
\end{align*}
It follows that
\begin{align*}
E[\hat V_n(\omega_m)] &= \frac{1}{m^2}\sum_{1 \leq a < b \leq m}\omega_{ab,m} ( \var[\hat\Delta_{a,n}]+\var[\hat\Delta_{b,n}]+(\Delta_{a,n}-\Delta_{b,n})^2)\\
&= \frac{1}{m^2}\sum_{1 \leq a < b \leq m}\omega_{ab,m}(\Delta_{a,n}-\Delta_{b,n})^2 + \frac{1}{m^2}\sum_{1 \le j \le m} \mathrm{deg}_{j,m}\var[\hat\Delta_{j,n}]~.
\end{align*}
Since
\[
\var[\hat\Delta_n] = \var\bigg[\frac{1}{m}\sum_{1 \le j \le m}\hat\Delta_{j,n}\bigg] = \frac{1}{m^2}\sum_{1 \le j \le m} \var[\hat\Delta_{j,n}]~,
\]
Theorem \ref{thm:main}(a) follows.

It remains to prove Theorem \ref{thm:main}(b). Let $\varepsilon_{j,n}=\hat\Delta_{j,n}-\Delta_{j,n}$ and $\varepsilon_n=(\varepsilon_{1,n},\ldots,\varepsilon_{m,n})'$, and note $E[\varepsilon_n] = 0$. Next, because $L_m(\omega_m)$ is symmetric by \eqref{eq:graph_weights}, \eqref{eq:laplacian_quadratic} implies that
\begin{align*}
\hat V_n(\omega_m) & = \frac{1}{m^2}(\delta_n+\varepsilon_n)'L_m(\omega_m)(\delta_n+\varepsilon_n) \\
& = \frac{1}{m^2} \big ( \delta_n'L_m(\omega_m)\delta_n + 2\delta_n'L_m(\omega_m)\varepsilon_n + \varepsilon_n'L_m(\omega_m)\varepsilon_n \big )~.
\end{align*}
Since $E[\varepsilon_n] = 0$, we have
\[
E[\hat V_n(\omega_m)] = \frac{1}{m^2} \big ( \delta_n'L_m(\omega_m)\delta_n + E[\varepsilon_n'L_m(\omega_m)\varepsilon_n] \big )~.
\]
Subtracting the previous two displays and multiplying by $m$, we get
\begin{equation} \label{eq:Vomega_decomp}
m \big ( \hat V_n(\omega_m)-E[\hat V_n(\omega_m)] \big ) = A_n+B_n~,    
\end{equation}
where
\begin{align*}
A_n & = \frac{2}{m}\delta_n'L_m(\omega_m)\varepsilon_n \\
B_n & = \frac{1}{m} \big ( \varepsilon_n'L_m(\omega_m)\varepsilon_n - E[\varepsilon_n'L_m(\omega_m)\varepsilon_n] \big )~.
\end{align*}
It suffices to show $A_n$ and $B_n$ in \eqref{eq:Vomega_decomp} satisfy
\begin{align}
\label{eq:An_0} A_n & \xrightarrow{P} 0 \\
\label{eq:Bn_0} B_n & \xrightarrow{P} 0~,
\end{align}
which together with \eqref{eq:Vomega_decomp} implies Theorem \ref{thm:main}(b) because $n = mk$.

To show \eqref{eq:An_0} and \eqref{eq:Bn_0}, we first record some simple implications of the fourth-moment condition in \eqref{eq:four_moment}. First consider $\Delta_{j,n}$ for $1 \leq j \leq m$. Since
\[
\Delta_{j,n}
=
\frac{1}{k}\sum_{i\in\lambda_j} (Y_i(1)-Y_i(0))~,
\]
it follows from the convexity of $x \mapsto x^4$ and Jensen's inequality that
\[
\Delta_{j,n}^4
\leq
\frac{1}{k}\sum_{i\in\lambda_j} (Y_i(1)-Y_i(0))^4~.
\]
Next, it follows from Jensen's inequality that $(|x| + |y|)^r \leq 2^{r - 1} (|x|^r + |y|^r)$ for $r \geq 1$. Taking $r = 4$ and using $|x - y| \leq |x| + |y|$, we get
\[
\Delta_{j,n}^4
\leq
\frac{8}{k}
\sum_{i\in\lambda_j} (Y_i(1)^4+Y_i(0)^4)~,
\]
and hence
\begin{equation} \label{eq:delta_four}
\sum_{1 \le j \le m} \Delta_{j,n}^4
\leq
\frac{8}{k}
\sum_{1 \le i \le n} (Y_i(1)^4+Y_i(0)^4)
=
o(m^2)
\end{equation}
because of \eqref{eq:four_moment}.
Next, consider $\varepsilon_{j,n}=\hat\Delta_{j,n}-\Delta_{j,n}$ for $1 \leq j \leq m$. To begin, we have
\[
\varepsilon_{j,n}
=
\sum_{i\in\lambda_j}
\left(\frac{D_i}{\ell}-\frac{1}{k}\right)Y_i(1)
-
\sum_{i\in\lambda_j}
\left(\frac{1-D_i}{k-\ell}-\frac{1}{k}\right)Y_i(0)~.
\]
Because $D_i \in \{0, 1\}$,
\[
\left|\frac{D_i}{\ell}-\frac{1}{k}\right|
\leq
\frac{1}{\ell}+\frac{1}{k}~,
\qquad
\left|\frac{1-D_i}{k-\ell}-\frac{1}{k}\right|
\leq
\frac{1}{k-\ell}+\frac{1}{k}~.
\]
Therefore, setting
\[
M_{k,\ell}
=
\max\left\{
\frac{1}{\ell}+\frac{1}{k}~,
\frac{1}{k-\ell}+\frac{1}{k}
\right\}~,
\]
we have the bound
\[
|\varepsilon_{j,n}|
\leq
M_{k,\ell}
\sum_{i\in\lambda_j} (|Y_i(1)|+|Y_i(0)|)~.
\]
Taking fourth powers and applying Jensen's inequality to the average over the $k$ terms, we have
\[
\bigg(
\sum_{i\in\lambda_j} (|Y_i(1)|+|Y_i(0)|)
\bigg)^4
\leq
k^3
\sum_{i\in\lambda_j} (|Y_i(1)|+|Y_i(0)|)^4~.
\]
Applying again the inequality $(|x| + |y|)^r \leq 2^{r - 1} (|x|^r + |y|^r)$ for $r \geq 1$ with $r = 4$ and taking expectations yields
\[
E[\varepsilon_{j,n}^4]
\leq
8k^3 M_{k, \ell}^4
\sum_{i\in\lambda_j} (Y_i(1)^4+Y_i(0)^4)~,
\]
and hence it follows from \eqref{eq:four_moment} again that
\begin{equation} \label{eq:epsilon_four}
\sum_{1 \le j \le m} E[\varepsilon_{j,n}^4]
\leq
8k^3 M_{k, \ell}^4
\sum_{1 \le i \le n} (Y_i(1)^4+Y_i(0)^4)
=
o(m^2)~.
\end{equation}
Finally, since Jensen's inequality implies $(\var[\hat \Delta_{j, n}])^2 = (\operatorname{Var}[\varepsilon_{j,n}])^2 \leq E[\varepsilon_{j,n}^2]^2 \leq E[\varepsilon_{j,n}^4]$, we obtain
\begin{equation} \label{eq:var^2}
\sum_{1 \le j \le m} (\var[\hat \Delta_{j, n}])^2
\leq
\sum_{1 \le j \le m} E[\varepsilon_{j,n}^4]
=
o(m^2)~.   
\end{equation}

We proceed to showing \eqref{eq:An_0}. Because $E[\varepsilon_n]=0$ and $E[A_n]=0$, independence across strata implies
\[
\var[A_n]
=
\frac{4}{m^2}\sum_{1 \le j \le m}
\big ( [L_m(\omega_m)'\delta_n]_j \big )^2\var[\hat \Delta_{j, n}]~,
\]
where $[L_m(\omega_m)'\delta_n]_j$ is the $j$th entry of $L_m(\omega_m)'\delta_n$. By the symmetry of $L_m(\omega_m)$ and the Cauchy--Schwarz inequality,
\[
\sum_{1 \le j \le m}\big ( [L_m(\omega_m)'\delta_n]_j \big )^2\var[\hat \Delta_{j, n}]
\le
\bigg(\sum_{1 \le j \le m}\big ( [L_m(\omega_m)\delta_n]_j \big )^4 \bigg)^{1/2}
\bigg(\sum_{1 \le j \le m} \var[\hat \Delta_{j, n}]^2\bigg)^{1/2}~.
\]
It follows from \eqref{eq:laplacian_diag} and \eqref{eq:laplacian_offdiag} that for $1 \leq j \leq m$,
\[
[L_m(\omega_m)\delta_n]_j = \sum_{b\ne j}\omega_{jb,m}(\Delta_{j,n}-\Delta_{b,n})~.
\]
Next, note because $\omega_{ab,m} \geq 0$, Hölder's inequality implies
\begin{align*}
\sum_{b \ne j} \omega_{jb,m}|x_b| & = \sum_{b \ne j} \omega_{jb,m}^{3/4} \big ( \omega_{jb,m}^{1/4}|x_b| \big ) \\
& \leq \Big ( \sum_{b \ne j} \big ( \omega_{jb,m}^{3/4} \big )^{4/3} \Big )^{3/4} \Big ( \sum_{b \ne j} \big ( \omega_{jb,m}^{1/4}|x_b| \big )^4 \Big )^{1/4} = \Big ( \sum_{b \ne j}\omega_{jb,m} \Big )^{3/4} \Big ( \sum_{b \ne j}\omega_{jb,m}|x_b|^4 \Big )^{1/4}    
\end{align*}
and hence $\big ( \sum_{b \ne j} \omega_{jb,m}|x_b| \big )^4 \le \big ( \sum_{b \ne j}\omega_{jb,m} \big )^3\sum_{b \ne j}\omega_{jb,m}|x_b|^4$. Therefore,
\begin{align*}
\sum_{1 \le j \le m} \big ( (L_m(\omega_m)\delta_n)_j \big )^4
&\le
\sum_{1 \le j \le m} \Big ( \sum_{b \neq j} \omega_{jb,m} \Big )^3 \Big ( \sum_{b\ne j}\omega_{jb,m}(\Delta_{j,n}-\Delta_{b,n})^4 \Big )\\
&\le
8C^3\sum_{1 \le j \le m}\sum_{b\ne j}\omega_{jb,m}(\Delta_{j,n}^4+\Delta_{b,n}^4)\\
&\le
16C^4\sum_{1 \le j \le m}\Delta_{j,n}^4
=o(m^2)~,
\end{align*}
where the second inequality follows from \eqref{eq:deg_bound} and the inequality $(x-y)^4\le 8(x^4+y^4)$ and the last inequality follows from \eqref{eq:delta_four}. Therefore, $\var[A_n]=o(1)$, and \eqref{eq:An_0} follows from Markov's inequality.

Finally, we show \eqref{eq:Bn_0}. Since $E[B_n]=0$, by Markov's inequality, it suffices to show $\var[B_n] \to 0$. Write
\[
\varepsilon_n'L_m(\omega_m)\varepsilon_n
=
\sum_{1 \le j \le m}\mathrm{deg}_{j,m}\varepsilon_{j,n}^2
-2\sum_{1 \leq a<b \leq m}\omega_{ab,m}\varepsilon_{a,n}\varepsilon_{b,n}~.
\]
For the first term,
\begin{equation} \label{eq:Bn_1}
\var\bigg[\sum_{1 \le j \le m}\mathrm{deg}_{j,m}\varepsilon_{j,n}^2\bigg]
=
\sum_{1 \le j \le m}\mathrm{deg}_{j,m}^2\var[\varepsilon_{j,n}^2]
\le
C^2\sum_{1 \le j \le m}E[\varepsilon_{j,n}^4]
=o(m^2)~,
\end{equation}
where the inequality follows from \eqref{eq:deg_bound} and the last equality follows from \eqref{eq:epsilon_four}. For the second term,
\begin{align} 
\nonumber \var\bigg[\sum_{1 \leq a<b \leq m}\omega_{ab,m}\varepsilon_{a,n}\varepsilon_{b,n}\bigg] & = E \bigg[\bigg ( \sum_{1 \leq a<b \leq m}\omega_{ab,m}\varepsilon_{a,n}\varepsilon_{b,n} \bigg )^2 \bigg] \\
\label{eq:Bn_2} & = \sum_{1 \leq a < b \leq m}\sum_{1 \leq c < d \leq m}\omega_{ab,m}\omega_{cd,m}E[\varepsilon_{a,n}\varepsilon_{b,n}\varepsilon_{c,n}\varepsilon_{d,n}]~.
\end{align}
Now consider $1 \leq a < b \leq m$ and $1 \leq c < d \leq m$. If $a$, $b$, $c$, $d$ are all different, then $E[\varepsilon_{a,n}\varepsilon_{b,n}\varepsilon_{c,n}\varepsilon_{d,n}] = 0$ by independence. If they share exactly one endpoint, so that they can be denoted without loss of generality by $\{a,b\}$ and $\{a,d\}$, then
\[
E[\varepsilon_{a,n}\varepsilon_{b,n}\varepsilon_{c,n}\varepsilon_{d,n}]
=
E[\varepsilon_{a,n}^2\varepsilon_{b,n}\varepsilon_{d,n}]
=
E[\varepsilon_{a,n}^2]E[\varepsilon_{b,n}]E[\varepsilon_{d,n}]
=
0~.
\]
Therefore, the only nonzero terms on the right-hand side of \eqref{eq:Bn_2} arise when $c = a$ and $d = b$. For such terms,
\[
E[\varepsilon_{a,n}\varepsilon_{b,n}\varepsilon_{c,n}\varepsilon_{d,n}]
=
E[\varepsilon_{a,n}^2\varepsilon_{b,n}^2]
=
E[\varepsilon_{a,n}^2]E[\varepsilon_{b,n}^2]
=
\var[\hat \Delta_{a, n}] \var[\hat \Delta_{b, n}]~,
\]
and hence
\[
\var\bigg[ \sum_{1 \leq a<b \leq m}\omega_{ab,m}\varepsilon_{a,n}\varepsilon_{b,n} \bigg ]
=
\sum_{1 \leq a<b \leq m}\omega_{ab,m}^2\var[\hat \Delta_{a, n}] \var[\hat \Delta_{b, n}]~.
\]
Because $0\le \omega_{ab,m}\le \mathrm{deg}_{a,m}\le C$, it follows from the inequality $xy \leq (x^2 + y^2) / 2$ and \eqref{eq:var^2} that
\begin{align}
\nonumber \sum_{1 \leq a<b \leq m}\omega_{ab,m}^2 \var[\hat \Delta_{a, n}] \var[\hat \Delta_{b, n}]
&\le
\frac{C}{2}\sum_{1 \leq a<b \leq m}\omega_{ab,m}(\var[\hat \Delta_{a, n}]^2 + \var[\hat \Delta_{b, n}]^2)\\
\label{eq:Bn_3} &=
\frac{C}{2}\sum_{1 \le j \le m}\mathrm{deg}_{j,m}\var[\hat \Delta_{j, n}]^2
\le
\frac{C^2}{2}\sum_{1 \le j \le m} \var[\hat \Delta_{j, n}]^2
=o(m^2)~.
\end{align}
The desired result in \eqref{eq:Bn_0} now follows from \eqref{eq:Bn_1}--\eqref{eq:Bn_3}. \qed

\subsection{Proof of Proposition \ref{prop:no_uniform_ordering}}
Let $L_m$ and $\tilde L_m$ be the graph-Laplacians associated with $\omega_m$ and $\tilde \omega_m$. The identity \eqref{eq:x_quadratic} and the inequality \eqref{eq:order} imply $\tilde L_m-L_m$ is positive semidefinite. Because $\omega_m$ and $\tilde \omega_m$ are both degree-calibrated, $L_m$ and $\tilde L_m$ have the same diagonal entries, all equal to one, so $\tr(\tilde L_m-L_m) = 0$. A positive semidefinite matrix with zero trace is the zero matrix, so $\tilde L_m-L_m = 0$.  It then follows from \eqref{eq:graph_weights}, \eqref{eq:laplacian_diag}, and \eqref{eq:laplacian_offdiag} that $\omega_m = \tilde \omega_m$. \qed

\subsection{Proof of Theorem \ref{thm:minimax}}
To show Theorem \ref{thm:minimax}(a), recall that any $\omega_m \in \Omega_m$ is degree-calibrated and therefore satisfies $L_m(\omega_m) = I_m - \omega_m$. Next, note that $\delta_n = P_m \delta_n + \frac{1}{m} \iota_m \iota_m' \delta_n$, where the two components are orthogonal to each other and $\frac{1}{m} \iota_m \iota_m' \delta_n = \Delta_n \iota_m \in \mathrm{span}(\iota_m)$, so $P_m \delta_n$ is the projection of $\delta_n$ on $\mathrm{span}(\iota_m)^\perp$, the orthogonal complement of $\mathrm{span}(\iota_m)$. Because $L_m(\omega_m) \iota_m = (I_m - \omega_m) \iota_m = 0$ by degree calibration, we have
\[
\delta_n' L_m(\omega_m) \delta_n = (P_m \delta_n)' L_m(\omega_m) (P_m \delta_n)~.
\]
Therefore, it follows from Theorem 4.2.2 in \cite{horn1985matrix} and $\|P_m\delta_n\|^2 = m \|P_m\delta_n\|_m^2$, where $\|\cdot\|$ denotes the Euclidean norm, that
\begin{equation} \label{eq:quad_ev}
\sup_{\delta_n \in \mathbb R^m: \|P_m \delta_n\|_m^2 \leq C}
\frac{1}{m}\delta_n' L_m(\omega_m) \delta_n
= C \lambda_{\mathrm{max}}(L_m(\omega_m))~,
\end{equation}
where $\lambda_{\mathrm{max}}(L_m(\omega_m))$ is the maximum eigenvalue of $L_m(\omega_m)$. Next, \eqref{eq:x_quadratic} shows $L_m(\omega_m)$ is positive semidefinite, so its eigenvalues are all nonnegative. Furthermore, $L_m(\omega_m)$ has at least one zero eigenvalue because $L_m(\omega_m) \iota_m = 0$. Because the sum of the remaining $m - 1$ eigenvalues equals $\tr(L_m(\omega_m)) = \tr(I_m) - \tr(\omega_m) = m - 0 = m$, $\lambda_{\mathrm{max}}(L_m(\omega_m))$ is at least $m / (m - 1)$, and therefore
\begin{equation} \label{eq:quad_bound}
\sup_{\delta_n \in \mathbb R^m: \|P_m \delta_n\|_m^2 \leq C}
\frac{1}{m}\delta_n' L_m(\omega_m) \delta_n
= C \lambda_{\mathrm{max}}(L_m(\omega_m)) \geq C \frac{m}{m - 1}~.
\end{equation}
On the other hand, for $x \in \mathrm{span}(\iota_m)^\perp$, we have $\iota_m' x = 0$ and therefore
\[ L_m(\omega_m^{\rm CG}) x = (I_m - \omega_m^{\rm CG}) x = x - \frac{1}{m - 1} (\iota_m \iota_m' x - x) = \frac{m}{m - 1} x~. \]
Because $\iota_m$ and any basis of $\mathrm{span}(\iota_m)^\perp$ form a basis for $\mathbb R^m$, we have that $L_m(\omega_m^{\rm CG})$ has one eigenvalue of zero and $m - 1$ eigenvalues $\frac{m}{m - 1}$, and therefore it follows from \eqref{eq:quad_ev} that
\begin{equation} \label{eq:quad_attained}
\sup_{\delta_n \in \mathbb R^m: \|P_m \delta_n\|_m^2 \leq C}
\frac{1}{m}\delta_n'L_m(\omega_m^{\rm CG})\delta_n
= C \frac{m}{m- 1}~.
\end{equation}
The results in \eqref{eq:quad_bound} and \eqref{eq:quad_attained} show that the infimum in the theorem is attained by $\omega_m^{\rm CG}$.

To prove uniqueness, suppose that some $\omega_m\in\Omega_m$ also attains the infimum. Then, it follows from \eqref{eq:quad_bound} and \eqref{eq:quad_attained} that $L_m(\omega_m)$ also has one eigenvalue of zero and $m - 1$ eigenvalues $\frac{m}{m - 1}$. $L_m(\omega_m)$ is positive semidefinite, so it is diagonalizable with respect to an orthonormal basis. Because $L_m(\omega_m) \iota_m = 0$, $\iota_m$ is the eigenvector associated with the zero eigenvalue. It therefore follows that for any vector $u \in \mathrm{span}(\iota_m)^\perp$, $L_m(\omega_m) u = \frac{m}{m - 1} u$. Because $\iota_m$ and any basis of $\mathrm{span}(\iota_m)^\perp$ form a basis for $\mathbb R^m$, and $L_m(\omega_m)$ and $L_m(\omega_m^{\rm CG})$ coincide on this basis, we have $L_m(\omega_m) = L_m(\omega_m^{\rm CG})$. It then follows from \eqref{eq:graph_weights}, \eqref{eq:laplacian_diag}, and \eqref{eq:laplacian_offdiag} that $\omega_m=\omega_m^{\rm CG}$.

To show Theorem \ref{thm:minimax}(b), suppose $L_m(\omega_m) \preceq \kappa_m P_m$. Then
\[
\frac{1}{m}\delta_n' L_m(\omega_m) \delta_n
\leq \frac{\kappa_m}{m}\delta_n' P_m \delta_n
= \kappa_m\|P_m\delta_n\|_m^2
\le \kappa_m C~.
\]

Finally, we show Theorem \ref{thm:minimax}(c). Because $\omega_m\in\mathcal M_m$ is a perfect matching graph and the set $\{\delta_n \in \mathbb R^m: \|P_m \delta_n\|_m^2 \leq C\}$ is invariant to relabeling, we can reorder the indices without loss of generality so that the matched pairs are $\{2j-1,2j\}$ for $1 \le j \le m/2$. Then, $L_m(\omega_m)$ is block-diagonal with $2 \times 2$ blocks of
\[
\begin{pmatrix}
    1 & -1 \\ -1 & 1
\end{pmatrix}~.
\]
Each of these blocks has two eigenvalues 0 and 2 and corresponding eigenvectors $(1, 1)'$ and $(1, -1)'$, which form a basis in $\mathbb R^2$. Letting $f_1, \dots, f_{m / 2}$ denote the standard basis in $\mathbb R^{m / 2}$, we see that $\{f_j\otimes (1, 1)', f_j\otimes (1, -1)': 1 \leq j \leq m / 2\}$ form a basis in $\mathbb R^m$, so that $L_m(\omega_m)$ also has only two eigenvalues 0 and 2, each with multiplicity $m / 2$. The desired conclusion now follows from \eqref{eq:quad_ev}.
\qed

\subsection{Proof of Theorem \ref{thm:limit_V}}

In the proof of Theorem \ref{thm:normal} we have obtained
\[
n\cdot\var[\hat\Delta_n]
=
k\eta^2\left(\frac{1}\ell-\frac{1}{k}\right)\frac{1}{k-1}
\frac{1}{m}\sum_{1 \le j \le m}\sum_{i\in\lambda_j}(R_i-\bar R_j)^2~.
\]
Below we will argue that 
\begin{equation}\label{eq:R_limit}
\frac{1}{m}\sum_{1 \le j \le m}\sum_{i\in\lambda_j}(R_i-\bar R_j)^2
\to
(k-1)E_Q[\var_Q[\tilde R\mid \tilde X]]~.
\end{equation}
for
\[\tilde R=\frac{\tilde Y(1)}{\eta}+\frac{\tilde Y(0)}{1-\eta}~.\]
If \eqref{eq:R_limit} holds, then because $k\eta^2(1/\ell-1/k)=\eta(1-\eta)$, we have
\[
n \var[\hat\Delta_n]
\to
\eta(1-\eta)E_Q[\var_Q[\tilde R\mid \tilde X]]~,
\]
which leads to the conclusion of the theorem because
\begin{align*}
\eta(1-\eta)\var_Q[\tilde R\mid \tilde X] & = \eta(1-\eta)\var_Q\left[\frac{\tilde Y(1)}{\eta}+\frac{\tilde Y(0)}{1-\eta}\Bigm|\tilde X\right] \\
&=
\frac{1-\eta}{\eta}\var_Q[\tilde Y(1)\mid\tilde X]
+
\frac{\eta}{1-\eta}\var_Q[\tilde Y(0)\mid\tilde X] +
2\cov_Q[\tilde Y(1),\tilde Y(0)\mid\tilde X]\\
&=
\frac{\var_Q[\tilde Y(1)\mid\tilde X]}{\eta}
+
\frac{\var_Q[\tilde Y(0)\mid\tilde X]}{1-\eta}-\var_Q[\tilde Y(1)-\tilde Y(0)\mid\tilde X]~.
\end{align*}
To show \eqref{eq:R_limit}, first note
\begin{equation} \label{eq:Ri2_converge}
\frac{1}n\sum_{1 \le i \le n} R_i^2
=
\frac{1}{\eta^2}\frac{1}n\sum_{1 \le i \le n}Y_i(1)^2
+
\frac{1}{(1-\eta)^2}\frac{1}n\sum_{1 \le i \le n}Y_i(0)^2
+
\frac{2}{\eta(1-\eta)}\frac{1}n\sum_{1 \le i \le n}Y_i(1)Y_i(0) \to E_Q[\tilde R^2]~.    
\end{equation}
by Assumption \ref{ass:super}(a). Next, because
\[
\frac{1}{\binom{k}{2}}\sum_{i<i'\in\lambda_j}R_iR_{i'}
=
\frac{1}{k(k-1)}
\sum_{\substack{i,i'\in\lambda_j\\ i\ne i'}}R_iR_{i'}~,
\]
it follows from Assumption \ref{ass:super}(b) that
\begin{align}
\nonumber & \frac{1}{m}\sum_{1 \le j \le m}\frac{1}{\binom{k}{2}}
\sum_{i<i'\in\lambda_j}R_iR_{i'} \\
& \to \frac{1}{\eta^2}E_Q[E_Q[\tilde{Y}(1)\mid \tilde{X}]^2]
+
\nonumber \frac{1}{(1-\eta)^2}E_Q[E_Q[\tilde{Y}(0)\mid \tilde{X}]^2] + 
\frac{2}{\eta(1-\eta)}E_Q[E_Q[\tilde{Y}(1)\mid \tilde{X}]E_Q[\tilde{Y}(0)\mid \tilde{X}]]
 \\
\label{eq:Rcross_converge} & = E_Q[E_Q[\tilde R\mid \tilde X]^2]~.
\end{align}
Finally,
\begin{align}
\nonumber \frac{1}{m}\sum_{1 \le j \le m}\sum_{i\in\lambda_j}(R_i-\bar R_j)^2
&=
\frac{1}{m}\sum_{1 \le j \le m}\sum_{i\in\lambda_j}R_i^2
-
\frac{k}{m}\sum_{1 \le j \le m}\bar R_j^2\\
\nonumber &=
\frac{1}{m}\sum_{1 \le j \le m}\sum_{i\in\lambda_j}R_i^2
-
\frac{1}{mk}\sum_{1 \le j \le m}\sum_{i,i'\in\lambda_j}R_iR_{i'}\\
\label{eq:R_converge_decomp} &=
(k-1)\left(
\frac{1}n\sum_{1 \le i \le n}R_i^2
-
\frac{1}{m}\sum_{1 \le j \le m}\frac{1}{\binom{k}{2}}\sum_{i<i'\in\lambda_j}R_iR_{i'}
\right)~,
\end{align}
and \eqref{eq:R_limit} follows from \eqref{eq:Ri2_converge}--\eqref{eq:R_converge_decomp}. \qed

\subsection{Proof of Theorem \ref{thm:main_var}}
Because $\omega_m$ is degree-calibrated, it follows from the definition in \eqref{eq:graph_estimator} that we can rewrite the variance estimator as
\[
\hat V_n(\omega_m)
=
\frac{1}{m^2}\left(
\sum_{1 \le j \le m}\hat\Delta_{j,n}^2
-
\sum_{1 \le a \le m}\sum_{b\ne a}\omega_{ab,m}\hat\Delta_{a,n}\hat\Delta_{b,n}
\right)~.
\]
Taking expectations and multiplying by $n=mk$, we have
\begin{equation} \label{eq:main_var1}
nE[\hat V_n(\omega_m)]
=
k\frac{1}{m}\sum_{1 \le j \le m}E[\hat\Delta_{j,n}^2]
-
k\frac{1}{m}\sum_{1 \le a \le m}\sum_{b\ne a}\omega_{ab,m}\Delta_{a,n}\Delta_{b,n}~,    
\end{equation}
where we use $E[\hat\Delta_{a,n}\hat\Delta_{b,n}]=\Delta_{a,n}\Delta_{b,n}$ for $a\ne b$ by independence across strata. Under Assumption \ref{ass:dn}, for $i \ne i' \in \lambda_j$ and $1 \leq j \leq m$,
\begin{align*}
E[D_i] &= \frac{\ell}{k} \\
E[D_iD_{i'}] &= \frac{\ell(\ell-1)}{k(k-1)}\\
E[(1-D_i)(1-D_{i'})] &= \frac{(k-\ell)(k-\ell-1)}{k(k-1)} \\
E[D_i(1-D_{i'})] &= \frac{\ell(k-\ell)}{k(k-1)}~.
\end{align*}
These results imply
\begin{align*}
E[\hat\Delta_{j,n}^2]
&=
\frac{1}{\ell}\frac{1}{k}\sum_{i\in\lambda_j}Y_i(1)^2
+
\frac{\ell-1}{\ell}\frac{1}{\binom{k}{2}}\sum_{i<i'\in\lambda_j}Y_i(1)Y_{i'}(1)\\
&\hspace{2em}+
\frac{1}{k-\ell}\frac{1}{k}\sum_{i\in\lambda_j}Y_i(0)^2
+
\frac{k-\ell-1}{k-\ell}\frac{1}{\binom{k}{2}}\sum_{i<i'\in\lambda_j}Y_i(0)Y_{i'}(0)\\
&\hspace{2em}-
\frac{1}{\binom{k}{2}}\sum_{\substack{i,i'\in\lambda_j\\ i\ne i'}}Y_i(1)Y_{i'}(0)~.
\end{align*}
Averaging over $j$ and applying Assumption \ref{ass:super}(a) and (b) term by term, we get
\begin{align*}
\frac{1}{m}\sum_{1 \le j \le m}E[\hat\Delta_{j,n}^2]
&\to
\frac{1}\ell E_Q[\tilde Y(1)^2]
+
\frac{\ell-1}{\ell}E_Q[E_Q[\tilde{Y}(1)\mid \tilde X]^2]\\
&\hspace{2em}+
\frac{1}{k-\ell}E_Q[\tilde Y(0)^2]
+
\frac{k-\ell-1}{k-\ell}E_Q[E_Q[\tilde{Y}(0)\mid \tilde X]^2]\\
&\hspace{2em}-
2E_Q[E_Q[\tilde{Y}(1)\mid \tilde X]E_Q[\tilde{Y}(0)\mid \tilde X]]~.
\end{align*}
Multiplying by $k$ and using $\eta=\ell/k$ and the conditional variance decomposition, we have
\begin{equation} \label{eq:main_var2}
k\cdot\frac{1}{m}\sum_{1 \le j \le m}E[\hat\Delta_{j,n}^2] \to V^{\rm obs}+kE_Q[E_Q[\tilde Y(1) - \tilde Y(0) | \tilde X]^2]~.
\end{equation}
Next, because
\[
\Delta_{a,n}\Delta_{b,n}
=
(\bar Y_{a,n}(1)-\bar Y_{a,n}(0)) (\bar Y_{b,n}(1)-\bar Y_{b,n}(0))~,
\]
Assumption \ref{ass:limit_graph_products} implies
\begin{equation} \label{eq:main_var3}
\frac{1}{m}\sum_{1 \le a \le m}\sum_{b\ne a}\omega_{ab,m}\Delta_{a,n}\Delta_{b,n}
\to
E_Q[E_Q[\theta(\tilde{X})\mid\Gamma_\omega(\tilde{X})]^2]~.    
\end{equation}
The results in \eqref{eq:main_var1}--\eqref{eq:main_var3} imply
\[
nE[\hat V_n(\omega_m)]
\to
V^{\rm obs}+k E_Q\left[\var_Q[\theta(\tilde{X})\mid\Gamma_\omega(\tilde{X})]\right]~.
\]

To prove the convergence of the coverage probability, note as shown in the proof of Theorem \ref{thm:upward}, \eqref{eq:pop_moments} and \eqref{eq:pop_nondegen} imply
\[
\frac{\hat \Delta_n-\Delta_n}{\var[\hat\Delta_n]^{1/2}}
\stackrel{d}{\to} N(0,1)~.
\]
By Theorem \ref{thm:limit_V}, $n \var[\hat\Delta_n]\to V$. On the other hand,
\[
n \hat V_n(\omega_m)
=
nE[\hat V_n(\omega_m)]
+
n(\hat V_n(\omega_m)-E[\hat V_n(\omega_m)])
\xrightarrow{P}
V_\omega~.
\]
Therefore, it follows from Slutsky's theorem that
\[
\frac{\hat\Delta_n-\Delta_n}{\sqrt{\hat V_n(\omega_m)}}
=
\frac{\sqrt n(\hat\Delta_n-\Delta_n)}
{\sqrt{n \hat V_n(\omega_m)}}
\stackrel{d}{\to}
N(0,V/V_\omega)~.
\]
The desired conclusion now follows because $\varsigma_\omega=(V/V_\omega)^{1/2}$.
\qed

\subsection{Proof of Proposition \ref{prop:complete_graph_product}}
Because $\omega_{ab,m}^{\rm CG}=1/(m-1)$,
\begin{align}
\nonumber & \frac{1}{m}\sum_{1 \le a \le m}\sum_{b\ne a}\omega_{ab,m}^{\rm CG}
\bar Y_{a,n}(d)\bar Y_{b,n}(d') \\
\nonumber & = \frac{1}{m}\sum_{1 \le a \le m}\sum_{b\ne a}\frac{1}{m-1}
\bar Y_{a,n}(d)\bar Y_{b,n}(d')\\
\label{eq:cg_converge1} & =
\frac{m}{m-1}
\bigg(\frac{1}{m}\sum_{1 \le j \le m}\bar Y_{j,n}(d)\bigg)
\bigg(\frac{1}{m}\sum_{1 \le j \le m}\bar Y_{j,n}(d')\bigg)
-
\frac{1}{m(m-1)}
\sum_{1 \le j \le m}
\bar Y_{j,n}(d)\bar Y_{j,n}(d')~.
\end{align}
By Assumption \ref{ass:super}(a) with $r = 1$,
\begin{align}
\label{eq:cg_converge2} & \frac{1}{m}\sum_{1 \le j \le m}\bar Y_{j,n}(d) \to E_Q[\tilde{Y}(d)] \\
\label{eq:cg_converge3} & \frac{1}{m}\sum_{1 \le j \le m}\bar Y_{j,n}(d') \to E_Q[\tilde{Y}(d')]~.
\end{align}
Meanwhile, by the Cauchy--Schwarz inequality, Jensen's inequality, and Assumption \ref{ass:super}(a) with $r = 2$,
\begin{align}
\frac{1}{m}\sum_{1 \le j \le m}
|\bar Y_{j,n}(d)\bar Y_{j,n}(d')|
\nonumber &\le
\bigg(\frac{1}{m}\sum_{1 \le j \le m}\bar Y_{j,n}(d)^2\bigg)^{1/2}
\bigg(\frac{1}{m}\sum_{1 \le j \le m}\bar Y_{j,n}(d')^2\bigg)^{1/2}\\
\nonumber &\le
\bigg(\frac{1}n\sum_{1 \le i \le n}Y_i(d)^2\bigg)^{1/2}
\bigg(\frac{1}n\sum_{1 \le i \le n}Y_i(d')^2\bigg)^{1/2} \\
\label{eq:cg_converge4} & \to (E_Q[\tilde Y(d)^2])^{1/2} (E_Q[\tilde Y(d')^2])^{1/2}~.
\end{align}
The desired conclusion of the proposition now follows from \eqref{eq:cg_converge1}--\eqref{eq:cg_converge4}.
\qed

\subsection{Proof of Proposition \ref{prop:local_graph_product}}
We first show that
\begin{equation}\label{eq:smooth_bilinear_local}
\frac{1}{m}\sum_{1 \le a \le m}\sum_{b \ne a}\omega_{ab,m}\bar{\mu}_{a,n}(d)\bar{\mu}_{b,n}(d')-\frac{1}n\sum_{1 \le i \le n}\mu_d(X_i)\mu_{d'}(X_i)\to0~.
\end{equation}
Because $\omega_m$ is degree-calibrated,
\begin{align*}
&\frac{1}{m}\sum_{1 \le a \le m}\sum_{b \ne a}\omega_{ab,m}\bar{\mu}_{a,n}(d)\bar{\mu}_{b,n}(d')
-
\frac{1}{m}\sum_{1 \le j \le m}\bar \mu_{j,n}(d)\bar \mu_{j,n}(d')\\
&\qquad=
\frac{1}{m}\sum_{1 \le a \le m}\sum_{b\ne a}\omega_{ab,m}
\bar \mu_{a,n}(d)(\bar \mu_{b,n}(d')-\bar \mu_{a,n}(d'))~.
\end{align*}
Therefore, by the Cauchy--Schwarz inequality,
\begin{align*}
&\bigg|
\frac{1}{m}\sum_{1 \le a \le m}\sum_{b \ne a}\omega_{ab,m}\bar{\mu}_{a,n}(d)\bar{\mu}_{b,n}(d')
-
\frac{1}{m}\sum_{1 \le j \le m}\bar \mu_{j,n}(d)\bar \mu_{j,n}(d')
\bigg|^2\\
&\le
\bigg(
\frac{1}{m}\sum_{1 \le a \le m}\sum_{b\ne a}\omega_{ab,m}\bar \mu_{a,n}(d)^2
\bigg)
\bigg(
\frac{1}{m}\sum_{1 \le a \le m}\sum_{b\ne a}\omega_{ab,m}
(\bar \mu_{b,n}(d')-\bar \mu_{a,n}(d'))^2
\bigg)~.
\end{align*}
The first bracket equals $\frac{1}{m}\sum_{1 \le j \le m}\bar{\mu}_{j,n}(d)^2$ by degree calibration and thus is $O(1)$ by Jensen's inequality and Assumption \ref{ass:lipschitz}(b). Because $\mu_{d'}$ is Lipschitz by Assumption \ref{ass:lipschitz}(a),
\[
\frac{1}{m}\sum_{1 \le a \le m}\sum_{b\ne a}\omega_{ab,m} (\bar \mu_{b,n}(d')-\bar \mu_{a,n}(d'))^2
\lesssim
\frac{1}{m}\sum_{1 \le a \le m}\sum_{b\ne a}\omega_{ab,m} \max_{i,i'\in\lambda_a\cup\lambda_b}\|X_i-X_{i'}\|^2 \leq \mathcal{D}_n(\omega_m) \to 0
\]
by assumption. Therefore,
\begin{equation}\label{eq:smooth_graph_to_stratum}
\frac{1}{m}\sum_{1 \le a \le m}\sum_{b \ne a}\omega_{ab,m}\bar{\mu}_{a,n}(d)\bar{\mu}_{b,n}(d')
-
\frac{1}{m}\sum_{1 \le j \le m}\bar \mu_{j,n}(d)\bar \mu_{j,n}(d')
\to0~.
\end{equation}

Next, for $1 \le j \le m$,
\[
\bar \mu_{j,n}(d)\bar \mu_{j,n}(d')
-
\frac{1}{k}\sum_{i\in\lambda_j}\mu_d(X_i)\mu_{d'}(X_i)
=
-\frac{1}{2k^2}\sum_{i,i'\in\lambda_j}
(\mu_d(X_i)-\mu_d(X_{i'}))(\mu_{d'}(X_i)-\mu_{d'}(X_{i'}))~.
\]
Because $\mu_{d}$ and $\mu_{d'}$ are Lipschitz by Assumption \ref{ass:lipschitz}(a),
\begin{equation} \label{eq:local_1}
\bigg|
\bar \mu_{j,n}(d)\bar \mu_{j,n}(d')
-
\frac{1}{k}\sum_{i\in\lambda_j}\mu_d(X_i)\mu_{d'}(X_i)
\bigg|
\lesssim
\max_{i,i'\in\lambda_j}\|X_i-X_{i'}\|^2~.    
\end{equation}
Again because $\omega_m$ is degree-calibrated, we have
\begin{equation} \label{eq:local_2}
\frac{1}{m}\sum_{1 \le j \le m}\max_{i,i'\in\lambda_j}\|X_i-X_{i'}\|^2
=
\frac{1}{m}\sum_{1 \le a \le m}\sum_{b\ne a}\omega_{ab,m}\max_{i,i'\in\lambda_a}\|X_i-X_{i'}\|^2
\le
\mathcal D_n(\omega_m)
\to0~.    
\end{equation}
Combining \eqref{eq:local_1} and \eqref{eq:local_2} and using $n=mk$, we have
\begin{align*}
\bigg|
\frac{1}{m}\sum_{1 \le j \le m}\bar \mu_{j,n}(d)\bar \mu_{j,n}(d')
-
\frac{1}n\sum_{1 \le i \le n}\mu_d(X_i)\mu_{d'}(X_i)
\bigg|
&\le
\frac{1}{m}\sum_{1 \le j \le m}
\bigg|
\bar \mu_{j,n}(d)\bar \mu_{j,n}(d')
-
\frac{1}{k}\sum_{i\in\lambda_j}\mu_d(X_i)\mu_{d'}(X_i)
\bigg|\\
&\to0~.
\end{align*}
Together with \eqref{eq:smooth_graph_to_stratum}, this implies \eqref{eq:smooth_bilinear_local}. It therefore follows from Assumption \ref{ass:lipschitz}(b) that
\begin{equation} \label{eq:local_3}
\frac{1}{m}\sum_{1 \le a \le m}\sum_{b \ne a}\omega_{ab,m}\bar{\mu}_{a,n}(d)\bar{\mu}_{b,n}(d') \rightarrow E_Q[E_Q[\tilde{Y}(d)\mid \tilde{X}]E_Q[\tilde{Y}(d')\mid \tilde{X}]]~.    
\end{equation}
The conclusion of the proposition then follows from $\bar Y_{a,n}(d)=\bar \mu_{a,n}(d)+\bar u_{a,n}(d)$, \eqref{eq:local_3}, and Assumption \ref{ass:lipschitz}(c).
\qed

\subsection{Proof of Proposition \ref{prop:graph_locality}}
For any $\omega_m\in\Omega_m$, define 
\begin{equation}\label{eq:ordered_mean_cost}
C_n(\omega_m)
:=
\frac{1}{m}\sum_{1 \le a \le m}\sum_{b\ne a}
\omega_{ab,m}\|\bar X_{a,n}-\bar X_{b,n}\|^2
=
\frac{2}{m}\sum_{1\le a<b\le m}\|\bar X_{a,n}-\bar X_{b,n}\|^2\omega_{ab,m}~.
\end{equation}
For any partition $\mathcal B_m$ of $\{1,\ldots,m\}$, define
\begin{equation} \label{eq:A_n}
A_n(\mathcal B_m)
:=
\frac{1}{m}\sum_{B\in\mathcal B_m}\sum_{a\in B}
\bigg\|\bar X_{a,n}-\frac{1}{|B|}\sum_{b\in B}\bar X_{b,n}\bigg\|^2~.
\end{equation}
To show part (a), note that by Lemma \ref{lem:mean_cost_to_graph_locality}, it is enough to show that
\[
C_n(\omega_m^{\rm PM})\to0~.
\]
Let $\tilde{\mathcal P}_m$ be the pair partition from Lemma \ref{lem:spatial_sorting_consequence}(a), and let $\tilde \omega_m$ the associated perfect-matching graph. Then, Lemma \ref{lem:spatial_sorting_consequence}(a) implies
\[
C_n(\tilde \omega_m)
=
4A_n(\tilde{\mathcal P}_m)
\le
4K_{p_X}M_{n}^2m^{-2/(p_X+1)}~.
\]
By the range condition \eqref{eq:pm_range_condition}, $C_n(\tilde \omega_m)\to0$. Because $\omega_m^{\rm PM}$ minimizes $(m/2)C_n(\omega)$ and therefore $C_n(\omega)$,
\[
C_n(\omega_m^{\rm PM})\le C_n(\tilde \omega_m)\to0~.
\]
Lemma \ref{lem:mean_cost_to_graph_locality} and $R_n\to0$ now imply $\mathcal D_n(\omega_m^{\rm PM})\to0$.

To show part (b), note that by Lemma \ref{lem:mean_cost_to_graph_locality}, it is enough to prove that
\[
C_n(\omega_m^{\rm SDP}(\kappa_m))\to0~.
\]
Let $\mathcal B_m=\{B_{1,m},\ldots,B_{S_m,m}\}$ be the block partition from Lemma \ref{lem:spatial_sorting_consequence}(b). The lemma implies that for all large $m$,
\[
b_m\le |B_{s,m}|\le 2b_m
\]
and combined with \eqref{eq:lp_range_condition}, we have
\begin{equation} \label{eq:blockwise_An}
A_n(\mathcal B_m)
\le
K_{p_X}M_{n}^2\left(\frac{m}{b_m}\right)^{-2/(p_X+1)}
\to 0~.    
\end{equation}
Define a blockwise complete graph $\omega_m^0$ by
\[
\omega_{ab,m}^0
=
\frac{1}{|B_{s,m}|-1}
\text{ if }a\ne b\text{ and }a,b\in B_{s,m}~,
\]
and set $\omega_{ab,m}^0=0$ if $a$ and $b$ are in different blocks. This graph is symmetric, has zero diagonal, and every row sums to one, so $\omega_m^0\in\Omega_m$. We will verify that $\omega_m^0$ is feasible for the constraints in \eqref{eq:sdp}, and that
\begin{equation} \label{eq:blockwise_Cn}
C_n(\omega_m^0) \to 0~,
\end{equation}
because then by optimality of $\omega_m^{\rm SDP}(\kappa_m)$,
\[
C_n(\omega_m^{\rm SDP}(\kappa_m))
\le
C_n(\omega_m^0)
\to0~.
\]

To show $\omega_m^0$ is feasible for the problem in \eqref{eq:sdp}, fix $z\in\mathbb R^m$. For $1 \leq s \leq S_m$, define $\bar z_{s,m} = |B_{s,m}|^{-1}\sum_{a\in B_{s,m}}z_a$. For any block $B_{s,m}$, the identity
\[
\sum_{\substack{a,b\in B_{s,m}\\ a<b}}(z_a-z_b)^2
=
|B_{s,m}|\sum_{a\in B_{s,m}}(z_a-\bar z_{s,m})^2
\]
implies that
\begin{align}
\nonumber z'L_m(\omega_m^0)z
&=
\sum_{1 \le s \le S_m}
\frac{1}{|B_{s,m}|-1}
\sum_{\substack{a,b\in B_{s,m}\\ a<b}}(z_a-z_b)^2\\
\nonumber &=
\sum_{1 \le s \le S_m}\frac{|B_{s,m}|}{|B_{s,m}|-1}
\sum_{a\in B_{s,m}}(z_a-\bar z_{s,m})^2 \\
\label{eq:blockwise_laplacian_quadratic} &\le
\left(\max_{1 \le s \le S_m}\frac{|B_{s,m}|}{|B_{s,m}|-1}\right)
\sum_{1 \le s \le S_m}\sum_{a\in B_{s,m}}(z_a-\bar z_{s,m})^2~.
\end{align}
Meanwhile, note from direct calculation that for $\bar z_m=m^{-1}\sum_{1 \le j \le m} z_j$, we have
\begin{equation} \label{eq:blockwise_quadratic_decomp}
z'P_m z = \sum_{1 \le a \le m}(z_a-\bar z_m)^2
=
\sum_{1 \le s \le S_m}\sum_{a\in B_{s,m}}(z_a-\bar z_{s,m})^2
+
\sum_{1 \le s \le S_m} |B_{s,m}|(\bar z_{s,m}-\bar z_m)^2~.    
\end{equation}
The results in \eqref{eq:blockwise_laplacian_quadratic} and \eqref{eq:blockwise_quadratic_decomp} imply
\begin{equation} \label{eq:blockwise_laplacian_bound1}
z'L_m(\omega_m^0)z
\le
\left(\max_{1 \le s \le S_m}\frac{|B_{s,m}|}{|B_{s,m}|-1}\right)z'P_mz~.    
\end{equation}
Since $|B_{s,m}|\ge b_m$ and $q/(q-1)$ is decreasing in $q$,
\begin{equation} \label{eq:blockwise_laplacian_bound2}
\max_{1 \le s \le S_m}\frac{|B_{s,m}|}{|B_{s,m}|-1}
\le
\frac{b_m}{b_m-1}
\le
\kappa_m~.
\end{equation}
Because $z \in \mathbb R^m$ is arbitrary, \eqref{eq:blockwise_laplacian_bound1} and \eqref{eq:blockwise_laplacian_bound2} imply $L_m(\omega_m^0)\preceq\kappa_mP_m$, so $\omega_m^0$ is feasible for the problem in \eqref{eq:sdp}.

It remains to show \eqref{eq:blockwise_Cn}. For $1 \leq s \leq S_m$,
\[
\sum_{a\in B_{s,m}}\sum_{\substack{b\in B_{s,m}\\ b\ne a}}
\|\bar X_{a,n}-\bar X_{b,n}\|^2
=
2|B_{s,m}|\sum_{a\in B_{s,m}}\bigg\|\bar X_{a,n}-\frac{1}{|B_{s,m}|}\sum_{b \in B_{s,m}}\bar X_{b,n}\bigg\|^2~.
\]
Consequently, because $|B_{s,m}|\ge2$ for $m$ large enough,
\begin{align*}
C_n(\omega_m^0)
&=
\frac{1}{m}\sum_{1 \le s \le S_m}
\frac{1}{|B_{s,m}|-1}
\sum_{a\in B_{s,m}}\sum_{\substack{b\in B_{s,m}\\ b\ne a}}
\|\bar X_{a,n}-\bar X_{b,n}\|^2  \\
&=
\frac{1}{m}\sum_{1 \le s \le S_m}
\frac{2|B_{s,m}|}{|B_{s,m}|-1}
\sum_{a\in B_{s,m}}\bigg\|\bar X_{a,n}-\frac{1}{|B_{s,m}|}\sum_{b \in B_{s,m}}\bar X_{b,n}\bigg\|^2  \\
&\le
\frac4m\sum_{1 \le s \le S_m}\sum_{a\in B_{s,m}}\bigg\|\bar X_{a,n}-\frac{1}{|B_{s,m}|}\sum_{b \in B_{s,m}}\bar X_{b,n}\bigg\|^2
\le
4A_n(\mathcal B_m)
\to0~,
\end{align*}
where the final convergence follows from \eqref{eq:blockwise_An}. The conclusion in (b) then follows.
\qed

% For the paired-strata graph, $\omega^{\rm Bai}_{2j-1,2j,n}=1$ and all other weights are zero.  Under the original adjacent-pair product condition,
% \[
% \frac{2}{m}\sum_{j=1}^{m/2}
% \bar Y_{2j-1,n}(d)\bar Y_{2j,n}(d')
% \to
% E_Q[\mu_d(\tilde X)\mu_{d'}(\tilde X)]
% \]
% and
% \[
% \frac{2}{m}\sum_{j=1}^{m/2}
% \bar Y_{2j,n}(d)\bar Y_{2j-1,n}(d')
% \to
% E_Q[\mu_d(\tilde X)\mu_{d'}(\tilde X)].
% \]
% Therefore
% \[
% \Gamma_{\rm Bai}(d,d')
% =
% \frac{1}2E_Q[\mu_d(\tilde X)\mu_{d'}(\tilde X)]
% +
% \frac{1}2E_Q[\mu_d(\tilde X)\mu_{d'}(\tilde X)]
% =
% E_Q[\mu_d(\tilde X)\mu_{d'}(\tilde X)].
% \]
% It follows that
% \[
% \Gamma_{\rm Bai}^\Delta=E_Q[\theta(\tilde X)^2],
% \]
% and hence
% \[
% nE[\hat V_n(\omega^{\rm Bai})]\to V^{\rm obs}.
% \]

\section{Analysis of Fogarty's Estimator}\label{app:fogarty}
In this section, we perform an analysis similar to Section \ref{sec:variances} for the variance estimator proposed by \cite{fogarty2018mitigating}. Let $R$ be an $m \times L$ matrix for $L < m$ with rank $L$ and $H_R = R (R' R)^{-1} R'$ denote the projection matrix onto its column space. For each $R$, consider the estimator given by
\begin{equation} \label{eq:fogarty}
\hat V_n^{\rm F}(R) =\frac{1}{m^2}\hat{\delta}_n' (\diag(I-H_R))^{-1/2}(I-H_R)(\diag(I-H_R))^{-1/2}\hat{\delta}_n,
\end{equation}
where $\hat{\delta}_n=(\hat \Delta_{1, n},\dots,\hat \Delta_{m, n})'$. Note that $\hat{V}_n^{\rm F}(\iota_m) = \hat{V}_n^{\rm IM}$, so that $\hat{V}_n^{\rm F}(R)$ is indeed a generalization of $\hat{V}^{\rm IM}_n$. \cite{fogarty2018mitigating} focuses on the estimator given by $\hat V_n^{\rm F} = \hat V_n^{\rm F}(Q_2)$, where
\[
Q_2 = \begin{pmatrix}
     & (\bar X_{1, n} - \mu_{X, n})' \\
    \iota_m & \vdots \\
     & (\bar X_{m, n} - \mu_{X, n})'
\end{pmatrix}
\]
for $\mu_{X, n} = \frac{1}n \sum_{1 \le i \le n} X_i$. In particular, \cite{fogarty2018mitigating} argues that $\hat{V}_n^{\rm F}$ can be less conservative than $\hat{V}_n^{\rm IM}$ whenever the covariates are predictive of the treatment effects in an appropriate sense. The bias for $\hat V_n^{\rm F}$ is \citep[see][]{fogarty2018mitigating}
\[ E[\hat V_n^{\rm F}] - \var[\hat \Delta_n] = \frac{1}{m^2} \delta_n' (\diag(I-H_{Q_2}))^{-1/2}(I-H_{Q_2})(\diag(I-H_{Q_2}))^{-1/2} \delta_n \geq 0~,  \]
so $\hat V_n^{\rm F}$ is also an upward-biased estimator for $\var[\hat \Delta_n]$.

\begin{assumption} \label{ass:X}
For $\tilde W \sim Q$ in Assumption \ref{ass:super}, $W^{(n)}$ and $\Lambda_n$ satisfy the following requirements:
\begin{enumerate}
\item[(a)] As $n \to \infty$,
\[
\frac{1}n\max_{1 \le i \le n} \lVert X_i \rVert^2 \to 0~.
\]

\item[(b)] As $n \to \infty$,
\begin{align*}
\frac{1}n \sum_{1 \le i \le n} X_i
&\to E_Q[\tilde X]~,\\
\frac{1}n \sum_{1 \le i \le n} X_iX_i'
&\to E_Q[\tilde X \tilde X']~,\\
\frac{1}n \sum_{1 \le i \le n} X_iY_i(d)
&\to E_Q[\tilde X \tilde Y(d)]
\text{ for } d \in \{0,1\}~.
\end{align*}
Furthermore, $\var_Q[\tilde X]$ is nonsingular.

\item[(c)] As $n \to \infty$,
\[
\frac{1}{m}\sum_{1 \le j \le m}
\frac{1}{k(k-1)}
\sum_{\substack{i,i'\in\lambda_j\\ i\ne i'}}
X_iX_{i'}'
\to
E_Q[\tilde X\tilde X']~,
\]
and, for each $d\in\{0,1\}$,
\[
\frac{1}{m}\sum_{1 \le j \le m}
\frac{1}{k(k-1)}
\sum_{\substack{i,i'\in\lambda_j\\ i\ne i'}}
X_iY_{i'}(d)
\to
E_Q[\tilde X\tilde Y(d)]~.
\]
\end{enumerate}
\end{assumption}
Assumption \ref{ass:X} is the analogue of Assumption \ref{ass:super} for settings with covariates. The assumptions on the cross products between the covariates and potential outcomes are needed because the estimator in \eqref{eq:fogarty} projects stratum-level treatment-effect estimates on stratum-level average baseline covariates. 
\begin{theorem}\label{thm:fogarty_var}
Suppose Assumptions \ref{ass:dn}, \ref{ass:super}, and \ref{ass:X} hold. Then,
\[
nE[\hat V_n^{\rm F}]
\to
V^{\rm obs}
+
k\var_Q\left[
E_Q[\tilde Y(1)-\tilde Y(0)\mid\tilde X]-\blp_Q(\tilde Y(1) - \tilde Y(0) \mid 1,\tilde X)
\right]~.
\]
\end{theorem}
The limit of the expectation of the variance in Theorem \ref{thm:fogarty_var} is weakly smaller than the complete-graph limit $V^{\rm obs}+k\var_Q[E[\tilde Y(1) - \tilde Y(0) | \tilde X]]$, because the linear projection on $(1,\tilde X)$ weakly reduces the variance. However, the limiting variance does not coincide with $V^{\rm obs}$ unless $E[\tilde Y(1) - \tilde Y(0) | \tilde X]$ is itself a linear function of $\tilde X$, $Q$-a.s.
\begin{proof}
We first show that $\Sigma_{X,n}\to\var_Q[\tilde X]$. By direct calculation and Assumption \ref{ass:X}(b) and (c),
\begin{align*}
\frac{1}{m}\sum_{1 \le j \le m}\bar X_{j,n}\bar X_{j,n}'
&=
\frac{1}{k}\frac{1}n\sum_{1 \le i \le n}X_iX_i'
+
\frac{k-1}{k}
\frac{1}{m}\sum_{1 \le j \le m}
\frac{1}{k(k-1)}
\sum_{\substack{i,i'\in\lambda_j\\ i\ne i'}}
X_iX_{i'}' \to E_Q[\tilde X\tilde X']~.
\end{align*}
Because $m^{-1}\sum_{1 \le j \le m}\bar X_{j,n}=\mu_{X,n}\to E_Q[\tilde X]$ by Assumption \ref{ass:X}(b), it follows that $\Sigma_{X,n}\to\var_Q[\tilde X]$. Next, we show that the diagonal entries of the projection matrix $H_{Q_2}$ converge to $0$. Let $\Sigma_{X, n}= \frac{1}{m} \sum_{1 \le j \le m}(\bar X_{j,n}-\mu_{X,n})(\bar X_{j,n}-\mu_{X,n})'$, and we have
\begingroup
\allowdisplaybreaks
\begin{align*}
H_{Q_2}&=Q_2(Q_2'Q_2)^{-1}Q_2' \\
&=
\begin{pmatrix}
     & (\bar X_{1, n} - \mu_{X, n})' \\
    \iota_m & \vdots \\
     & (\bar X_{m, n} - \mu_{X, n})'
\end{pmatrix}
\begin{pmatrix}
m & 0_p' \\
0_p & m \Sigma_{X, n}
\end{pmatrix}^{-1}
\begin{pmatrix}
 & \iota_m' & \\
\bar X_{1, n} - \mu_{X, n} & \cdots & \bar X_{m, n} - \mu_{X, n}
\end{pmatrix}
\\
&=
\frac{1}{m}\iota_m\iota_m'
+
\frac{1}{m} \begin{pmatrix}
(\bar X_{1, n} - \mu_{X, n})' \\
\vdots \\
(\bar X_{m, n} - \mu_{X, n})'
\end{pmatrix}
\Sigma_{X, n}^{-1}
\begin{pmatrix}
\bar X_{1, n} - \mu_{X, n} & \cdots & \bar X_{m, n} - \mu_{X, n}
\end{pmatrix}~.
\end{align*}
\endgroup
Therefore, denoting the operator norm by $\|\cdot\|_{\rm op}$, we have that for $1 \le j \le m$, the $j$-th diagonal entry of $H_{Q_2}$ may be uniformly bounded as
\begin{align*}
\max_{1 \leq j \leq m} [H_{Q_2}]_{jj}& = \max_{1 \leq j \leq m} \left ( \frac{1}{m}+\frac{1}{m}(\bar X_{j, n} - \mu_{X, n})' \Sigma_{X, n}^{-1}(\bar X_{j, n} - \mu_{X, n}) \right )
\\
&\le
\frac{1}{m} + \frac{1}{m}\left\lVert \Sigma_{X, n}^{-1} \right\rVert_{\rm op} \max_{1 \leq j \leq m} \lVert \bar X_{j, n} - \mu_{X, n} \rVert^2
\\
&\lesssim
\frac{1}{m}+ k \left\lVert \Sigma_{X, n}^{-1}\right\rVert_{\rm op} \frac{1}{n}\max_{1 \le i \le n} \lVert X_i \rVert^2 \to 0,
\end{align*}
where the first inequality follows from the definition of the operator norm, the second inequality follows from repeated application of the inequality that $(a - b)^2 \leq 2(a^2 + b^2)$, and the convergence follows from Assumption \ref{ass:X}(a) and that $\Sigma_{X, n} \to \var_Q[\tilde X]$.

Define
\[Q_{2, X} = \begin{pmatrix}
     (\bar X_{1, n} - \mu_{X, n})' \\
     \vdots \\
     (\bar X_{m, n} - \mu_{X, n})'
\end{pmatrix}~. \]
Let $H_{\iota_m}$ denote the projection matrix for $\iota_m$ and $H_{Q_{2, X}}$ denote the projection matrix for $Q_{2, X}$. Note that by construction $\iota_m' Q_{2, X} = 0$. It follows from Lemma \ref{lem:quadratic} that
\begin{align}
\nonumber n E[\hat{V}_n^{\rm F}] &= \frac{k}{m} E[\hat{\delta}_n'(\diag(I-H_{Q_2}))^{-1/2}(I-H_{Q_2})(\diag(I-H_{Q_2}))^{-1/2}\hat{\delta}_n]
\\
\label{eq:f1} & = \frac{k}{m} \delta_n' (\diag(I-H_{Q_2}))^{-1/2}(I-H_{Q_2})(\diag(I-H_{Q_2}))^{-1/2} \delta_n \\
\label{eq:f2} & \hspace{3em} + \frac{k}{m} \tr((I-H_{Q_2}) \var[(\diag(I-H_{Q_2}))^{-1/2} \hat \delta_n])~.
\end{align}
Define $\epsilon_n = (\diag(I-H_{Q_2}))^{-1/2} \delta_n - \delta_n$. Note \eqref{eq:f1} equals
\begin{equation} \label{eq:f1-expanded}
\frac{k}{m} \delta_n' (I - H_{Q_2}) \delta_n + \frac{k}{m} \epsilon_n' (I - H_{Q_2}) \delta_n + \frac{k}{m} \delta_n' (I - H_{Q_2}) \epsilon_n + \frac{k}{m} \epsilon_n' (I - H_{Q_2}) \epsilon_n~.    
\end{equation}
Note $\frac{1}{m} \|\delta_n\|^2$ converges to a constant by the derivations in the proof for the limit of $n E[\hat V_n^{\rm IM}]$. Because $\max_{1 \leq j \leq m} [H_{Q_2}]_{jj} \to 0$ as $n \to \infty$,
\[ \frac{1}{m} \|\epsilon_n\|^2 \leq \frac{1}{m} \|\delta_n\|^2 \max_{1 \leq j \leq m} \left ( \frac{1}{\sqrt{1 - [H_{Q_2}]_{jj}}} - 1 \right )^2 \to 0~. \]
Because $I - H_{Q_2}$ is a projection matrix, all of its eigenvalues are either 0 or 1, and hence its operator norm is 1. Therefore,
\[ \frac{k}{m} \epsilon_n' (I - H_{Q_2}) \delta_n \leq k \left ( \frac{1}{m}\|\epsilon_n\|^2 \right )^{1/2} \|I - H_{Q_2}\|_{\rm op} \left ( \frac{1}{m} \|\delta_n\|^2 \right )^{1/2} \to 0~. \]
Similarly, the last two terms in \eqref{eq:f1-expanded} also converge to 0 as $n \to \infty$. Recalling $\iota_m' Q_{2, X} = 0$, we get that \eqref{eq:f1} equals $o(1)$ plus
\begin{equation} \label{eq:muAmu}
\frac{k}{m} \delta_n' (I - H_{Q_2}) \delta_n = \frac{k}{m} \delta_n' (I-H_{\iota_m}-H_{Q_{2, X}}) \delta_n = \frac{k}{m} \delta_n' (I-H_{\iota_m}) \delta_n - \frac{k}{m} \delta_n' Q_{2, X} \Sigma_{X, n}^{-1} \frac{1}{m} Q_{2, X}' \delta_n~.
\end{equation}
Next, because $\hat \Delta_{j, n}$ are independent across $1 \leq j \leq m$, note \eqref{eq:f2} equals
\begin{align}
\nonumber & \frac{k}{m} \tr \left ( (I-H_{Q_2}) \diag \bigg ( \frac{\var[\hat \Delta_{j, n}]}{1 - [H_{Q_2}]_{jj}}: 1 \leq j \leq m \bigg ) \right ) \\
\nonumber & = \frac{k}{m} \sum_{1 \leq j \leq m} (1 - [H_{Q_2}]_{jj}) \frac{\var[\hat \Delta_{j, n}]}{1 - [H_{Q_2}]_{jj}} \\
\label{eq:Avar} & = \frac{k}{m} \sum_{1 \leq j \leq m} \var[\hat \Delta_{j, n}]~.
\end{align}
Moreover, this relationship holds with $Q_2$ replaced by $\iota_m$. Therefore, the \eqref{eq:muAmu}--\eqref{eq:Avar} imply
\begin{align*}
& \nonumber n E[\hat{V}_n^{\rm F}] \\
& = \frac{k}{m} \delta_n' (I-H_{\iota_m}) \delta_n - \frac{k}{m} \delta_n' Q_{2, X} \Sigma_{X, n}^{-1} \frac{1}{m} Q_{2, X}' \delta_n + \frac{k}{m} \tr((I - H_{\iota_m}) \var[(\diag(I - H_{\iota_m}))^{-1/2} \hat \delta_n]) + o(1) \\
& = n E[\hat V_n^{\rm IM}] - \frac{k}{m} \delta_n' Q_{2, X} \Sigma_{X, n}^{-1} \frac{1}{m} Q_{2, X}' \delta_n + o(1)
\end{align*}
For $d\in\{0,1\}$, it follows from Assumption \ref{ass:X}(b) and (c) that
\[ \frac{1}{m}\sum_{1 \le j \le m}\bar X_{j,n}\bar Y_{j,n}(d) = \frac{1}{k}\frac{1}n\sum_{1 \le i \le n}X_iY_i(d) + \frac{k-1}{k} \frac{1}{m}\sum_{1 \le j \le m} \frac{1}{k(k-1)} \sum_{\substack{i,i'\in\lambda_j\\ i\ne i'}} X_iY_{i'}(d) \to E_Q[\tilde X\tilde Y(d)]~. \]
Additionally applying Assumptions \ref{ass:super}(a) and \ref{ass:X}(b), we have
\[ \frac{1}{m} Q_{2, X}' \delta_n = \frac{1}{m} \sum_{1 \leq j \leq m} \bar X_{j, n} (\bar Y_{j, n}(1) - \bar Y_{j, n}(0)) - \mu_{X, n} \Delta_n \to \cov_Q[\tilde X,\tilde Y(1)-\tilde Y(0)]~. \]
Further note $\frac{m - 1}{m} \to 1$, $\Sigma_{X, n} \to \var_Q[\tilde X]$, and
\[ nE[\hat V_n^{\rm IM}] \to V^{\rm obs} + k \var_Q[E_Q[\tilde Y(1) - \tilde Y(0) | \tilde X]] \]
by Theorem \ref{thm:main_var} and Proposition \ref{prop:complete_graph_product}, so that
\begin{align*}
n E[\hat{V}_n^{\rm F}] &\to
V^{\rm obs}
+
k\var_Q[ E_Q[ \tilde{Y}(1)-\tilde{Y}(0)|\tilde{X} ] ]
-
k\cov_Q[\tilde X,\tilde Y(1)-\tilde Y(0)]'\var_Q[\tilde X]^{-1}\cov_Q[\tilde X,\tilde Y(1)-\tilde Y(0)]
\\
&=
V^{\rm obs}
+
k\var_Q[ E_Q[ \tilde{Y}(1)-\tilde{Y}(0)|\tilde{X} ] ]
-
k\var_Q[\text{BLP}_Q(\tilde Y(1)-\tilde Y(0)|1,\tilde X)]
\\
&=
V^{\rm obs}
+
k\var_Q[E_Q[ \tilde{Y}(1)-\tilde{Y}(0)|\tilde{X} ]-\text{BLP}_Q(\tilde Y(1)-\tilde Y(0)|1,\tilde X)]~,
\end{align*}
where the first equality follows from Lemma \ref{lem:blp}(b) and the second follows from Lemma \ref{lem:blp}(d). 
\end{proof}

\section{Auxiliary Lemmas}

Here we state and prove a standard central limit theorem for design-based inference \citep[see for instance][]{hajek1960limiting,li2017general}.

\begin{theorem}\label{thm:normal}
Suppose that Assumption \ref{ass:dn} holds and that
\begin{equation} \label{eq:highlevel}
\frac{\max_{1 \leq j \leq m} \max_{i \in \lambda_j} (R_i - \bar R_j)^2}{\sum_{1 \leq j \leq m} \sum_{i \in \lambda_j} (R_i - \bar R_j)^2} \rightarrow 0
\end{equation}
for $R_i$ and $\bar{R}_j$ in \eqref{eq:Ri} and \eqref{eq:Rbar} as $n \rightarrow \infty$. Then,
\[ \frac{\hat \Delta_n - \Delta_n}{\var[\hat \Delta_n]^{1/2}} \stackrel{d}{\rightarrow} N(0,1)~. \]
\end{theorem}

\begin{proof}
Re-writing $\hat{\Delta}_n$ and $\Delta_n$ over the blocks $\{\lambda_j: 1 \le j \le m\}$, we obtain
\[\hat{\Delta}_n = \frac{1}{n}\left(\sum_{1 \leq j \leq m}\sum_{i \in \lambda_j}\left(\frac{Y_i(1)}{\eta} + \frac{Y_i(0)}{1 - \eta}\right)D_i - \sum_{1 \leq j \leq m}\sum_{i \in \lambda_j}\frac{Y_i(0)}{1 - \eta}\right)~,\]
and 
\[\Delta_n = \frac{1}{n}\left(\sum_{1 \leq j \leq m}\sum_{i \in \lambda_j}\left(\bar{R}_j(1) + \bar{R}_j(0)\right)D_i - \sum_{1 \leq j \leq m}\sum_{i \in \lambda_j}\frac{Y_i(0)}{1 - \eta}\right)~,\]
where $\bar{R}_j(1) = \frac{1}{k}\sum_{i \in \lambda_j}\frac{Y_i(1)}{\eta}$ and $\bar{R}_j(0) = \frac{1}{k}\sum_{i \in \lambda_j}\frac{Y_i(0)}{1-\eta}$. Putting both together, using the definitions in \eqref{eq:Ri} and \eqref{eq:Rbar}, and noting $\bar{R}_j = \bar{R}_j(1) + \bar{R}_j(0)$, we have
\[\hat{\Delta}_n - \Delta_n = \frac{1}{n}\sum_{1 \leq j \leq m}\sum_{i \in \lambda_j}(R_i - \bar{R}_j)D_i~.\]
Hence we obtain that 
\[\sqrt{m}(\hat{\Delta}_n - \Delta_n) = \sum_{1 \leq j \leq m}A_j~,\]
where $A_j = \frac{1}{\sqrt{m}}\frac{\eta}{\ell}\sum_{i \in \lambda_j}(R_i - \bar{R}_j)D_i$. Note that $\frac{1}{\ell}\sum_{i \in \lambda_j}(R_i - \bar{R}_j)D_i$ is the sample mean when sampling a subset of size $\ell$ without replacement from the finite population $\{(R_i - \bar{R}_j): i \in \lambda_j\}$. Accordingly, from the properties of the sample mean when sampling from a finite population \citep[Theorems 2.1 and 2.2 in][]{cochran1977sampling}, we obtain immediately that 
\begin{align*}
E[A_j] & = 0 \\
\var[A_j] & = \frac{1}{m}\eta^2\left(\frac{1}{\ell} - \frac{1}{k}\right)\left(\frac{1}{k-1}\right)\sum_{i \in \lambda_j}(R_i - \bar{R}_j)^2~.
\end{align*}
Moreover, by the definition of $A_j$ and Jensen's inequality, 
\[A_j^2 \lesssim \frac{1}{m}\max_{i \in \lambda_j}(R_i - \bar{R}_j)^2~.\]
To establish asymptotic normality, we verify the Lindeberg condition:
\begin{equation} \label{eq:lindeberg}
\sum_{1 \leq j \leq m} E \bigg [ \frac{A_j^2}{s_n^2} I \bigg\{\frac{A_j^2}{s_n^2} > \epsilon^2\bigg\} \bigg ] \to 0~,    
\end{equation}
where $s_n^2 = \sum_{1 \leq j \leq m}\var[A_j]$. Note that by our previous calculations 
\[ s_n^2 \propto \frac{1}{m}\sum_{1 \leq j \leq m} \sum_{i \in \lambda_j} (R_i - \bar R_j)^2~,\]
and hence 
\[\frac{A_j^2}{s_n^2} \lesssim \frac{\max_{i \in \lambda_j}(R_i - \bar{R}_j)^2}{\sum_{1 \leq j \leq m}\sum_{i \in \lambda_j}(R_i - \bar{R}_j)^2}~.\]
Therefore, if \eqref{eq:highlevel} holds, then for each fixed $\epsilon > 0$, for $m$ large enough,
\[ \max_{1 \leq j \leq m} \frac{A_j^2}{s_n^2} \leq \epsilon^2~, \]
so
\[ \sum_{1 \leq j \leq m} \frac{A_j^2}{s_n^2} I \bigg\{\frac{A_j^2}{s_n^2} > \epsilon^2\bigg\} \leq \sum_{1 \leq j \leq m} \frac{A_j^2}{s_n^2} I \bigg\{\max_{1 \leq j \leq m} \frac{A_j^2}{s_n^2} > \epsilon^2\bigg\} = 0~. \]
As a result, the Lindeberg condition in \eqref{eq:lindeberg} holds, and the conclusion follows.
\end{proof}

\begin{lemma} \label{lem:ratio}
Let $X_n, n \geq 1$ be a sequence of random variables such that $E[X_n] = 0$, $\var[X_n] = \sigma_n^2$ and $X_n / \sigma_n \stackrel{d}{\to} N(0, 1)$. Further suppose $\liminf_{n \to \infty} \sigma_n^2 > 0$, $\limsup_{n \to \infty} \sigma_n^2 < \infty$, and there exists a sequence of random variables $\hat \sigma_n^2 \geq 0, n \geq 1$ such that $\hat \sigma_n^2 - E[\hat \sigma_n^2] \stackrel{P}{\to} 0$. Then,
\begin{align*}
\liminf_{n \to \infty} P \left \{ \frac{X_n}{\hat \sigma_n} \leq x \right \} \geq  \Phi(x) & \text{ for } x > 0 \\
\limsup_{n \to \infty} P \left \{ \frac{X_n}{\hat \sigma_n} \leq x \right \} \leq  \Phi(x) & \text{ for } x < 0~,
\end{align*}
where $\Phi(x) = P \{N(0, 1) \leq x\}$, if and only if $\liminf_{n \to \infty} (E[\hat \sigma_n^2] - \sigma_n^2) \geq 0$.
\end{lemma}

\begin{proof}
We first prove the assertion for $x \geq 0$. The assertion for $x<0$ is obtained by applying the $x>0$ result to $-X_n$. Because $\liminf_{n \to \infty} \sigma_n^2 > 0$, $\sigma_n^2 \ge c > 0$ for some $c > 0$, at least for $n$ large enough. Since further $\liminf_{n \to \infty} (E[\hat \sigma_n^2] - \sigma_n^2) \geq 0$, $E[\hat \sigma_n^2] \geq \frac{3}{4} c$ for $n$ large enough. Because further $\hat \sigma_n^2 - E[\hat \sigma_n^2] \stackrel{P}{\to} 0$, we have
\[ P \{\hat \sigma_n^2 > c / 2\} \to 1~. \]
Therefore, in what follows, all statements should be understood as conditioning on this event. In that case,
\begin{equation} \label{eq:x/v}
\frac{X_n}{\hat \sigma_n} = \frac{X_n}{\sigma_n} \frac{\sigma_n}{\hat \sigma_n}~.    
\end{equation}
Next, we claim that for each $\epsilon > 0$,
\begin{equation} \label{eq:v/v}
P \left \{ \frac{\sigma_n}{\hat \sigma_n} > 1 + \epsilon \right \} \to 0~.    
\end{equation}
Suppose for a moment that \eqref{eq:v/v} holds. Then, because $X_n / \sigma_n \stackrel{d}{\to} N(0, 1)$, $\sigma_n > 0$, and $\hat \sigma_n > 0$, \eqref{eq:x/v} implies that for each $x \geq 0$ and $\epsilon > 0$,
\[ P \left \{ \frac{X_n}{\hat \sigma_n} > x (1 + \epsilon) \right \} \leq P \left \{ \frac{X_n}{\sigma_n} > x \right \} + P \left \{ \frac{\sigma_n}{\hat \sigma_n} > 1 + \epsilon \right \} \to 1 - \Phi(x) \]
as $n \to \infty$. In other words, for $x > 0$ and $\epsilon > 0$,
\[ \limsup_{n \to \infty} P \left \{ \frac{X_n}{\hat \sigma_n} > x \right \} \leq 1 - \Phi(x / (1 + \epsilon))~. \]
The conclusion of the theorem follows by letting $\epsilon \to 0$ and noting $\Phi$ is continuous everywhere. It therefore suffices to prove \eqref{eq:v/v}. Because \eqref{eq:v/v} holds if and only if
\begin{equation} \label{eq:indicatorP}
I \left \{ \frac{\sigma_n}{\hat \sigma_n} > 1 + \epsilon \right \} \stackrel{P}{\to} 0~,    
\end{equation}
it suffices to prove that for each subsequence $\{n(k)\}$, there exists a further subsequence $\{n(k(\ell))\}$ along which the convergence in \eqref{eq:indicatorP} holds with probability one. Fix an arbitrary subsequence $\{n(k)\}$. Because $\liminf_{n \to \infty} \sigma_n^2 > 0$, $\liminf_{n \to \infty} (E[\hat \sigma_n^2] - \sigma_n^2) \geq 0$, and $\hat \sigma_n^2 - E[\hat \sigma_n^2] \stackrel{P}{\to} 0$, there exists a constant $c > 0$ and a further subsequence $\{n(k(\ell))\}$ along which $\sigma_{n(k(\ell))}^2 \geq c > 0$, $E[\hat \sigma_{n(k(\ell))}^2] \geq c / 2 > 0$, and $\hat \sigma_{n(k(\ell))}^2 - E[\hat \sigma_{n(k(\ell))}^2] \to 0$ with probability one. Furthermore,
\[ \limsup_{\ell \to \infty} \frac{\sigma_{n(k(\ell))}^2}{E[\hat \sigma_{n(k(\ell))}^2]} \leq 1~. \]
Along this subsequence,
\[ \limsup_{\ell \to \infty} \frac{\sigma_{n(k(\ell))}^2}{\hat \sigma_{n(k(\ell))}^2} = \limsup_{\ell \to \infty} \frac{\sigma_{n(k(\ell))}^2}{E[\hat \sigma_{n(k(\ell))}^2]} \frac{E[\hat \sigma_{n(k(\ell))}^2]}{\hat \sigma_{n(k(\ell))}^2} \leq 1~, \]
so \eqref{eq:indicatorP} holds with probability one, and hence \eqref{eq:v/v} follows.

Next, we show that $\liminf_{n \rightarrow \infty} (E[\hat{\sigma}_n^2] - \sigma_n^2) \ge 0$ is necessary for the result to hold. We present the proof for $x > 0$; the case $x<0$ is obtained by applying the displayed inequalities to $-X_n$. To that end, first suppose $\liminf_{n \rightarrow \infty} (E[\hat{\sigma}_n^2] - \sigma_n^2) = - c_1 < 0$ and $\liminf_{n \to \infty} E[\hat \sigma_n^2] = c_2 > 0$. Define $c = c_1 \wedge c_2$. Then, there exists a subsequence $\{n(k)\}$ along which $E[\hat \sigma_{n(k)}^2] < \sigma_{n(k)}^2 - c / 2$ and $E[\hat \sigma_{n(k)}^2] \ge c/2$. Because by assumption $\limsup_{n \to \infty} \sigma_n^2 < \infty$, there exists $M > 0$ such that $\sigma_{n(k)}^2 \leq M$ for all $k$. Define $\epsilon = \left(\frac{M}{2M - c }\right)^{1/2} - 1 > 0$. Along this subsequence,
\[
\frac{E[\hat \sigma_{n(k)}^2]}{\sigma_{n(k)}^2}
<
1-\frac{c}{2\sigma_{n(k)}^2}
\le
1-\frac{c}{2M}
=
\frac{1}{2(1+\epsilon)^2} < \frac{1}{(1+\epsilon)^2} ~.
\]
Since $\hat \sigma_{n(k)}^2-E[\hat \sigma_{n(k)}^2]\xrightarrow{P}0$ and $E[\hat \sigma_{n(k)}^2]\ge c/2$, it follows that
\begin{equation} \label{eq:indicator2}
P \left \{ \frac{\sigma_{n(k)}}{\hat \sigma_{n(k)}} \geq 1 + \epsilon \right \} \to 1~.
\end{equation}
If instead $\liminf_{n \to \infty} E[\hat \sigma_n^2] = 0$, then by passing to a subsequence along which $E[\hat \sigma_{n(k)}^2] \to 0$, we can construct a further subsequence along which
\[\frac{\sigma_{n(k(\ell))}}{\hat \sigma_{n(k(\ell))}}\rightarrow \infty~,\]
with probability one, and thus \eqref{eq:indicator2} still holds. Let
\[
A_k=\left\{\frac{X_{n(k)}}{\sigma_{n(k)}}>\frac{x}{1+\epsilon}\right\}~,
\qquad
B_k=\left\{\frac{\sigma_{n(k)}}{\hat\sigma_{n(k)}}\ge 1+\epsilon\right\}~.
\]
Then
\begin{align*}
P \left \{ \frac{X_{n(k)}}{\hat \sigma_{n(k)}} > x \right \}
&\ge P(A_k\cap B_k)\\
&=P(A_k)-P(A_k\cap B_k^c)\\
&\to 1-\Phi(x/(1+\epsilon))>1-\Phi(x)~.
\end{align*}
Therefore,
\[ \liminf_{n \to \infty} P \left \{ \frac{X_n}{\hat \sigma_n} \leq x \right \} = 1 - \limsup _{n \to \infty} P \left \{ \frac{X_n}{\hat \sigma_n} > x \right \} \leq \Phi(x / (1 + \epsilon)) < \Phi(x)~, \]
and the desired result follows.
\end{proof}

\begin{lemma} \label{lem:quadratic}
Suppose $E[\hat \delta_n] = \delta_n$ and $A$ is deterministic. Then, $E[\hat \delta_n' A \hat \delta_n] = \delta_n' A \delta_n + \tr(A \var[\hat \delta_n])$.   
\end{lemma}

\begin{proof}
Note
\begin{align*}
E[\hat \delta_n' A \hat \delta_n] & = E[(\delta_n + \hat \delta_n - \delta_n)' A ( \delta_n + \hat \delta_n - \delta_n)] \\
& = \delta_n' A \delta_n + E[(\hat \delta_n - \delta_n)' A (\hat \delta_n - \delta_n)] \\
& = \delta_n' A \delta_n + \tr( E[(\hat \delta_n - \delta_n)' A (\hat \delta_n - \delta_n)]) \\
& = \delta_n' A \delta_n + E[\tr((\hat \delta_n - \delta_n)' A (\hat \delta_n - \delta_n))] \\
& = \delta_n' A \delta_n + E[\tr(A (\hat \delta_n - \delta_n) (\hat \delta_n - \delta_n)')] \\
& = \delta_n' A \delta_n + \tr(E[A (\hat \delta_n - \delta_n) (\hat \delta_n - \delta_n)']) \\
& = \delta_n' A \delta_n + \tr(A \var[\hat \delta_n])~,
\end{align*}
where the second equality uses $E[\hat \delta_n - \delta_n] = 0$, the third uses that $E[(\hat \delta_n - \delta_n)' A (\hat \delta_n - \delta_n)]$ is a scalar, the fourth and sixth use linearity of expectation, and the fifth uses $\tr(AB) = \tr(BA)$ for conformable matrices.
\end{proof}

\begin{lemma}\label{lem:mean_cost_to_graph_locality}
For any $\omega_m\in\Omega_m$, $C_n(\omega_m)$ in \eqref{eq:ordered_mean_cost}, and $\mathcal D_n(\omega_m)$ in \eqref{eq:graph_edge_diameter},
\[
\mathcal D_n(\omega_m)
\le
12(2R_n+C_n(\omega_m))~.
\]
Therefore, if $R_n\to0$ and $C_n(\omega_m)\to0$, then $\mathcal D_n(\omega_m)\to0$.
\end{lemma}

\begin{proof}
Define $r_{j,n}^2 = \max_{i\in\lambda_j}\|X_i-\bar X_{j,n}\|^2$ for $1 \leq j \leq m$. Fix two strata $1 \leq a \neq b \leq m$. For any $i,i'\in\lambda_a\cup\lambda_b$, the triangle inequality implies
\[
\|X_i-X_{i'}\|
\le
2r_{a,n}+\|\bar X_{a,n}-\bar X_{b,n}\|+2r_{b,n}~.
\]
Hence, using $(x+y+z)^2\le3(x^2+y^2+z^2)$ for nonnegative $x,y,z$, we get
\[
\max_{i,i'\in\lambda_a\cup\lambda_b}\|X_i-X_{i'}\|^2
\le
12(r_{a,n}^2+\|\bar X_{a,n}-\bar X_{b,n}\|^2+r_{b,n}^2)~.
\]
By definition of $\mathcal D_n(\omega_m)$ and $C_n(\omega_m)$,
\begin{align*}
\mathcal D_n(\omega_m)
&\le
\frac{12}{m}\sum_{1 \le a \le m}\sum_{b\ne a}\omega_{ab,m}
(r_{a,n}^2+\|\bar X_{a,n}-\bar X_{b,n}\|^2+r_{b,n}^2)  \\
&=
12\bigg(
\frac{1}{m}\sum_{1 \le a \le m} r_{a,n}^2\sum_{b\ne a}\omega_{ab,m}
+C_n(\omega_m)
+\frac{1}{m}\sum_{1 \le b \le m} r_{b,n}^2\sum_{a\ne b}\omega_{ab,m}
\bigg)~.
\end{align*}
Both sums of the weights are one by degree calibration. The conclusion now follows.
\end{proof}

\begin{lemma}\label{lem:spatial_sorting_consequence}
For $A_n(\cdot)$ in \eqref{eq:A_n}, there is a constant $K_{p_X}<\infty$ depending only on $p_X$ such that
\begin{enumerate}
\item[(a)] If $m$ is even, there exists a partition $\tilde{\mathcal P}_m$ of $\{1, \dots, m\}$ into pairs of strata such that
\[
A_n(\tilde{\mathcal P}_m)
\le
K_{p_X}M_{n}^2m^{-2/(p_X+1)}~.
\]
\item[(b)] For any sequence $b_m\to\infty$ such that $b_m=o(m)$, there exists a partition $\mathcal B_m$ of $\{1, \dots, m\}$ with
\[
b_m\le |B|\le 2b_m
\text{ for all }B\in\mathcal B_m
\]
for $m$ large enough, such that
\[
A_n(\mathcal B_m)
\le
K_{p_X}M_{n}^2\left(\frac{m}{b_m}\right)^{-2/(p_X+1)}~.
\]
\end{enumerate}
\end{lemma}

\begin{proof}
In what follows, we will show parts (a) and (b) hold under two different constants and define $K_{p_X}$ to be the maximum of these two constants.

Part (a) follows from Theorem A.3 of \cite{cytrynbaum2021optimal} applied in his notation with sample size $m$, $L_n = 1$, and group size $\bar{k}_n = 2$.

To show part (b), first apply Theorem A.3 of \cite{cytrynbaum2021optimal} in his notation with sample size $m$, $L_n = 1$, and group size $\bar{k}_n = b_m$. We then have a partition $\mathcal B_m^0$ which consists of blocks of size $b_m$, together with a possible remainder block $B_m^{\rm rem}$ satisfying $|B_m^{\rm rem}|<b_m$, and
\[
A_n(\mathcal B_m^0)
\le
C_{p_X}M_{n}^2
\left(\frac{m}{b_m}\right)^{-2/(p_X+1)}~.
\]
If $|B_m^{\rm rem}|=0$, set $\mathcal B_m=\mathcal B_m^0$. It remains to study the case that $|B_m^{\rm rem}|>0$. In that case, choose any non-remainder block $B_m'\in\mathcal B_m^0$. Such a block must exist for $m$ large enough because $b_m=o(m)$. Define
\[
\mathcal B_m
:=
\bigl(\mathcal B_m^0\setminus\{B_m',B_m^{\rm rem}\}\bigr)
\cup
\{B_m'\cup B_m^{\rm rem}\}~.
\]
In other words, we merge $B_m'$ and $B_m^{\rm rem}$ in the partition. Note that
\[
b_m\le |B|\le 2b_m \text{ for every }B\in\mathcal B_m~.
\]
We now bound $A_n(\mathcal B_m) - A_n(\mathcal B_m^0)$. For any nonempty $S\subseteq\{1,\ldots,m\}$, define
\[
\bar X_{S,n}:=\frac{1}{|S|}\sum_{a\in S}\bar X_{a,n}~,
\qquad
W(S):=
\sum_{a\in S}
\bigg\|
\bar X_{a,n}-\bar X_{S,n}
\bigg\|^2~.
\]
Then,
\[
A_n(\mathcal B)=\frac{1}{m}\sum_{S\in\mathcal B}W(S)
\]
for any partition $\mathcal B$. We claim that
\begin{equation} \label{eq:merge-diff}
W(B_m'\cup B_m^{\rm rem})-W(B_m')-W(B_m^{\rm rem})
=
\frac{|B_m'|\,|B_m^{\rm rem}|}{|B_m'|+|B_m^{\rm rem}|}
\bigg\|
\bar X_{B_m',n}-\bar X_{B_m^{\rm rem},n}
\bigg\|^2~.    
\end{equation}
To show \eqref{eq:merge-diff}, note
\[
\bar X_{B_m'\cup B_m^{\rm rem},n}
=
\frac{|B_m'|\bar X_{B_m',n}+|B_m^{\rm rem}|\bar X_{B_m^{\rm rem},n}}{|B_m'|+|B_m^{\rm rem}|}~.
\]
As a result,
\begin{align*}
W(B_m'\cup B_m^{\rm rem})
&=
\sum_{a\in B_m'}\|\bar X_{a,n}-\bar X_{B_m'\cup B_m^{\rm rem},n}\|^2
+
\sum_{a\in B_m^{\rm rem}}\|\bar X_{a,n}-\bar X_{B_m'\cup B_m^{\rm rem},n}\|^2\\
&=
\sum_{a\in B_m'}\|\bar X_{a,n}-\bar X_{B_m',n}\|^2
+
|B_m'|\|\bar X_{B_m',n}-\bar X_{B_m'\cup B_m^{\rm rem},n}\|^2\\
&\hspace{2em}+
\sum_{a\in B_m^{\rm rem}}\|\bar X_{a,n}-\bar X_{B_m^{\rm rem},n}\|^2
+
|B_m^{\rm rem}|\|\bar X_{B_m^{\rm rem},n}-\bar X_{B_m'\cup B_m^{\rm rem},n}\|^2~,
\end{align*}
where the cross terms vanish because $\sum_{a\in B_m'}(\bar X_{a,n}-\bar X_{B_m',n})=0$ and $\sum_{a\in B_m^{\rm rem}}(\bar X_{a,n}-\bar X_{B_m^{\rm rem},n})=0$. Therefore,
\begin{align*}
&W(B_m'\cup B_m^{\rm rem})-W(B_m')-W(B_m^{\rm rem})\\
&\hspace{2em}=
|B_m'|\|\bar X_{B_m',n}-\bar X_{B_m'\cup B_m^{\rm rem},n}\|^2
+
|B_m^{\rm rem}|\|\bar X_{B_m^{\rm rem},n}-\bar X_{B_m'\cup B_m^{\rm rem},n}\|^2\\
&\hspace{2em}=
\left[
|B_m'|\left(\frac{|B_m^{\rm rem}|}{|B_m'|+|B_m^{\rm rem}|}\right)^2
+
|B_m^{\rm rem}|\left(\frac{|B_m'|}{|B_m'|+|B_m^{\rm rem}|}\right)^2
\right]
\|\bar X_{B_m',n}-\bar X_{B_m^{\rm rem},n}\|^2\\
&\hspace{2em}=
\frac{|B_m'||B_m^{\rm rem}|}{|B_m'|+|B_m^{\rm rem}|}\|\bar X_{B_m',n}-\bar X_{B_m^{\rm rem},n}\|^2~,
\end{align*}
so \eqref{eq:merge-diff} follows. Next, since each coordinate of $\bar X_{B_m',n}$ and $\bar X_{B_m^{\rm rem},n}$ lies between the corresponding coordinate-wise minimum and maximum of $\{\bar X_{a,n}:1 \le a \le m\}$,
\[
\|
\bar X_{B_m',n}-\bar X_{B_m^{\rm rem},n}
\|^2
\le
p_X M_{n}^2~.
\]
Therefore,
\[
W(B_m'\cup B_m^{\rm rem})-W(B_m')-W(B_m^{\rm rem})
\le
p_X M_{n}^2 |B_m^{\rm rem}|
\le
p_X M_{n}^2 b_m~.
\]
Dividing by $m$, we obtain
\begin{equation} \label{eq:merge_1}
A_n(\mathcal B_m)
\le
A_n(\mathcal B_m^0)
+
p_X M_{n}^2\frac{b_m}{m}~.
\end{equation}
Finally, because $p_X\ge 1$ and $b_m/m<1$ for $m$ large enough,
\begin{equation} \label{eq:merge_2}
\frac{b_m}{m}
\le
\left(\frac{b_m}{m}\right)^{2/(p_X+1)}
=
\left(\frac{m}{b_m}\right)^{-2/(p_X+1)}~.
\end{equation}
Combining \eqref{eq:merge_1} and \eqref{eq:merge_2}, we have
\[
A_n(\mathcal B_m)
\le
(C_{p_X}+p_X)M_{n}^2
\left(\frac{m}{b_m}\right)^{-2/(p_X+1)}~.
\]
The desired conclusion in part (b) now follows.
\end{proof}

\begin{lemma} \label{lem:blp}
Let $(Y, X)$ be random vectors where $Y$ lives in $\mathbb R$ and $X$ lives in $\mathbb R^k$. Further suppose that $E[Y^2] < \infty$, $E[X X'] < \infty$, and $\var[X]$ is nonsingular. Then,
\begin{enumerate}[\rm (a)]
    \item $E[\blp(Y | 1, X)] = E[Y]$.
    % \item $\blp(Y | 1, X) = \blp(Y | 1, X - E[X])$.
    \item $(X - E[X])' \var[X]^{-1} \cov[X, Y] = \blp(Y | 1, X) - E[\blp(Y | 1, X)]$.
    \item $E[(Y - E[Y] - (\blp(Y | 1, X) - E[\blp(Y | 1, X)])) (\blp(Y | 1, X) - E[\blp(Y | 1, X)])] = 0$ and therefore $\cov[Y, \blp(Y|1, X)] = \var[\blp(Y|1, X)]$.
    \item $E([E[Y | X] - E[Y] - (\blp(Y | 1, X) - E[\blp(Y | 1, X)])) (\blp(Y | 1, X) - E[\blp(Y | 1, X)])] = 0$ and therefore $\cov[E[Y | X], \blp(Y| 1, X)] = \var[\blp(Y|1, X)]$ and $\var[E[Y | X]] = \var[E[Y | X] - \blp(Y|1, X)] + \var[\blp(Y|1, X)]$.
\end{enumerate}
\end{lemma}

\begin{proof}
(a) follows because $E[(Y - \blp(Y | 1, X)) \cdot 1] = 0$. (b) follows because if $\blp(Y | 1, X) = \beta_0 + X' \beta_1$ then $\beta_1 = \var[X]^{-1} \cov[X, Y]$. (c) follows because recalling $E[\blp(Y | 1, X)] = E[Y]$, we have
\begin{align*}
& E[(Y - E[Y] - (\blp(Y | 1, X) - E[\blp(Y | 1, X)])) (\blp(Y | 1, X) - E[\blp(Y | 1, X)])] \\
& = E[(Y - \blp(Y | 1, X)) (\blp(Y | 1, X) - E[Y])] \\
& = E[(Y - \blp(Y | 1, X)) \blp(Y | 1, X)] + E[Y - \blp(Y | 1, X)] E[Y] \\
& = 0~.
\end{align*}
(d) follows similarly.
\end{proof}

\newpage
\bibliography{finpop}
\end{document}